\documentclass[letterpaper, 10 pt, conference, onecolumn]{ieeeconf}  % Comment this line out if you need a4paper

\IEEEoverridecommandlockouts                              % This command is only needed if 
\usepackage{xcolor}
\usepackage{cite}
\usepackage{amsmath,amssymb,amsfonts}
\usepackage{algorithmic}
\usepackage{algorithm}
\usepackage{graphicx}
\usepackage{subfig}
\usepackage{textcomp}

\graphicspath{{pics/}}

\usepackage{mathtools}

\usepackage{amsthm}

\usepackage[acronym, style=alttree, toc=true, shortcuts, xindy, nomain, nonumberlist]{glossaries}

\glsdisablehyper % disable hyperlink for glossary entries.

\usepackage{cleveref}

\usepackage{hyphenat}
\usepackage{bm}

\usepackage{siunitx}  

\newcommand{\prob}[2][]{\var{\mathbb{P}}{#1}\left(#2\right)}  % Probability
\newcommand{\norm}[1]{\left\lVert#1\right\rVert} % general norm
\newcommand{\vertlr}[1]{\left\lvert#1\right\rvert} % general norm
\newcommand{\set}[1]{\mathcal{#1}} % mathematical set
\newcommand{\setdef}[2][]{
	\left\{
		\ifblank{#1}{}{#1 \hspace{.1cm} \middle| \hspace{.1cm}}
		#2
	\right\}
} % set with definition
\newcommand{\supp}[0]{\hspace{.01cm}\text{supp}\hspace{.01cm}} % support
\newcommand{\argmin}[1]{\underset{#1}{\hspace{.1cm}\mathrm{arg}\hspace{.05cm}\mathrm{min}} \hspace{.1cm}} %argmin
\newcommand{\argmax}[1]{\underset{#1}{\mathrm{arg}\hspace{.05cm}\mathrm{max}} \hspace{.1cm}} %argmax

\newcommand{\rv}[1]{\text{#1}} %normaltext
\newcommand{\var}[2]{{#1}_{\text{#2}}} %variable, normaltext name
\newcommand{\vark}[2]{{#1}_{\text{#2},k}} %variable, normaltext name + k
\newcommand{\varj}[2]{{#1}_{\text{#2},j}} %variable, normaltext name + j
\newcommand{\vari}[2]{{#1}_{\text{#2},i}} %variable, normaltext name + j
\newcommand{\varz}[2]{{#1}_{\text{#2},0}} %variable, normaltext name + 0
\newcommand{\varkk}[2]{{#1}_{\text{#2},1:k}} %variable, normaltext name + k
\newcommand{\f}[2]{#1\left(#2\right)} % f(.)
\newcommand{\lr}[1]{\left(#1\right)} % parenthesis with \left, \right

\newcommand{\T}[0]{T}

\DeclareMathOperator\erf{erf} % error function

\DeclarePairedDelimiterXPP\onenorm[1]{}\lVert\rVert{_1}{\ifblank{#1}{\:\cdot\:}{#1}} % onenorm
\DeclarePairedDelimiterXPP\twonorm[1]{}\lVert\rVert{_2}{\ifblank{#1}{\:\cdot\:}{#1}} % twonorm
\newtheorem{theorem}{Theorem}
\newtheorem{assumption}{Assumption}
\newtheorem{remark}{Remark}
\newtheorem{definition}{Definition}
\newtheorem{lemma}{Lemma}
\newtheorem{corollary}{Corollary}
\newtheorem{proposition}{Proposition}

\newglossary{symbols}{sym}{sbl}{List of Symbols}
\newglossary{notation}{not}{nt}{Notation}

\glsaddkey{dimension}{\glsentrytext{\glslabel}}{\glsentryd}{\GLsentryd}{\glsd}{\Glsd}{\GLSd}
\glsaddkey{body}{\glsentrytext{\glslabel}}{\glsentrybody}{\GLsentrybody}{\glsbody}{\Glsbody}{\GLSbody}

\newacronym{MPC}{MPC}{Model Predictive Control}
\newacronym{MPCa}{MPC algorithm}{MPC algorithm}
\newacronym{RMPC}{RMPC}{Robust Model Predictive Control}
\newacronym{SMPC}{SMPC}{Stochastic Model Predictive Control}
\newacronym{SCMPC}{SCMPC}{Scenario Model Predictive Control}
\newacronym{MILP}{MILP}{Mixed Integer Linear Program}
\newacronym{PIT}{PIT}{Pointwise\hyp{}In\hyp{}Time}
\newacronym{POMDP}{POMDP}{Partially Observable Markov Decision Process}
\newacronym{MDP}{MDP}{Markov Decision Process}
\newacronym{KKT}{KKT}{Karush\hyp{}Kuhn\hyp{}Tucker}
\newacronym{pdf}{PDF}{probability density function}

\newacronym{CVPM}{CVPM}{constraint violation probability minimization}
\newacronym{nSMPC}{CVPM-SMPC}{constraint violation probability minimization in \gls{SMPC}}

\newacronym{cv}{CV}{controlled vehicle}
\newacronym{ob}{OB}{obstacle}

\newglossaryentry{simux}{type=symbols,
	sort={xk},
	name={\ensuremath{x}},
	dimension={n},
	body={\ensuremath{\mathbb{R}}},
	description={State vector to dynamic linear system}
}
\newglossaryentry{simuy}{type=symbols,
	sort={xk},
	name={\ensuremath{y}},
	dimension={n},
	body={\ensuremath{\mathbb{R}}},
	description={State vector to dynamic linear system}
}
\newglossaryentry{simuxr}{type=symbols,
	sort={xk},
	name={\ensuremath{x_\text{r}}},
	dimension={n},
	body={\ensuremath{\mathbb{R}}},
	description={State vector to dynamic linear system}
}
\newglossaryentry{simuyr}{type=symbols,
	sort={xk},
	name={\ensuremath{y_\text{r}}},
	dimension={n},
	body={\ensuremath{\mathbb{R}}},
	description={State vector to dynamic linear system}
}

\newglossaryentry{xk}{type=symbols,
	sort={xk},
	name={\ensuremath{\bm{x}_k}},
	dimension={n},
	body={\ensuremath{\mathbb{R}}},
	description={State vector to dynamic linear system}
}

\newglossaryentry{y10wp}{type=symbols,
	sort={xk},
	name={\ensuremath{\gls{Mm1}+\gls{wimax0}}},
	dimension={n},
	body={\ensuremath{\mathbb{R}}},
	description={State vector to dynamic linear system}
}
\newglossaryentry{y10wm}{type=symbols,
	sort={xk},
	name={\ensuremath{\gls{Mm1}-\gls{wimax0}}},
	dimension={n},
	body={\ensuremath{\mathbb{R}}},
	description={State vector to dynamic linear system}
}

\newglossaryentry{hy10wp}{type=symbols,
	sort={xk},
	name={\ensuremath{\gls{hi}\lr{\gls{Mm1}+\gls{wimax0}}}},
	dimension={n},
	body={\ensuremath{\mathbb{R}}},
	description={State vector to dynamic linear system}
}
\newglossaryentry{hy10wm}{type=symbols,
	sort={xk},
	name={\ensuremath{\gls{hi}\lr{\gls{Mm1}-\gls{wimax0}}}},
	dimension={n},
	body={\ensuremath{\mathbb{R}}},
	description={State vector to dynamic linear system}
}
\newglossaryentry{hyk}{type=symbols,
	sort={xk},
	name={\ensuremath{\gls{hi}\lr{\gls{Mkm1}}}},
	dimension={n},
	body={\ensuremath{\mathbb{R}}},
	description={State vector to dynamic linear system}
}
\newglossaryentry{hykb}{type=symbols,
	sort={xk},
	name={\ensuremath{\gls{hi}\lr{\gls{Mkm1b}}}},
	dimension={n},
	body={\ensuremath{\mathbb{R}}},
	description={State vector to dynamic linear system}
}
\newglossaryentry{h0ykb}{type=symbols,
	sort={xk},
	name={\ensuremath{\gls{hi}\lr{\norm{\bm{0} - \overline{\bm{y}}_{\text{r},k}}}}},
	dimension={n},
	body={\ensuremath{\mathbb{R}}},
	description={State vector to dynamic linear system}
}

\newglossaryentry{n}{type=symbols,
	sort={n},
	name={\ensuremath{l}},
	dimension={},
	body={\ensuremath{\mathbb{R}^+}},
	description={Conflict probability for the $i$-th step, given that no collision occured in previous steps}
}

\newglossaryentry{Pcoli}{type=symbols,
	sort={Pcoli},
	name={\ensuremath{p_{\text{col},k}}}, %name={\ensuremath{\mathbb{P}_{\text{col},k}}},
	dimension={},
	body={\ensuremath{\mathbb{R}^+}},
	description={Conflict probability for the $i$-th step, given that no collision occured in previous steps}
}
\newglossaryentry{Pcoli1}{type=symbols,
	sort={Pcoli},
	name={\ensuremath{p_{\text{col},1}}}, %name={\ensuremath{\mathbb{P}_{\text{col},k}}},
	dimension={},
	body={\ensuremath{\mathbb{R}^+}},
	description={Conflict probability for the $i$-th step, given that no collision occured in previous steps}
}
\newglossaryentry{d}{type=symbols,
	sort={d},
	name={\ensuremath{d_k}},
	dimension={},
	body={\ensuremath{\mathbb{R}^+}},
	description={Conflict probability for the $i$-th step, given that no collision occured in previous steps}
}
\newglossaryentry{rr}{type=symbols,
	sort={d},
	name={\ensuremath{r_{\text{r}}}},
	dimension={},
	body={\ensuremath{\mathbb{R}^+}},
	description={Conflict probability for the $i$-th step, given that no collision occured in previous steps}
}
\newglossaryentry{rcv}{type=symbols,
	sort={d},
	name={\ensuremath{r_{\text{c}}}},
	dimension={},
	body={\ensuremath{\mathbb{R}^+}},
	description={Conflict probability for the $i$-th step, given that no collision occured in previous steps}
}
\newglossaryentry{rcomb}{type=symbols,
	sort={d},
	name={\ensuremath{r_{\text{comb}}}},
	dimension={},
	body={\ensuremath{\mathbb{R}^+}},
	description={Conflict probability for the $i$-th step, given that no collision occured in previous steps}
}
\newglossaryentry{rwi}{type=symbols,
	sort={Rsi},
	name={\ensuremath{r_{\text{occ},k}}},
	dimension={},
	body={\ensuremath{\mathbb{R}^+}},
	description={Radius defining the boundary of $\supp \gls{pconi}$}
}
\newglossaryentry{rti}{type=symbols,
	sort={Rsi},
	name={\ensuremath{d_{\text{safe},k}}},
	dimension={},
	body={\ensuremath{\mathbb{R}^+}},
	description={Radius defining the boundary of $\supp \gls{pconi}$}
}
\newglossaryentry{Pw}{type=symbols,
	sort={Pw},
	name={\ensuremath{p_{\bm{W}_k}}}, %name={\ensuremath{\mathbb{P}_{\bm{W}_k}}},
	dimension={},
	body={\ensuremath{\mathbb{R}}},
	description={Probability density function for random variable \gls{Wi}}
}
\newglossaryentry{fPw}{type=symbols,
	sort={fPw},
	name={\ensuremath{f_{\bm{W}_k}}},
	dimension={},
	body={\ensuremath{\mathbb{R}}},
	description={Probability density function for random variable \gls{Wi}}
}
\newglossaryentry{fPw2}{type=symbols,
	sort={fPw},
	name={\ensuremath{f_{\bm{W}_{k,\text{pol}}}}},
	dimension={},
	body={\ensuremath{\mathbb{R}}},
	description={Probability density function for random variable \gls{Wi}}
}

\newglossaryentry{fPwi}{type=symbols,
	sort={fPw},
	name={\ensuremath{f_{\bm{W}_i}}},
	dimension={},
	body={\ensuremath{\mathbb{R}}},
	description={Probability density function for random variable \gls{Wi}}
}
\newglossaryentry{pp}{type=symbols,
	sort={pp},
	name={\ensuremath{\bm{u}^*_0}},
	dimension={},
	body={\ensuremath{\mathbb{R}}},
	description={Probability density function for random variable \gls{Wi}}
}
\newglossaryentry{convfunc}{type=symbols,
	sort={convfunc},
	name={\ensuremath{z}},
	dimension={},
	body={\ensuremath{\mathbb{R}}},
	description={Probability density function for random variable \gls{Wi}}
}
\newglossaryentry{wi}{type=symbols,
	sort={wi},
	name={\ensuremath{\bm{w}_{k}}},
	dimension={\glsd{wi}},
	body={\ensuremath{\mathbb{R}}},
	description={Realization of \gls{Yri}}
}
\newglossaryentry{wii}{type=symbols,
	sort={wi},
	name={\ensuremath{\bm{w}_{i}}},
	dimension={\glsd{wi}},
	body={\ensuremath{\mathbb{R}}},
	description={Realization of \gls{Yri}}
}
\newglossaryentry{wi0}{type=symbols,
	sort={wi},
	name={\ensuremath{\bm{w}_{0}}},
	dimension={\glsd{wi}},
	body={\ensuremath{\mathbb{R}}},
	description={Realization of \gls{Yri}}
}
\newglossaryentry{wimax}{type=symbols,
	sort={wi},
	name={\ensuremath{\vark{w}{max}}},
	dimension={\glsd{wi}},
	body={\ensuremath{\mathbb{R}}},
	description={Realization of \gls{Yri}}
}
\newglossaryentry{wimaxt}{type=symbols,
	sort={wi},
	name={\ensuremath{\tilde{\bm{w}}_{\text{max},k}}},
	dimension={\glsd{wi}},
	body={\ensuremath{\mathbb{R}}},
	description={Realization of \gls{Yri}}
}
\newglossaryentry{wimaxii}{type=symbols,
	sort={wi},
	name={\ensuremath{\vari{w}{max}}},
	dimension={\glsd{wi}},
	body={\ensuremath{\mathbb{R}}},
	description={Realization of \gls{Yri}}
}
\newglossaryentry{wimaxm1}{type=symbols,
	sort={wi},
	name={\ensuremath{w_{\text{max},k-1}}},
	dimension={\glsd{wi}},
	body={\ensuremath{\mathbb{R}}},
	description={Realization of \gls{Yri}}
}
\newglossaryentry{wimax0}{type=symbols,
	sort={wi},
	name={\ensuremath{\varz{w}{max}}},
	dimension={\glsd{wi}},
	body={\ensuremath{\mathbb{R}}},
	description={Realization of \gls{Yri}}
}
\newglossaryentry{pcoli}{type=symbols,
	sort={pcoli},
	name={\ensuremath{f_{\text{col},k}}},
	dimension={},
	body={\ensuremath{\mathbb{R}^+}},
	description={Conflict probability density for the $i$-th step, given that no collision occured in previous steps}
}
\newglossaryentry{pcv}{type=symbols,
	sort={pcon},
	name={\ensuremath{\var{pp}{cv}}},
	dimension={},
	body={\ensuremath{\mathbb{R}^+}},
	description={Conflict probability}
}
\newglossaryentry{pocc}{type=symbols,
	sort={pcon},
	name={\ensuremath{\vark{f}{occ}}},
	dimension={},
	body={\ensuremath{\mathbb{R}^+}},
	description={Conflict probability}
}
\newglossaryentry{Pcv}{type=symbols,
	sort={pcon},
	name={\ensuremath{p_{\text{cv},k}}},
	dimension={},
	body={\ensuremath{\mathbb{R}^+}},
	description={Conflict probability}
}
\newglossaryentry{Pcv1}{type=symbols,
	sort={pcon},
	name={\ensuremath{p_{\text{cv},1}}},
	dimension={},
	body={\ensuremath{\mathbb{R}^+}},
	description={Conflict probability}
}
\newglossaryentry{att0}{type=symbols,
	sort={pcon},
	name={\ensuremath{\lr{\bm{u}_{0}}}},
	dimension={},
	body={\ensuremath{\mathbb{R}^+}},
	description={Conflict probability}
}
\newglossaryentry{attj}{type=symbols,
	sort={pcon},
	name={\ensuremath{\lr{\bm{u}_{j}}}},
	dimension={},
	body={\ensuremath{\mathbb{R}^+}},
	description={Conflict probability}
}
\newglossaryentry{Pcv1att}{type=symbols,
	sort={pcon},
	name={\ensuremath{p_{\text{cv},1}\gls{att0}}},
	dimension={},
	body={\ensuremath{\mathbb{R}^+}},
	description={Conflict probability}
}
\newglossaryentry{fPcv}{type=symbols,
	sort={pcon},
	name={\ensuremath{f_{\text{cv},k}}},
	dimension={},
	body={\ensuremath{\mathbb{R}^+}},
	description={Conflict probability}
}
\newglossaryentry{XXnorm}{type=symbols,
	sort={XXf},
	name={\ensuremath{\mathcal{X}_{\text{cv},k}}},
	dimension={},
	body={\ensuremath{\mathbb{R}^+}},
	description={set forbidden by norm}
}
\newglossaryentry{Vf}{type=symbols,
	sort={Vf},
	name={\ensuremath{V_{\text{f}}\lr{\bm{x}_{N}}}},
	dimension={},
	body={\ensuremath{\mathbb{R}^+}},
	description={terminal cost}
}
\newglossaryentry{XXf}{type=symbols,
	sort={XXf},
	name={\ensuremath{\mathcal{X}_\text{f}}},
	dimension={},
	body={\ensuremath{\mathbb{R}^+}},
	description={terminal set}
}
\newglossaryentry{Rswj}{type=symbols,
	sort={Rswj},
	name={\ensuremath{R_{\text{s}_{\text{W}},j}}},
	dimension={},
	body={\ensuremath{\mathbb{R}^+}},
	description={Radius of support boundary of \gls{Wi}}
}
\newglossaryentry{M}{type=symbols,
	sort={M},
	name={\ensuremath{\norm{\bm{y}_1-\bm{y}_{\text{r},1}}_2}},
	dimension={},
	body={\ensuremath{\mathbb{R}^+}},
	description={Measure}
}
\newglossaryentry{Mm1}{type=symbols,
	sort={M},
	name={\ensuremath{\norm{\bm{y}_1-\bm{y}_{\text{r},0}}_2}},
	dimension={},
	body={\ensuremath{\mathbb{R}^+}},
	description={Measure}
}
\newglossaryentry{Mm1b}{type=symbols,
	sort={M},
	name={\ensuremath{\norm{\bm{y}_1-\overline{\bm{y}}_{\text{r},1}}_2}},
	dimension={},
	body={\ensuremath{\mathbb{R}^+}},
	description={Measure}
}
\newglossaryentry{Mm1us}{type=symbols,
	sort={M},
	name={\ensuremath{\norm{\bm{y}_1\lr{\gls{pp}}-\overline{\bm{y}}_{\text{r},1}}_2}},
	dimension={},
	body={\ensuremath{\mathbb{R}^+}},
	description={Measure}
}
\newglossaryentry{Mk}{type=symbols,
	sort={M},
	name={\ensuremath{\norm{\gls{yi}-\gls{yri}}_2}},
	dimension={},
	body={\ensuremath{\mathbb{R}^+}},
	description={Measure}
}
\newglossaryentry{Mkb}{type=symbols,
	sort={M},
	name={\ensuremath{\norm{\gls{yi}-\gls{yrib}}_2}},
	dimension={},
	body={\ensuremath{\mathbb{R}^+}},
	description={Measure}
}
\newglossaryentry{Mkp1}{type=symbols,
	sort={M},
	name={\ensuremath{\norm{\gls{yip1}-\gls{yrip1}}_2}},
	dimension={},
	body={\ensuremath{\mathbb{R}^+}},
	description={Measure}
}
\newglossaryentry{Mkj}{type=symbols,
	sort={M},
	name={\ensuremath{\norm{\gls{yj}-\gls{yrj}}_2}},
	dimension={},
	body={\ensuremath{\mathbb{R}^+}},
	description={Measure}
}
\newglossaryentry{Mkm1}{type=symbols,
	sort={M},
	name={\ensuremath{\norm{\gls{yi}-\bm{y}_{\text{r},k-1}}_2}},
	dimension={},
	body={\ensuremath{\mathbb{R}^+}},
	description={Measure}
}
\newglossaryentry{Mkm1b}{type=symbols,
	sort={M},
	name={\ensuremath{\norm{\gls{yi}-\overline{\bm{y}}_{\text{r},k}}_2}},
	dimension={},
	body={\ensuremath{\mathbb{R}^+}},
	description={Measure}
}
\newglossaryentry{Mkm1ur}{type=symbols,
	sort={M},
	name={\ensuremath{\norm{\gls{yi}-\lr{\bm{y}_{\text{r},k-1} + \bm{u}_{\text{r},k-1}}}_2}},
	dimension={},
	body={\ensuremath{\mathbb{R}^+}},
	description={Measure}
}
\newglossaryentry{Mkm1urt}{type=symbols,
	sort={M},
	name={\ensuremath{\norm{\gls{yi}-\tilde{\bm{y}}_{\text{r},k-1}}_2}},
	dimension={},
	body={\ensuremath{\mathbb{R}^+}},
	description={Measure}
}
\newglossaryentry{Mp}{type=symbols,
	sort={M},
	name={\ensuremath{\norm{\bm{y}_{j}-\bm{y}_{\text{r},j}}_p}},
	dimension={},
	body={\ensuremath{\mathbb{R}^+}},
	description={Measure}
}
\newglossaryentry{Mj}{type=symbols,
	sort={M},
	name={\ensuremath{\norm{\bm{y}_{j}-\bm{y}_{\text{r},j}}_2}},
	dimension={},
	body={\ensuremath{\mathbb{R}^+}},
	description={Measure}
}
\newglossaryentry{Mjb}{type=symbols,
	sort={M},
	name={\ensuremath{\norm{\bm{y}_{j}-\overline{\bm{y}}_{\text{r},j}}_2}},
	dimension={},
	body={\ensuremath{\mathbb{R}^+}},
	description={Measure}
}
\newglossaryentry{Mcon}{type=symbols,
	sort={Mmin},
	name={\ensuremath{c}},
	dimension={},
	body={\ensuremath{\mathbb{R}^+}},
	description={constraint for measure}
}
\newglossaryentry{Mmin}{type=symbols,
	sort={Mmin},
	name={\ensuremath{c_{1}}},
	dimension={},
	body={\ensuremath{\mathbb{R}^+}},
	description={constraint for measure}
}
\newglossaryentry{Mtot1}{type=symbols,
	sort={Mmin},
	name={\ensuremath{\gls{Mmin}+\gls{wimax0}}},%name={\ensuremath{c_{\text{total},1}}},
	dimension={},
	body={\ensuremath{\mathbb{R}^+}},
	description={constraint for measure}
}
\newglossaryentry{Mtotl}{type=symbols,
	sort={Mmin},
	name={\ensuremath{\gls{Mcon}+\sum_{k=0}^{\gls{n}-1}\gls{wimax}}},%name={\ensuremath{c_{\text{total},1}}},
	dimension={},
	body={\ensuremath{\mathbb{R}^+}},
	description={constraint for measure}
}
\newglossaryentry{Mtotli}{type=symbols,
	sort={Mmin},
	name={\ensuremath{\gls{Mcon}+\sum_{i=0}^{\gls{n}-1}\gls{wimaxii}}},%name={\ensuremath{c_{\text{total},1}}},
	dimension={},
	body={\ensuremath{\mathbb{R}^+}},
	description={constraint for measure}
}
\newglossaryentry{Mtot}{type=symbols,
	sort={Mmin},
	name={\ensuremath{\gls{Mmin}+\gls{wimax}}},%name={\ensuremath{c_{\text{total},1}}},
	dimension={},
	body={\ensuremath{\mathbb{R}^+}},
	description={constraint for measure}
}
\newglossaryentry{Mmink}{type=symbols,
	sort={Mmin},
	name={\ensuremath{{c}_k}},%name={\ensuremath{\vark{c}{\text{min}}}},
	dimension={},
	body={\ensuremath{\mathbb{R}^+}},
	description={constraint for measure}
}
\newglossaryentry{Mminj}{type=symbols,
	sort={Mmin},
	name={\ensuremath{{c}_j}},%name={\ensuremath{\vark{c}{\text{min}}}},
	dimension={},
	body={\ensuremath{\mathbb{R}^+}},
	description={constraint for measure}
}
\newglossaryentry{Mminkj}{type=symbols,
	sort={Mmin},
	name={\ensuremath{{c}_j}},%name={\ensuremath{\vark{c}{\text{min}}}},
	dimension={},
	body={\ensuremath{\mathbb{R}^+}},
	description={constraint for measure}
}
\newglossaryentry{Mtotk}{type=symbols,
	sort={Mmin},
	name={\ensuremath{\vark{c}{\text{total}}}},
	dimension={},
	body={\ensuremath{\mathbb{R}^+}},
	description={constraint for measure}
}
\newglossaryentry{CC}{type=symbols,
	sort={CC},
	name={\ensuremath{\mathbb{C}_{\textnormal{ch}}}},
	dimension={},
	body={},
	description={Chance Constraint}
}
\newglossaryentry{u1j}{type=symbols,
	sort={u1j},
	name={\ensuremath{\bm{u}_{1:j}}},
	dimension={m},
	body={\ensuremath{\mathbb{R}}},
	description={Input vector to dynamic linear system}
}
\newglossaryentry{hiyy}{type=symbols,
	sort={hi},
	name={\ensuremath{\f{\gls{hi}}{\gls{Mkm1}}}},%name={\ensuremath{\f{\gls{hi}}{\gls{yi},\gls{yri}}}},
	dimension={},
	body={\ensuremath{\mathbb{R}}},
	description={Measure function at step $i$; $\f{h^{(i)}}{\f{\gls{yi}}{\bm{u}^{(0:i-1)}},\gls{yri}} = \f{g^{(i)}}{\norm{\f{\gls{yi}}{\bm{u}^{(0:i-1)}},\gls{yri}}_2}$}
}
\newglossaryentry{hiyyurt}{type=symbols,
	sort={hi},
	name={\ensuremath{\f{\gls{hi}}{\gls{Mkm1urt}}}},%name={\ensuremath{\f{\gls{hi}}{\gls{yi},\gls{yri}}}},
	dimension={},
	body={\ensuremath{\mathbb{R}}},
	description={Measure function at step $i$; $\f{h^{(i)}}{\f{\gls{yi}}{\bm{u}^{(0:i-1)}},\gls{yri}} = \f{g^{(i)}}{\norm{\f{\gls{yi}}{\bm{u}^{(0:i-1)}},\gls{yri}}_2}$}
}
\newglossaryentry{hiyy1}{type=symbols,
	sort={hi},
	name={\ensuremath{\f{\gls{hi}}{\gls{Mm1}}}},%name={\ensuremath{\f{\gls{hi}}{\gls{yi},\gls{yri}}}},
	dimension={},
	body={\ensuremath{\mathbb{R}}},
	description={Measure function at step $i$; $\f{h^{(i)}}{\f{\gls{yi}}{\bm{u}^{(0:i-1)}},\gls{yri}} = \f{g^{(i)}}{\norm{\f{\gls{yi}}{\bm{u}^{(0:i-1)}},\gls{yri}}_2}$}
}
\newglossaryentry{hiyy1b}{type=symbols,
	sort={hi},
	name={\ensuremath{\f{\gls{hi}}{\gls{Mm1b}}}},%name={\ensuremath{\f{\gls{hi}}{\gls{yi},\gls{yri}}}},
	dimension={},
	body={\ensuremath{\mathbb{R}}},
	description={Measure function at step $i$; $\f{h^{(i)}}{\f{\gls{yi}}{\bm{u}^{(0:i-1)}},\gls{yri}} = \f{g^{(i)}}{\norm{\f{\gls{yi}}{\bm{u}^{(0:i-1)}},\gls{yri}}_2}$}
}
\newglossaryentry{hiyy1wmax}{type=symbols,
	sort={hi},
	name={\ensuremath{\f{\gls{hi}}{\gls{Mm1}-\gls{wimax0}}}},%name={\ensuremath{\f{\gls{hi}}{\gls{yi},\gls{yri}}}},
	dimension={},
	body={\ensuremath{\mathbb{R}}},
	description={Measure function at step $i$; $\f{h^{(i)}}{\f{\gls{yi}}{\bm{u}^{(0:i-1)}},\gls{yri}} = \f{g^{(i)}}{\norm{\f{\gls{yi}}{\bm{u}^{(0:i-1)}},\gls{yri}}_2}$}
}
\newglossaryentry{hiyy1wmaxus}{type=symbols,
	sort={hi},
	name={\ensuremath{\f{\gls{hi}}{\gls{Mm1us}-\gls{wimax0}}}},%name={\ensuremath{\f{\gls{hi}}{\gls{yi},\gls{yri}}}},
	dimension={},
	body={\ensuremath{\mathbb{R}}},
	description={Measure function at step $i$; $\f{h^{(i)}}{\f{\gls{yi}}{\bm{u}^{(0:i-1)}},\gls{yri}} = \f{g^{(i)}}{\norm{\f{\gls{yi}}{\bm{u}^{(0:i-1)}},\gls{yri}}_2}$}
}
\newglossaryentry{hiyy1us}{type=symbols,
	sort={hi},
	name={\ensuremath{\f{\gls{hi}}{\gls{Mm1us}}}},%name={\ensuremath{\f{\gls{hi}}{\gls{yi},\gls{yri}}}},
	dimension={},
	body={\ensuremath{\mathbb{R}}},
	description={Measure function at step $i$; $\f{h^{(i)}}{\f{\gls{yi}}{\bm{u}^{(0:i-1)}},\gls{yri}} = \f{g^{(i)}}{\norm{\f{\gls{yi}}{\bm{u}^{(0:i-1)}},\gls{yri}}_2}$}
}
\newglossaryentry{hiyy11}{type=symbols,
	sort={hi},
	name={\ensuremath{\f{\gls{hi}}{\gls{M}}}},%name={\ensuremath{\f{\gls{hi}}{\gls{yi},\gls{yri}}}},
	dimension={},
	body={\ensuremath{\mathbb{R}}},
	description={Measure function at step $i$; $\f{h^{(i)}}{\f{\gls{yi}}{\bm{u}^{(0:i-1)}},\gls{yri}} = \f{g^{(i)}}{\norm{\f{\gls{yi}}{\bm{u}^{(0:i-1)}},\gls{yri}}_2}$}
}
\newglossaryentry{xip1}{type=symbols,
	sort={xi},
	name={\ensuremath{\bm{x}_{k+1}}},
	dimension={n},
	body={\ensuremath{\mathbb{R}}},
	description={State vector to dynamic linear system}
}

\newglossaryentry{A}{type=symbols,
	sort={A},
	name={\ensuremath{A}},
	dimension={\glsd{xi}\times \glsd{xi}},
	body={\ensuremath{\mathbb{R}}},
	description={System matrix}
}
\newglossaryentry{B}{type=symbols,
	sort={B},
	name={\ensuremath{B}},
	dimension={\glsd{xi}\times \glsd{ui}},
	body={\ensuremath{\mathbb{R}}},
	description={Input matrix}
}
\newglossaryentry{C}{type=symbols,
	sort={C},
	name={\ensuremath{C}},
	dimension={\glsd{yi}\times \glsd{xi}},
	body={\ensuremath{\mathbb{R}}},
	description={Output matrix}
}
\newglossaryentry{Nms}{type=symbols,
	sort={N},
	name={\ensuremath{\var{N}{multistep}}},
	dimension={},
	body={\ensuremath{\mathbb{N}}},
	description={Prediction horizon for calculation of \gls{success_prob}}
}
\newglossaryentry{Nmpc}{type=symbols,
	sort={Nmpc},
	name={\ensuremath{N}}, %\var{N}{mpc}
	dimension={},
	body={\ensuremath{\mathbb{N}}},
	description={Prediction horizon for MPC problem}
}
\newglossaryentry{di}{type=symbols,
	sort={di},
	name={\ensuremath{d_k}},
	dimension={},
	body={\ensuremath{\mathbb{R}^+}},
	description={Distance between estimated position of obstacle, \gls{yr}, and controlled object, \gls{yr}}
}
\newglossaryentry{gi}{type=symbols,
	sort={g},
	name={\ensuremath{\gls{hi}}},
	dimension={},
	body={\ensuremath{\mathbb{R}}},
	description={Alternative definition of \gls{hi}}
}
\newglossaryentry{hi}{type=symbols,
	sort={hi},
	name={\ensuremath{h}},
	dimension={},
	body={\ensuremath{\mathbb{R}}},
	description={Measure function at step $i$; $\f{h^{(i)}}{\f{\gls{yi}}{\bm{u}^{(0:i-1)}},\gls{yri}} = \f{g^{(i)}}{\norm{\f{\gls{yi}}{\bm{u}^{(0:i-1)}},\gls{yri}}_2}$}
}
\newglossaryentry{hic}{type=symbols,
	sort={hic},
	name={\ensuremath{\partial \gls{UUm3} }}, % \vark{h}{c}}},
	dimension={\glsd{ui}},
	body={\ensuremath{\mathbb{R}}},
	description={Level set of \gls{hi}, i.e. $\f{\gls{hi}}{\gls{yi},\gls{yr}}=c, \ c \in \mathbb{R}$}
}
\newglossaryentry{hic0}{type=symbols,
	sort={hic},
	name={\ensuremath{\partial \gls{UUm30} }}, % \vark{h}{c}}},
	dimension={\glsd{ui}},
	body={\ensuremath{\mathbb{R}}},
	description={Level set of \gls{hi}, i.e. $\f{\gls{hi}}{\gls{yi},\gls{yr}}=c, \ c \in \mathbb{R}$}
}
\newglossaryentry{hmaxi}{type=symbols,
	sort={hmaxi},
	name={\ensuremath{\vark{h}{max}}},
	dimension={},
	body={\ensuremath{\mathbb{R}}},
	description={Maximum of the value function for one step, i.e. $\gls{hmaxi}=\f{h}{\f{\bm{y}}{\gls{umaxi}},\gls{yr}}$}
}
\newglossaryentry{hmaxi1}{type=symbols,
	sort={hmaxi},
	name={\ensuremath{h_{\text{max},1}}},
	dimension={},
	body={\ensuremath{\mathbb{R}}},
	description={Maximum of the value function for one step, i.e. $\gls{hmaxi}=\f{h}{\f{\bm{y}}{\gls{umaxi}},\gls{yr}}$}
}
\newglossaryentry{hmaxil}{type=symbols,
	sort={hmaxi},
	name={\ensuremath{h_{\text{max},l}}},
	dimension={},
	body={\ensuremath{\mathbb{R}}},
	description={Maximum of the value function for one step, i.e. $\gls{hmaxi}=\f{h}{\f{\bm{y}}{\gls{umaxi}},\gls{yr}}$}
}
\newglossaryentry{hmaxi0}{type=symbols,
	sort={hmaxi},
	name={\ensuremath{h_{\text{max},0}}},
	dimension={},
	body={\ensuremath{\mathbb{R}}},
	description={Maximum of the value function for one step, i.e. $\gls{hmaxi}=\f{h}{\f{\bm{y}}{\gls{umaxi}},\gls{yr}}$}
}
\newglossaryentry{hmini}{type=symbols,
	sort={hmini},
	name={\ensuremath{\vark{h}{min}}},
	dimension={},
	body={\ensuremath{\mathbb{R}}},
	description={Minimum of the value function for one step, i.e. $\gls{hmini}=\f{h}{\f{\bm{y}}{\gls{umini}},\gls{yr}}$}
}
\newglossaryentry{hmini1}{type=symbols,
	sort={hmin},
	name={\ensuremath{h_{\text{min},1}}},
	dimension={},
	body={\ensuremath{\mathbb{R}}},
	description={Maximum of the value function for one step, i.e. $\gls{hmaxi}=\f{h}{\f{\bm{y}}{\gls{umaxi}},\gls{yr}}$}
}
\newglossaryentry{hmini0}{type=symbols,
	sort={hmin},
	name={\ensuremath{h_{\text{min},0}}},
	dimension={},
	body={\ensuremath{\mathbb{R}}},
	description={Maximum of the value function for one step, i.e. $\gls{hmaxi}=\f{h}{\f{\bm{y}}{\gls{umaxi}},\gls{yr}}$}
}
\newglossaryentry{hminil}{type=symbols,
	sort={hmin},
	name={\ensuremath{h_{\text{min},l}}},
	dimension={},
	body={\ensuremath{\mathbb{R}}},
	description={Maximum of the value function for one step, i.e. $\gls{hmaxi}=\f{h}{\f{\bm{y}}{\gls{umaxi}},\gls{yr}}$}
}

\newglossaryentry{J}{type=symbols,
	sort={J},
	name={\ensuremath{\f{J}{\hat{\bm{u}}_k;\bm{x}_k}}},
	dimension={},
	body={\ensuremath{\mathbb{R}}},
	description={Objective function of the MPC optimization problem}
}
\newglossaryentry{Oci}{type=symbols,
	sort={Oci},
	name={\ensuremath{\vark{\set{O}}{c}}},
	dimension={\glsd{yi}},
	body={\ensuremath{\glsbody{yi}}},
	description={Area in the output space occupied by the controlled object with center \gls{yi}}
}
\newglossaryentry{Ooi}{type=symbols,
	sort={Ooi},
	name={\ensuremath{\vark{\set{O}}{o}}},
	dimension={\glsd{yri}},
	body={\ensuremath{\glsbody{yri}}},
	description={Area in the output space occupied by the controlled object with center \gls{yri}}
}
\newglossaryentry{success_prob}{type=symbols,
	sort={success_prob},
	name={\ensuremath{\f{P_k}{\bm{u}_{(1:k)}}}},
	dimension={},
	body={\ensuremath{\left[0;1\right]}},
	description={Success probability for the $i$-th step, given that no collision occured in previous steps}
}
\newglossaryentry{success_prob_N}{type=symbols,
	sort={success_prob_N},
	name={\ensuremath{\prob{\bm{u}}}},
	dimension={},
	body={\ensuremath{\left[0;1\right]}},
	description={Total success probability over \gls{N} steps}
}
\newglossaryentry{conflict_prob}{type=symbols,
	sort={conflict_prob},
	name={\ensuremath{\f{\vark{P}{con}}{\bm{u}_{(1:k)}}}},
	dimension={},
	body={\ensuremath{\left[0;1\right]}},
	description={Conflict probability for the $i$-th step, given that no collision occured in previous steps; equivalent to $1-\gls{success_prob}$}
}
\newglossaryentry{conflict_prob_N}{type=symbols,
	sort={conflict_prob_N},
	name={\ensuremath{\prob[c]{\bm{u}}}},
	dimension={},
	body={\ensuremath{\left[0;1\right]}},
	description={Total conflict probability over \gls{N} steps}
}
\newglossaryentry{pconi}{type=symbols,
	sort={pconi},
	name={\ensuremath{\vark{p}{con}}},
	dimension={},
	body={\ensuremath{\mathbb{R}^+}},
	description={Conflict probability density for the $i$-th step, given that no collision occured in previous steps}
}
\newglossaryentry{pwi}{type=symbols,
	sort={pwi},
	name={\ensuremath{\var{p}{\gls{Wi}}}},
	dimension={},
	body={\ensuremath{\mathbb{R}}},
	description={Probability density function for random variable \gls{Wi}}
}
\newglossaryentry{pw1i}{type=symbols,
	sort={pw1i},
	name={\ensuremath{\varkk{p}{\bm{\rv{W}}}}},
	dimension={},
	body={\ensuremath{\mathbb{R}}},
	description={Probability density function for random variable $\bm{\rv{W}}^{(1:i)}=\sum_{j=1}^{i}\bm{\rv{W}}^{(j)}$}
}
\newglossaryentry{Rc}{type=symbols,
	sort={Rc},
	name={\ensuremath{\var{R}{c}}},
	dimension={},
	body={\ensuremath{\mathbb{R}^+}},
	description={Radius of boundary around controlled object}
}
\newglossaryentry{Ro}{type=symbols,
	sort={Ro},
	name={\ensuremath{\var{R}{o}}},
	dimension={},
	body={\ensuremath{\mathbb{R}^+}},
	description={Radius of boundary around obstacle}
}
\newglossaryentry{Rsi}{type=symbols,
	sort={Rsi},
	name={\ensuremath{\vark{R}{s}}},
	dimension={},
	body={\ensuremath{\mathbb{R}^+}},
	description={Radius defining the boundary of $\supp \gls{pconi}$}
}
\newglossaryentry{Rswi}{type=symbols,
	sort={Rswi},
	name={\ensuremath{R_{\text{s}_{\text{W}},k}}},
	dimension={},
	body={\ensuremath{\mathbb{R}^+}},
	description={Radius of support boundary of \gls{Wi}}
}
\newglossaryentry{T}{type=symbols,
	sort={T},
	name={\ensuremath{T}},
	dimension={},
	body={\ensuremath{\mathbb{R}^+}},
	description={Sampling time}
}
\newglossaryentry{Tsim}{type=symbols,
	sort={Tsim},
	name={\ensuremath{\var{T}{sim}}},
	dimension={},
	body={\ensuremath{\mathbb{R}^+}},
	description={Total simulation time}
}
\newglossaryentry{umaxi}{type=symbols,
	sort={umaxi},
	name={\ensuremath{\vark{\bm{u}}{max}}},
	dimension={},
	body={\ensuremath{\mathbb{R}}},
	description={Input vector for one step which achieves \gls{hmaxi}}
}
\newglossaryentry{umaxi0}{type=symbols,
	sort={umaxi},
	name={\ensuremath{\bm{u}_{\text{cvpm},0}}},
	dimension={},
	body={\ensuremath{\mathbb{R}}},
	description={Input vector for one step which achieves \gls{hmaxi}}
}
\newglossaryentry{umaxi0x}{type=symbols,
	sort={umaxi},
	name={\ensuremath{\bm{u}_{\text{max},0}}},
	dimension={},
	body={\ensuremath{\mathbb{R}}},
	description={Input vector for one step which achieves \gls{hmaxi}}
}
\newglossaryentry{umini}{type=symbols,
	sort={umini},
	name={\ensuremath{\vark{\bm{u}}{min}}},
	dimension={},
	body={\ensuremath{\mathbb{R}}},
	description={Input vector for one step which achieves \gls{hmini}}
}
\newglossaryentry{umini0}{type=symbols,
	sort={umini},
	name={\ensuremath{\bm{u}_{\text{min},0}}},
	dimension={},
	body={\ensuremath{\mathbb{R}}},
	description={Input vector for one step which achieves \gls{hmini}}
}
\newglossaryentry{UUpit}{type=symbols,
	sort={UUpit},
	name={\ensuremath{\var{\set{U}}{PIT}}},
	dimension={},
	body={\ensuremath{\mathbb{R}}},
	description={\gls{PIT} constraints}
}
\newglossaryentry{UUdiff}{type=symbols,
	sort={UUdiff},
	name={\ensuremath{\var{\set{U}}{diff}}},
	dimension={},
	body={\ensuremath{\mathbb{R}}},
	description={Difference constraints}
}
\newglossaryentry{uhati}{type=symbols,
	sort={uhati},
	name={\ensuremath{\hat{\bm{u}}_k}},
	dimension={},
	body={\ensuremath{\mathbb{R}}},
	description={Vector of set \gls{UU_total} at time step $i$}
}
\newglossaryentry{ui}{type=symbols,
	sort={ui},
	name={\ensuremath{\bm{u}_k}},
	dimension={m},
	body={\ensuremath{\mathbb{R}}},
	description={Input vector to dynamic linear system}
}
\newglossaryentry{ui0}{type=symbols,
	sort={ui},
	name={\ensuremath{\bm{u}_0}},
	dimension={m},
	body={\ensuremath{\mathbb{R}}},
	description={Input vector to dynamic linear system}
}
\newglossaryentry{uj}{type=symbols,
	sort={uj},
	name={\ensuremath{\bm{u}_j}},
	dimension={m},
	body={\ensuremath{\mathbb{R}}},
	description={Input vector to dynamic linear system}
}
\newglossaryentry{xj}{type=symbols,
	sort={xj},
	name={\ensuremath{\bm{x}_j}},
	dimension={m},
	body={\ensuremath{\mathbb{R}}},
	description={Input vector to dynamic linear system}
}
\newglossaryentry{UUi}{type=symbols,
	sort={UUi},
	name={\ensuremath{\set{U}}},%_k}}, %name={\ensuremath{\set{U}_k}}
	dimension={\glsd{ui}},
	body={\ensuremath{\mathbb{R}}},
	description={Input set of admissable \gls{ui}}
}
\newglossaryentry{UUj}{type=symbols,
	sort={UUj},
	name={\ensuremath{\set{U}}}, %name={\ensuremath{\set{U}_k}}
	dimension={\glsd{ui}},
	body={\ensuremath{\mathbb{R}}},
	description={Input set of admissable \gls{ui}}
}
\newglossaryentry{UUti}{type=symbols,
	sort={UUti},
	name={\ensuremath{\f{\hat{\set{U}}_k}{\gls{pp}} }}, %\f{\tilde{\set{U}}_k}{\bm{p}}
	dimension={\glsd{ui}},
	body={\ensuremath{\mathbb{R}}},
	description={Approximation of non-convex set using hyperplane with reference point $\bm{p}$}
}
\newglossaryentry{UUti0}{type=symbols,
	sort={UUti},
	name={\ensuremath{\f{\hat{\set{U}}_0}{\gls{pp}} }}, %\f{\tilde{\set{U}}_k}{\bm{p}}
	dimension={\glsd{ui}},
	body={\ensuremath{\mathbb{R}}},
	description={Approximation of non-convex set using hyperplane with reference point $\bm{p}$}
}
\newglossaryentry{UUkp1}{type=symbols,
	sort={UUti},
	name={\ensuremath{\tilde{\set{U}}_k^N }},
	dimension={\glsd{ui}},
	body={\ensuremath{\mathbb{R}}},
	description={UkN without first input vector}
}
\newglossaryentry{UUNmpc}{type=symbols,
	sort={UUNmpc},
	name={\ensuremath{\set{U}_k^N }}, %\set{U}_{\gls{Nmpc}}
	dimension={\glsd{ui}\cdot\gls{Nmpc}},
	body={\ensuremath{\mathbb{R}}},
	description={Given set of admissable inputs, excluding \gls{UUopt}}
}
\newglossaryentry{uNmpc}{type=symbols,
	sort={uNmpc},
	name={\ensuremath{\bm{U}_{0} }}, %\var{\bm{u}}{\gls{Nmpc}}
	dimension={\glsd{ui}\cdot\gls{Nmpc}},
	body={\ensuremath{\mathbb{R}}},
	description={Input sequence for the \gls{Nmpc} step MPC problem}
}
\newglossaryentry{uNmpck}{type=symbols,
	sort={uNmpc},
	name={\ensuremath{\bm{U}_{k} }}, %\var{\bm{u}}{\gls{Nmpc}}
	dimension={\glsd{ui}\cdot\gls{Nmpc}},
	body={\ensuremath{\mathbb{R}}},
	description={Input sequence for the \gls{Nmpc} step MPC problem}
}
\newglossaryentry{UUopt}{type=symbols,
	sort={UUopt},
	name={\ensuremath{\var{\set{U}}{opt}^N }}, %\var{\set{U}}{opt}
	dimension={\glsd{ui}\cdot \gls{N}},
	body={\ensuremath{\mathbb{R}}},
	description={Set containing $\bm{u}$ which maximize success probability for \gls{N} steps}
}
\newglossaryentry{UUc}{type=symbols, % new %% U satisfying constraints
	sort={UUX},
	name={\ensuremath{\mathbb{U}_k }},
	dimension={\glsd{ui}\cdot \gls{Nmpc}},
	body={\ensuremath{\mathbb{R}}},
	description={Set containing all \gls{uNmpc} which fulfill the state constraint \gls{XX} at each time instant}
}
\newglossaryentry{UUxc}{type=symbols, % new %% U satisfying constraints
	sort={UUX},
	name={\ensuremath{\mathbb{U}_{\bm{x},k}}},
	dimension={\glsd{ui}\cdot \gls{Nmpc}},
	body={\ensuremath{\mathbb{R}}},
	description={Set containing all \gls{uNmpc} which fulfill the state constraint \gls{XX} at each time instant}
}
\newglossaryentry{UUxc0}{type=symbols, % new %% U satisfying constraints
	sort={UUX},
	name={\ensuremath{\mathbb{U}_{\bm{x},0}}},
	dimension={\glsd{ui}\cdot \gls{Nmpc}},
	body={\ensuremath{\mathbb{R}}},
	description={Set containing all \gls{uNmpc} which fulfill the state constraint \gls{XX} at each time instant}
}
\newglossaryentry{UUxcj}{type=symbols, % new %% U satisfying constraints
	sort={UUX},
	name={\ensuremath{\mathbb{U}_{\bm{x},j}}},
	dimension={\glsd{ui}\cdot \gls{Nmpc}},
	body={\ensuremath{\mathbb{R}}},
	description={Set containing all \gls{uNmpc} which fulfill the state constraint \gls{XX} at each time instant}
}
\newglossaryentry{UUm3}{type=symbols, % new %% U satisfying constraints
	sort={UUX},
	name={\ensuremath{\mathbb{U}_{\text{mode3},k}}},
	dimension={\glsd{ui}\cdot \gls{Nmpc}},
	body={\ensuremath{\mathbb{R}}},
	description={Set containing all \gls{uNmpc} which fulfill the state constraint \gls{XX} at each time instant}
}
\newglossaryentry{UUm30}{type=symbols, % new %% U satisfying constraints
	sort={UUX},
	name={\ensuremath{\mathbb{U}_{\text{mode3},0}}},
	dimension={\glsd{ui}\cdot \gls{Nmpc}},
	body={\ensuremath{\mathbb{R}}},
	description={Set containing all \gls{uNmpc} which fulfill the state constraint \gls{XX} at each time instant}
}
\newglossaryentry{UUX}{type=symbols,
	sort={UUX},
	name={\ensuremath{\set{U}_{\gls{XX}}}},
	dimension={\glsd{ui}\cdot \gls{Nmpc}},
	body={\ensuremath{\mathbb{R}}},
	description={Set containing all \gls{uNmpc} which fulfill the state constraint \gls{XX} at each time instant}
}
\newglossaryentry{UU_total}{type=symbols,
	sort={UU_total},
	name={\ensuremath{\mathbb{U}_k^* }}, %\f{\mathbb{U}}{\hat{\bm{u}}_i^{(0)}}
	dimension={\glsd{ui}\cdot \gls{Nmpc}},
	body={\ensuremath{\mathbb{R}}},
	description={Set of admissable inputs $\bm{u}$ for the complete \gls{MPC} problem given previous input $\hat{\bm{u}}_i^{(0)}$}
}
\newglossaryentry{UU_total0}{type=symbols,
	sort={UU_total},
	name={\ensuremath{\mathbb{U}_0^* }}, %\f{\mathbb{U}}{\hat{\bm{u}}_i^{(0)}}
	dimension={\glsd{ui}\cdot \gls{Nmpc}},
	body={\ensuremath{\mathbb{R}}},
	description={Set of admissable inputs $\bm{u}$ for the complete \gls{MPC} problem given previous input $\hat{\bm{u}}_i^{(0)}$}
}
\newglossaryentry{UUopti}{type=symbols,
	sort={UUopti},
	name={\ensuremath{\vark{\set{U}}{cvpm}}}, % was opt
	dimension={\glsd{ui}},
	body={\ensuremath{\mathbb{R}}},
	description={Set containing \gls{ui} which maximize success probability for one step}
}
\newglossaryentry{UUopti0}{type=symbols,
	sort={UUopti},
	name={\ensuremath{\set{U}_{\text{cvpm},0}}}, % was opt
	dimension={\glsd{ui}},
	body={\ensuremath{\mathbb{R}}},
	description={Set containing \gls{ui} which maximize success probability for one step}
}
\newglossaryentry{apUUopti}{type=symbols,
	sort={apUUopti},
	name={\ensuremath{\vark{\hat{\set{U}}}{cvpm}}},
	dimension={\glsd{ui}},
	body={\ensuremath{\mathbb{R}}},
	description={Approximated Set containing \gls{ui} which maximize success probability for one step}
}
\newglossaryentry{apUUopti0}{type=symbols,
	sort={apUUopti},
	name={\ensuremath{\varz{\hat{\set{U}}}{cvpm}}},
	dimension={\glsd{ui}},
	body={\ensuremath{\mathbb{R}}},
	description={Approximated Set containing \gls{ui} which maximize success probability for one step}
}
\newglossaryentry{Wi}{type=symbols,
	sort={Wi},
	name={\ensuremath{\bm{\rv{W}}_k}},
	dimension={\glsd{yi}},
	body={\ensuremath{\mathbb{R}}},
	description={I.i.d. random variable $\forall i$ representing the uncertainty in the obstacle output}
}
\newglossaryentry{Wii}{type=symbols,
	sort={Wi},
	name={\ensuremath{\bm{\rv{W}}_i}},
	dimension={\glsd{yi}},
	body={\ensuremath{\mathbb{R}}},
	description={I.i.d. random variable $\forall i$ representing the uncertainty in the obstacle output}
}
\newglossaryentry{Xri}{type=symbols,
	sort={Xri},
	name={\ensuremath{\vark{\bm{\rv{X}}}{r}}},
	dimension={\glsd{xri}},
	body={\ensuremath{\mathbb{R}}},
	description={Random variable representing uncertain obstacle state at time $i$}
} 
\newglossaryentry{xri}{type=symbols,
	sort={xri},
	name={\ensuremath{\vark{\bm{x}}{r}}},
	dimension={r},
	body={\ensuremath{\mathbb{R}}},
	description={Realization of \gls{Xri}}
}
\newglossaryentry{XXri}{type=symbols,
	sort={XXri},
	name={\ensuremath{\vark{\set{X}}{r}}},
	dimension={\glsd{xri}},
	body={\ensuremath{\mathbb{R}}},
	description={Set of possible realizations \gls{xri} of the random variable \gls{Xri}}
}
\newglossaryentry{XX}{type=symbols,
	sort={XX},
	name={\ensuremath{\set{X}}},
	dimension={n},
	body={\ensuremath{\mathbb{R}}},
	description={Set of admissable \gls{xi} for each step}
}
\newglossaryentry{XXtilde}{type=symbols,
	sort={XX},
	name={\ensuremath{\tilde{\gls{XX}}}},
	dimension={n},
	body={\ensuremath{\mathbb{R}}},
	description={Set of admissable \gls{xi} for each step}
}
\newglossaryentry{xi}{type=symbols,
	sort={xi},
	name={\ensuremath{\bm{x}_0}},
	dimension={n},
	body={\ensuremath{\mathbb{R}}},
	description={State vector to dynamic linear system}
}
\newglossaryentry{yRi}{type=symbols,
	sort={yRi},
	name={\ensuremath{\vark{\bm{y}}{R}}},
	dimension={},
	body={\ensuremath{\mathbb{R}}},
	description={Reference point for the controlled object at time step $i$}
}
\newglossaryentry{yi}{type=symbols,
	sort={yi},
	name={\ensuremath{\bm{y}_k}},
	dimension={q},
	body={\ensuremath{\mathbb{R}}},
	description={Output vector of dynamic linear system}
}
\newglossaryentry{yip1}{type=symbols,
	sort={yi},
	name={\ensuremath{\bm{y}_{k+1}}},
	dimension={q},
	body={\ensuremath{\mathbb{R}}},
	description={Output vector of dynamic linear system}
}
\newglossaryentry{yj}{type=symbols,
	sort={yi},
	name={\ensuremath{\bm{y}_j}},
	dimension={q},
	body={\ensuremath{\mathbb{R}}},
	description={Output vector of dynamic linear system}
}
\newglossaryentry{yjp1}{type=symbols,
	sort={yi},
	name={\ensuremath{\bm{y}_{j+1}}},
	dimension={q},
	body={\ensuremath{\mathbb{R}}},
	description={Output vector of dynamic linear system}
}
\newglossaryentry{yi1}{type=symbols,
	sort={yi},
	name={\ensuremath{\bm{y}_1}},
	dimension={q},
	body={\ensuremath{\mathbb{R}}},
	description={Output vector of dynamic linear system}
}
\newglossaryentry{YYi}{type=symbols,
	sort={YYi},
	name={\ensuremath{\set{Y}_k}},
	dimension={\glsd{yi}},
	body={\ensuremath{\mathbb{R}}},
	description={Output set/space of admissible \gls{yi}}
}
\newglossaryentry{yri}{type=symbols,
	sort={yri},
	name={\ensuremath{\vark{\bm{y}}{r}}},
	dimension={\glsd{yi}},
	body={\ensuremath{\mathbb{R}}},
	description={Realization of \gls{Yri}}
}
\newglossaryentry{yrib}{type=symbols,
	sort={yri},
	name={\ensuremath{\overline{\bm{y}}_{\text{r},k}}},
	dimension={\glsd{yi}},
	body={\ensuremath{\mathbb{R}}},
	description={Realization of \gls{Yri}}
}
\newglossaryentry{uri}{type=symbols,
	sort={yri},
	name={\ensuremath{\vark{\bm{u}}{r}}},
	dimension={\glsd{yi}},
	body={\ensuremath{\mathbb{R}}},
	description={Realization of \gls{Yri}}
}
\newglossaryentry{yrj}{type=symbols,
	sort={yri},
	name={\ensuremath{\varj{\bm{y}}{r}}},
	dimension={\glsd{yi}},
	body={\ensuremath{\mathbb{R}}},
	description={Realization of \gls{Yri}}
}
\newglossaryentry{yri0}{type=symbols,
	sort={yri},
	name={\ensuremath{\varz{\bm{y}}{r}}},
	dimension={\glsd{yi}},
	body={\ensuremath{\mathbb{R}}},
	description={Realization of \gls{Yri}}
}
\newglossaryentry{yrip1}{type=symbols,
	sort={yri},
	name={\ensuremath{\bm{y}_{\text{r},k+1}}},
	dimension={\glsd{yi}},
	body={\ensuremath{\mathbb{R}}},
	description={Realization of \gls{Yri}}
}
\newglossaryentry{byrip1}{type=symbols,
	sort={yri},
	name={\ensuremath{\overline{\bm{y}}_{\text{r},k+1}}},
	dimension={\glsd{yi}},
	body={\ensuremath{\mathbb{R}}},
	description={Realization of \gls{Yri}}
}
\newglossaryentry{byri}{type=symbols,
	sort={yri},
	name={\ensuremath{\overline{\bm{y}}_{\text{r},k}}},
	dimension={\glsd{yi}},
	body={\ensuremath{\mathbb{R}}},
	description={Realization of \gls{Yri}}
}
\newglossaryentry{yr}{type=symbols,
	sort={yr},
	name={\ensuremath{\var{\overline{\bm{y}}}{r}}},
	dimension={\glsd{yi}},
	body={\ensuremath{\mathbb{R}}},
	description={Known initial obstacle output, i.e. $\var{\set{Y}}{r}^{(0)}=\var{\bar{\bm{y}}}{r}$}
}
\newglossaryentry{Yri}{type=symbols,
	sort={Yri},
	name={\ensuremath{\vark{\bm{\rv{Y}}}{r}}},
	dimension={\glsd{yi}},
	body={\ensuremath{\mathbb{R}}},
	description={Random variable representing the uncertain obstacle output at time $i$}
}

\newglossaryentry{UUcvpml}{type=symbols,
	sort={},
	name={\ensuremath{\mathcal{U}_{\text{cvpm},0:\gls{n}-1}}},
	dimension={\glsd{yi}},
	body={\ensuremath{\mathbb{R}}},
	description={Random variable representing the uncertain obstacle output at time $i$}
}
\newglossaryentry{apUUcvpml}{type=symbols,
	sort={},
	name={\ensuremath{\hat{\mathcal{U}}_{\text{cvpm},0:\gls{n}-1}}},
	dimension={\glsd{yi}},
	body={\ensuremath{\mathbb{R}}},
	description={Random variable representing the uncertain obstacle output at time $i$}
}
\newglossaryentry{Ucvpml}{type=symbols,
	sort={},
	name={\ensuremath{\bm{U}_{0:\gls{n}-1}}},
	dimension={\glsd{yi}},
	body={\ensuremath{\mathbb{R}}},
	description={Random variable representing the uncertain obstacle output at time $i$}
}
\newglossaryentry{UUxci}{type=symbols, % new %% U satisfying constraints
	sort={UUX},
	name={\ensuremath{\mathbb{U}_{\bm{x},i}}},
	dimension={\glsd{ui}\cdot \gls{Nmpc}},
	body={\ensuremath{\mathbb{R}}},
	description={Set containing all \gls{uNmpc} which fulfill the state constraint \gls{XX} at each time instant}
}
\newglossaryentry{Mlm1}{type=symbols,
	sort={M},
	name={\ensuremath{\norm{\bm{y}_{\gls{n}}-\bm{y}_{\text{r},0}}_2}},
	dimension={},
	body={\ensuremath{\mathbb{R}^+}},
	description={Measure}
}
\newglossaryentry{Mlm1b}{type=symbols,
	sort={M},
	name={\ensuremath{\norm{\bm{y}_{\gls{n}}-\overline{\bm{y}}_{\text{r},\gls{n}}}_2}},
	dimension={},
	body={\ensuremath{\mathbb{R}^+}},
	description={Measure}
}
\newglossaryentry{Pcvj}{type=symbols,
	sort={pcon},
	name={\ensuremath{p_{\text{cv},j}}},
	dimension={},
	body={\ensuremath{\mathbb{R}^+}},
	description={Conflict probability}
}
\newglossaryentry{Pcvjp1}{type=symbols,
	sort={pcon},
	name={\ensuremath{p_{\text{cv},j+1}}},
	dimension={},
	body={\ensuremath{\mathbb{R}^+}},
	description={Conflict probability}
}
\newglossaryentry{Pcvjp1att}{type=symbols,
	sort={pcon},
	name={\ensuremath{p_{\text{cv},j+1}\gls{attj}}},
	dimension={},
	body={\ensuremath{\mathbb{R}^+}},
	description={Conflict probability}
}
\newglossaryentry{jl}{type=symbols,
	sort={jl},
	name={\ensuremath{j \in \mathbb{I}_{0:\gls{n}-1}}},
	dimension={},
	body={\ensuremath{\mathbb{R}^+}},
	description={Conflict probability}
}
\newglossaryentry{il}{type=symbols,
	sort={jl},
	name={\ensuremath{i \in \mathbb{I}_{0:\gls{n}-1}}},
	dimension={},
	body={\ensuremath{\mathbb{R}^+}},
	description={Conflict probability}
}
\newglossaryentry{uii}{type=symbols,
	sort={ui},
	name={\ensuremath{\bm{u}_i}},
	dimension={m},
	body={\ensuremath{\mathbb{R}}},
	description={Input vector to dynamic linear system}
}
\newglossaryentry{Pcvl}{type=symbols,
	sort={pcon},
	name={\ensuremath{p_{\text{cv},l}}},
	dimension={},
	body={\ensuremath{\mathbb{R}^+}},
	description={Conflict probability}
}

\newglossaryentry{vector}{type=notation,
	sort={vector},
	name={\ensuremath{\bm{z}_k}},
	description={Vector at timestep $i$}
}
\newglossaryentry{seqi}{type=notation,
	sort={seqi},
	name={\ensuremath{(\cdot)_{(1:k)}}},
	description={Combined variable up to $i$-th step if $(\cdot)^{(i)}$ exists; e.g.: $\bm{z}^{(1:i)}=\begin{bmatrix}\bm{z}^{(1)} & \cdots & \bm{z}^{(i)}\end{bmatrix}^\T$}
}	
\newglossaryentry{seqN}{type=notation,
	sort={seqN},
	name={\ensuremath{(\cdot)}},
	description={Combined variable for \gls{N}/\gls{Nmpc} steps if $(\cdot)^{(i)}$ exists; e.g.: $\bm{z}=\begin{bmatrix}\bm{z}^{(1)} & \cdots & \bm{z}^{(N)}\end{bmatrix}^\T$}
}	
\newglossaryentry{seqNi}{type=notation,
	sort={seqNi},
	name={\ensuremath{(\cdot)_k}},
	description={Combined variable for \gls{Nmpc} steps if $(\cdot)^{(i)}$ exists; e.g.: $\bm{z}_i=\begin{bmatrix}\bm{z}_i^{(i+1)} & \cdots & \bm{z}_i^{(i+N)}\end{bmatrix}^\T$}
}	
\newglossaryentry{matrix}{type=notation,
	sort={matrix},
	name={\ensuremath{Z}},
	description={Matrix}
}		
\newglossaryentry{rv}{type=notation,
	sort={rv},
	name={\ensuremath{\bm{\rv{Z}}_k}},
	description={Random vector at timestep $i$}
}	
\newglossaryentry{set}{type=notation,
	sort={set},
	name={\ensuremath{\set{Z}_k}},
	description={Set of admissable \gls{vector}}
}
\newglossaryentry{I}{type=notation,
	sort={I},
	name={\ensuremath{I_{k\times k}}},
	description={$k$-by-$k$ identity matrix}
}

\title{\LARGE \bf
Minimization of Constraint Violation Probability\\ in Model Predictive Control
}

\author{Tim~Br\"udigam,~Victor~Ga{\ss}mann,~Dirk~Wollherr,~and~Marion~Leibold% <-this % stops a space
\thanks{T. Br\"udigam, D. Wollherr, M. Leibold are with the Chair of Automatic Control Engineering at the Technical University of Munich, Germany.
{\tt\small \{tim.bruedigam; dw; marion.leibold\}@tum.de}} 
\thanks{V. Ga{\ss}mann is with the Chair of Robotics, Artificial Intelligence and Real-time Systems at the Technical University of Munich, Germany.
{\tt\small victor.gassmann@tum.de}} 
}

\begin{document}

\maketitle
\thispagestyle{empty}
\pagestyle{empty}

\begin{abstract}
While Robust Model Predictive Control considers the worst-case system uncertainty, Stochastic Model Predictive Control, using chance constraints, provides less conservative solutions by allowing a certain constraint violation probability depending on a predefined risk parameter. However, for safety-critical systems it is not only important to bound the constraint violation probability but to reduce this probability as much as possible. Therefore, an approach is necessary that minimizes the constraint violation probability while ensuring that the Model Predictive Control optimization problem remains feasible.
We propose a novel Model Predictive Control scheme that yields a solution with minimal constraint violation probability for a norm constraint in an environment with uncertainty. After minimal constraint violation is guaranteed the solution is then also optimized with respect to other control objectives. Further, it is possible to account for changes over time of the support of the uncertainty. We first present a general method and then provide an approach for uncertainties with symmetric, unimodal probability density function. Recursive feasibility and convergence of the method are proved. 
A simulation example demonstrates the effectiveness of the proposed method.
\end{abstract}

\section{Introduction}
\label{sec:introduction}

\vspace{-9.4cm}
\mbox{\small This~is~the~pre-peer~reviewed~version~of~the~work~published~by~the~International~Journal~of~Robust~and~Nonlinear~Control.}
\vspace{9.0cm}

\vspace{-9.0cm}
\mbox{\small The~published~version~is~available~at~https://doi.org/10.1002/rnc.5636.}
\vspace{8.6cm}

Autonomous systems in safety-critical applications, such as autonomous driving or human-robot interaction, depend on controllers that are able to safely and efficiently cope with uncertainties. In these applications autonomous vehicles and robots must avoid collisions to ensure safety, while it is also of interest to optimize other objectives for efficiency, e.g., energy consumption. \gls{MPC} is a promising method for this problem setup as it is an online optimization that has the ability to cope with hard constraints.

Classic \gls{MPC} methods deal well with deterministic systems and provide guarantees for stability as well as recursive feasibility \cite{MayneEtalScokaert2000, PrimbsNevistic2000,Gruene2012}, where recursive feasibility ensures that the \gls{MPC} optimization problem remains feasible at future time steps if it is initially feasible. 

More advanced \gls{MPC} algorithms are necessary in the presence of uncertainty in the system. \gls{RMPC} methods provide control laws that satisfy the control objectives and constraints by accounting for the worst-case realization of the uncertainty, assuming that the bound, i.e., the support, of the probability distribution for the uncertainty is known a priori \cite{MayneSeronRakovic2005,MagniEtalAllgoewer2003}. Among the many approaches to ensure robustness, summarized by Mayne \cite{Mayne2014}, tube-based \gls{MPC} \cite{LangsonEtalMayne2004,MayneEtalFalugi2011,NikouDimarogonas2019} has received much attention. In tube-based \gls{MPC} it is required that the closed-loop trajectories remain within a tube to guarantee constraint satisfaction for all possible uncertainty realizations. A major drawback of \gls{RMPC} is its conservative control law due to accounting for the worst-case uncertainty realization. This can be problematic in applications with high levels of uncertainty, e.g., autonomous driving in dense traffic. 

This issue is addressed by \gls{SMPC} methods, which exploit knowledge of the uncertainty by introducing probabilistic constraints and potentially applying an expectation value based objective function. \gls{SMPC} methods, instead of considering the worst possible uncertainty realization as in \gls{RMPC}, introduce a risk parameter that specifies how likely a constraint violation may be. This probabilistic constraint is referred to as a chance constraint. In many applications it is acceptable to allow a small probability of constraint violation. This results in a positive effect on performance, as the control law is no longer required to account for unlikely uncertainty realizations. Extensive summaries of diverse \gls{SMPC} methods are given by Mesbah \cite{Mesbah2016} and specifically for linear \gls{SMPC} by Farina et al. \cite{FarinaGiulioniScattolini2016}. Stability of \gls{SMPC} without a terminal constraint is addressed by Lorenzen et al. \cite{LorenzenMuellerAllgoewer2019} while a performance analysis comparing \gls{MPC} and \gls{SMPC} is presented by Seron et al. \cite{SeronGoodwinCarrasco2018}.

In general, it is difficult to handle probabilistic constraints directly, requiring a transformation into a tractable deterministic representation of the chance constraint. If the uncertainty has a Gaussian distribution, analytic expressions can be determined \cite{SchwarmNikolaou1999}. These methods, however, are not applicable if the uncertainty does not have a Gaussian distribution or if it is unknown. Among various other methods used to cope with probabilistic constraints are the particle and the scenario approach. The particle \gls{SMPC} method \cite{BlackmoreEtalWilliams2010} is able to handle arbitrary probability distributions by sampling a finite number of particles to approximate the uncertainty, allowing only a certain percentage of the sampled particles to violate the chance constraint, depending on the risk parameter. In comparison, in \gls{SCMPC} \cite{SchildbachEtalMorari2014} a required number of samples, called scenarios, is obtained given a specified risk parameter, using the scenario approach \cite{CalafioreCampi2006,CampiGaratti2011}. Then, the chance constraint must be satisfied for all drawn scenarios. A low risk parameter yields a large number of necessary scenarios to be considered.

However, the benefits of using chance constraints come with the disadvantage of constraint violations if unlikely uncertainty realizations occur. This also leads to the problem of ensuring recursive feasibility, i.e., guaranteeing that the \gls{MPC} optimization problem remains solvable at every step. In case of an uncertainty realization with low probability, there might not be an admissible control input that can satisfy the chance constraints. Recursive feasibility also becomes an issue if the maximal uncertainty value is not constant, i.e., the support of the uncertainty probability density function changes over time. 

Recursive feasibility of \gls{SMPC} for bounded disturbances is addressed by Korda et al. \cite{KordaEtalOldewurtel2011}. A further approach addressing recursive feasibility in \gls{SMPC} are stochastic tube methods \cite{KouvaritakisEtalCheng2010,CannonEtalCheng2011} using constraint tightening. Lorenzen et al. \cite{LorenzenEtalAllgoewer2017} suggested an approach that combines the works of Korda et al. \cite{KordaEtalOldewurtel2011} and Kouvaritakis et al. \cite{KouvaritakisEtalCheng2010} where a tuning parameter is introduced that allows for shifting priority between performance and an increased feasible region. Recursive feasibility in \gls{SMPC} for probabilistically constrained Markovian jump linear systems is addressed by Lu et al. \cite{LuXiLi2019}.

Due to its ability to efficiently cope with environments subject to uncertainty, \gls{SMPC} has become increasingly popular in applications such as process control \cite{SchwarmNikolaou1999,JuradoEtalRubio2015}, energy control \cite{OldewurtelEtalMorari2014} and power systems \cite{KumarEtalZavala2018,JiangEtalDong2019}, finance \cite{GrafPlessenEtalBemporad2019}, general automotive applications \cite{RipaccioliEtalKolmanovsky2010}, as well as more specifically safety-critical applications, e.g., path planning \cite{BlackmoreEtalWilliams2010} and autonomous driving \cite{CarvalhoEtalBorrelli2014,SchildbachBorrelli2015,LenzEtalKnoll2015,CesariEtalBorrelli2017,BruedigamEtalWollherr2018b,SuhChaeYi2018}. However, the possible constraint violation and the resulting infeasibility of the optimization problem are limiting factors when designing an efficient \gls{SMPC} algorithm in practice, especially in safety-critical applications. 

A further drawback of chance constraints in \gls{SMPC} appears if the optimal solution is `on the chance constraint' even though other solutions are possible with no or only minimal effect on the cost function. In other words, a solution of the \gls{SMPC} optimization problem minimizes the cost function and satisfies the required probability for the chance constraint. There might be other solutions with low cost that have a chance constraint violation probability less than required by the risk parameter or even zero. However, the \gls{SMPC} optimization problem is solved once a solution is found with minimal cost and that satisfies the chance constraint, given the risk parameter. This means that the solution with a lower constraint violation probability is not found. Additionally, choosing a suitable risk parameter is challenging, as high values increase risk while low values reduce efficiency.

These issues are especially relevant in safety-critical systems. One example is an autonomous vehicle that plans to avoid collisions in an uncertain environment, e.g., a car avoiding a bicycle with uncertain behavior. If the support of the uncertainty is not known a priori, \gls{RMPC} algorithms are either not applicable or require that the vehicle does not move until all surrounding vehicles are distanced enough. This, however, is not practical. Therefore, the collision constraint, realized with a norm constraint, could be transformed into a chance constraint in an \gls{SMPC} approach, allowing a small collision probability. While this yields a more efficient solution than \gls{RMPC}, it might result in a collision. However, an autonomous vehicle must always choose the safest, sensible trajectory, even if it comes at the cost of increased energy or longer travel time. Further, if the chance constraint in \gls{SMPC} cannot be satisfied anymore because an unlikely scenario occurred or the uncertainty support changed, the optimization problem becomes infeasible. Alternative control laws, e.g., full breaking, and recovery strategies can then be used to regain a feasible controller. However, in such scenarios where the chance constraint cannot be satisfied, the controller ideally yields the safest solution possible, which is not guaranteed with standard recovery strategies. In the example of the autonomous vehicle this is the solution with the lowest collision probability.

In this paper we propose a novel \gls{MPC} strategy for linear, discrete-time systems which not only satisfies general hard constraints over the entire prediction horizon, but additionally minimizes the probability of violating a norm constraint in the next predicted step, while also optimizing for other control objectives. This is achieved by first calculating a set that constrains the system inputs such that only those inputs are allowed which minimize the constraint violation probability. This is then followed by an optimization problem which optimizes further required objectives such as fast reference tracking or energy consumption. In this subsequent optimization problem only those inputs are admissible which minimize the norm constraint violation probability. The proposed method can handle time-varying bounds for the support of the system uncertainty and considers hard constraints on the state and input for the entire prediction horizon, e.g., due to actuator limitations. 

We will first present the general method to minimize constraint violation probability. For the general method it can be difficult to determine a tightened set of admissible inputs which guarantee minimal constraint violation probability. Therefore, we suggest an approach which allows the computation of this tightened input set, given uncertainties with symmetric, unimodal probability density function, i.e., the relative likelihood of uncertainty realizations decreases with increased distance to the mean. This tractable approach yields a convex set of inputs which minimize the constraint violation probability. Guarantees are provided for recursive feasibility and convergence of the proposed \gls{MPC} algorithm. In the following we will refer to the proposed method as CVPM-\gls{MPC}, i.e., \gls{MPC} with \gls{CVPM}. A simulation for a vehicle collision avoidance scenario is shown to display the effectiveness of the proposed method and highlight its advantages compared to \gls{SMPC} and \gls{RMPC}. 

In summary, the contribution is as follows:
\begin{itemize}
\item Proposition of a general CVPM-\gls{MPC} method to minimize the constraint violation probability for the next predicted step, while satisfying state and input constraints and optimizing further objectives
\item Derivation of a CVPM-\gls{MPC} approach for uncertainties with symmetric, unimodal probability density function 
\item Guarantee of recursive feasibility and convergence 
\end{itemize}

The proposed CVPM-\gls{MPC} method can be beneficial to multiple applications, especially to safety-critical applications such as autonomous driving or human-robot interaction where the risk measure regarding collision is norm-based \cite{CarvalhoEtalBorrelli2014,SchulmanEtalAbbeel2014,ZanchettinEtalMatthias2016}. In these safety-critical applications there is a clear priority on maximizing safety, i.e., the constraint violation probability of safety constraints needs to be minimal, before optimizing other objectives, e.g., energy consumption. 

While in general it is possible to minimize the constraint violation probability not only for the first step but for multiple steps, this significantly increases conservatism, resulting in solutions which are more similar to \gls{RMPC} solutions. Minimizing the first step probability iteratively yields the advantage of safer solutions than \gls{SMPC} and less conservatism compared to \gls{RMPC}. We therefore focus on iteratively minimizing the constraint violation probability for the next step, i.e., the first predicted \gls{MPC} step, however, a solution approach for a multi-step CVPM-\gls{MPC} method is also provided. \\

The remainder of the paper is structured as follows. Section~\ref{sec:problem} introduces the system to be considered, the uncertainty, and the control objective. The proposed method is introduced in Section \ref{sec:method}, first focusing on minimizing the constraint violation probability, then introducing the resulting \gls{MPC} algorithm. Section \ref{sec:properties} analyzes the properties of the suggested method, guarantees on recursive feasibility and convergence, while the CVPM-\gls{MPC} method is discussed in Section \ref{sec:discussion}. An example of the applied method is given in Section \ref{sec:results}, simulating a vehicle collision avoidance scenario. Section \ref{sec:conclusion} provides conclusive remarks.\\

\textit{Notation:} Regular letters indicate scalars, bold lowercase letters denote vectors, and bold uppercase letters are used for matrices, e.g., $a$, $\bm{a}$, $\bm{A}$, respectively. Random variables are represented by bold uppercase letters. The Euclidean norm is denoted by~$\norm{.}_2$. The probability of an event is given by $\mathbb{P}(.)$. A probability distribution is denoted by $p$ and is described by the probability density function $f$ if a probability density function exists. The probability distribution and density function have support $\supp(p)$ and $\supp(f)$, respectively. Step $k$ of a state or parameter is represented by a subscript $k$, e.g., $\bm{x}_k$ for state $\bm{x}$. The integers in the interval between $a$ and $b$, including the boundaries, are denoted by $\mathbb{I}_{a:b}$. 

\section{Problem Setup}
\label{sec:problem}
In this section we define the system class and the general \gls{MPC} algorithm. Additionally, a probabilistic norm constraint is introduced.
\subsection{System Dynamics and Control Objective}

Consider the controlled linear, time-invariant, discrete-time system
\begin{IEEEeqnarray}{rl}
\IEEEyesnumber \label{eq:system}
\bm{x}_{k+1} &= \bm{A}\gls{xk}+\bm{B} \gls{ui}, \IEEEyessubnumber \label{eq:systemAB}\\
\bm{y}_{k} &= \bm{C}\gls{xk} \IEEEyessubnumber \label{eq:systemC}
\end{IEEEeqnarray}
with time step $k$, states $\gls{xk} \in \mathbb{R}^n$, control input $\gls{ui} \in \mathbb{R}^m$, output $\gls{yi} \in \mathbb{R}^q$, and matrices $\bm{A} \in \mathbb{R}^{n\times n}$, $\bm{B} \in \mathbb{R}^{n\times m}$, $\bm{C}~\in~\mathbb{R}^{q\times n}$.

Furthermore, we consider an uncertain system
\begin{IEEEeqnarray}{rcl}
\IEEEyesnumber \label{eq:unc_system}
\gls{yrip1} =& \gls{yri} +\gls{uri} &+ \gls{wi} \IEEEyessubnumber \label{eq:unc_system1} \\
=& \gls{byrip1} &+ \gls{wi} \IEEEyessubnumber \label{eq:unc_system2}
\end{IEEEeqnarray}
depending on the output $\gls{yri}\in \mathbb{R}^q$ at step $k$, a deterministic, known input $\gls{uri}\in \mathbb{R}^q$, and a stochastic part $\gls{wi}\in \mathbb{R}^q$ which is the realization of a random variable $\gls{Wi}$. The nominal prediction of \gls{yrip1} is indicated by $\gls{byrip1}=\gls{yri} +\gls{uri} $, consisting of the previous output \gls{yri} and the deterministic, known input \gls{uri}.

\begin{assumption}[Uncertainty Properties]
\label{ass:uncprop}
The random variables $\gls{Wi}\lr{\gls{wi}} \sim \gls{fPw}$ with the probability distribution $\gls{Pw}$ and density function $\gls{fPw}$ have zero mean and are truncated with the initially known, convex and bounded support $\supp\lr{\gls{fPw}}$.
\end{assumption}

The support of $\gls{fPw}$ is given by 
\begin{IEEEeqnarray}{c}
\supp\lr{\gls{fPw}}= \setdef[\bm{w}_k]{\norm{\bm{w}_k}_2 \leq \gls{wimax} }
\label{eq:supp_pw}
\end{IEEEeqnarray}
where $\gls{wimax} \in \mathbb{R}_{\geq0}$.

The controller for \eqref{eq:system} is designed to (approximately) optimize an infinite horizon objective function, while accounting for input and state constraints. For an initial state $\bm{x}_k$ at time step $k$ the objective function to be minimized is
 \begin{IEEEeqnarray}{c}
V_\infty = \sum_{j=0}^\infty \lr{ \bm{x}_{k+j}^\top \bm{Q} \bm{x}_{k+j} + \bm{u}_{k+j}^\top \bm{R} \bm{u}_{k+j} }\label{eq:infcost}
\end{IEEEeqnarray}
with $\bm{Q} \in \mathbb{R}^{n \times n}$, $\bm{R} \in \mathbb{R}^{m \times m}$ and $\bm{Q} \succeq 0$, $\bm{R} \succ 0$.

In the following the index $k$ represents regular time steps while the index $j$ indicates prediction steps within an \gls{MPC} optimization problem, similar to \eqref{eq:infcost}.

\subsection{Model Predictive Control}
We consider an \gls{MPC} algorithm to solve the control problem~\eqref{eq:infcost} with a finite horizon objective function $V_N$. \gls{MPC} repeatedly solves an optimization problem on a finite horizon. After the optimization only the first optimized control input is applied. Then the horizon is shifted by one step and a new optimization is performed. Without loss of generality it is assumed that an \gls{MPC} iteration starts with $\gls{xi}$. The finite horizon cost is then given by
 \begin{IEEEeqnarray}{c}
V_N\lr{\gls{xi},\gls{uNmpc}}= \sum_{j=0}^{\gls{Nmpc}-1} l\lr{\bm{x}_{j}, \bm{u}_{j}} + \gls{Vf}   \label{eq:fincost}
\end{IEEEeqnarray}
with the \gls{MPC} horizon $\gls{Nmpc}$, input sequence $\gls{uNmpc}~=~\lr{\bm{u}_0, \bm{u}_{1}, ..., \bm{u}_{N-1}}$, continuous stage cost $l\lr{\bm{x}_{j}, \bm{u}_{j}}~=~\bm{x}_{j}^\top \bm{Q} \bm{x}_{j} + \bm{u}_{j}^\top \bm{R} \bm{u}_{j}$ with $l\lr{\bm{0},\bm{0}}=0$, and continuous terminal cost $\gls{Vf}$ with $V_\text{f}\lr{\bm{0}}=0$.

We first formulate the \gls{MPC} optimization problem with input constraints \gls{UUj} and state constraints \gls{XX} that are independent of the uncertain system \eqref{eq:unc_system}, resulting in
\begin{IEEEeqnarray}{rll}
\IEEEyesnumber \label{eq:mpc_gen}
V_N^* = \underset{\gls{uNmpc}}{\min}~ &V_N\lr{\gls{xi},\gls{uNmpc}}, \IEEEyessubnumber  \label{eq:cf_gen}\\
\text{s.t. }& \bm{x}_{j+1} = \bm{A}\bm{x}_j+\bm{B} \bm{u}_j,~~&j \in \mathbb{I}_{0:N-1} \IEEEyessubnumber \label{eq:fcon}\\
&\gls{uj}\in \gls{UUj},~~&j \in \mathbb{I}_{0:N-1} \IEEEyessubnumber  \label{eq:ucon}\\
& \bm{x}_{j+1}\in \gls{XX},~~&j \in \mathbb{I}_{0:N-1} \IEEEyessubnumber \label{eq:xcon}\\
& \bm{x}_N \in \gls{XXf}. \IEEEyessubnumber \label{eq:xconf}
\end{IEEEeqnarray}
The input $\gls{uj}$ is bounded by the non-empty input value space $\gls{UUj} \subseteq \mathbb{R}^m$, i.e., the input constraint \eqref{eq:ucon}. The convex state constraint and terminal constraint set are given by $\gls{XX}$ and $\gls{XXf}$, respectively.  

\begin{assumption}[Control Invariant Terminal Set] \label{ass:XXf}
For all $\gls{xj} \in \gls{XXf}$, there exists an admissible $\gls{uj}$ 
such that $\bm{x}_{j+1} \in \gls{XXf}$. 
\end{assumption}

\begin{assumption}[Lyapunov Functions]
\label{ass:termcost}
The cost $V_N\lr{\gls{xi},\gls{uNmpc}}$ and the terminal cost $V_{\text{f}}\lr{\bm{x}_{k}}$ are Lyapunov functions in \gls{XX} and \gls{XXf}, respectively.
\end{assumption}

We denote with $\gls{UUxcj}$ the set of admissible inputs $\gls{uj}$ such that all constraints of \eqref{eq:mpc_gen} are satisfied for $j$, i.e.,
\begin{IEEEeqnarray}{c}
\gls{UUxcj}=\setdef[\gls{uj}]{\lr{\eqref{eq:ucon}, \eqref{eq:xcon}} \land \lr{\eqref{eq:xconf}~\text{if}~j=N-1}}. \label{eq:UUxc}  \IEEEeqnarraynumspace
\end{IEEEeqnarray}

\begin{remark}
\label{rem:ref}
Instead of steering $\gls{xk}$ to the origin as in \eqref{eq:fincost}, specific references $\bm{x}_{\text{ref},k}$ can also be tracked.
\end{remark}

\subsection{Model Predictive Control with Norm Constraint}

In the following the uncertain system \eqref{eq:unc_system} is considered. 

\begin{assumption}[Initially known Uncertainty]
\label{ass:initunc}
The initial state $\bm{y}_{\text{r},0}$ and deterministic input $\bm{u}_{\text{r},0}$ are known at the beginning of each \gls{MPC} optimization.
\end{assumption}

Here, we consider an additional constraint for the \gls{MPC} problem \eqref{eq:mpc_gen}, which is the norm constraint
\begin{IEEEeqnarray}{c}
\gls{Mk} \geq \gls{Mmink} \label{eq:hc2}
\end{IEEEeqnarray}
representing a constraint on the $2$-norm $\gls{Mk}$, e.g., the distance between two points must not be smaller than a minimal value \gls{Mmink}. While \eqref{eq:hc2} is a hard constraint, we will first transform \eqref{eq:hc2} into a chance constraint and later, in Section~\ref{sec:method}, we will minimize the probability that this norm constraint is violated. 

\begin{remark}
It is also possible to consider a $p$-norm constraint with $\norm{\gls{yi}-\gls{yri}}_p$ instead of the $2$-norm. Similar to the $2$-norm, all $p$-norms are convex. Without loss of generality we will consider the $2$-norm as most applications require a $2$-norm to represent the Euclidean distance. 
\end{remark}

As $\gls{yri}$ is subject to uncertainty, the norm constraint \eqref{eq:hc2} is difficult, potentially impossible, to fulfill, or it might lead to overly conservative control inputs. The hard norm constraint \eqref{eq:hc2} can be relaxed if substituted by the chance constraint
 \begin{IEEEeqnarray}{c}
\prob{\gls{Mk} < \gls{Mmink}} \leq \beta_k \label{eq:cc}
\end{IEEEeqnarray}
with
\begin{IEEEeqnarray}{c}
\gls{Pcv}\lr{\bm{u}_{k-1}} \coloneqq \prob{\gls{Mk} < \gls{Mmink}} \label{eq:cc_def}
\end{IEEEeqnarray}
where $\beta_k$ is a risk parameter and \gls{Pcv} denotes the constraint violation probability for the norm constraint \eqref{eq:hc2}. The constraint violation probability \gls{Pcv} for step $k$ is evaluated at step $k~-~1$, i.e. at the previous step. Therefore, the probability \gls{Pcv} depends on the input $\bm{u}_{k-1}$, yielding \gls{yi} according to \eqref{eq:system}. In the following, the dependence of \gls{Pcv} on $\bm{u}_{k-1}$ is omitted if it reduces notation complexity.

The following example will illustrate the idea of the chance constraint. We consider a controlled object with position $\gls{yi}$ and a dynamic obstacle with position $\bm{y}_{\text{r},k}$ where $\norm{\gls{yi} -\bm{y}_{\text{r},k}}_2$ is the distance between both objects. The objects collide if $\gls{Mk}~<~\gls{Mmink}$. An interpretation for \eqref{eq:cc} is that \gls{Pcv} represents the probability of a collision and this constraint violation probability is bounded by a predefined risk parameter~$\beta_k$. A similar example is analyzed in a simulation in Section~\ref{sec:results}.

The bounded support of \gls{Pcv} is given by
\begin{IEEEeqnarray}{rl}
&\supp\lr{\gls{Pcv}} = \setdef[\bm{u}_{k-1}]{\gls{Pcv} > 0 },\IEEEeqnarraynumspace
\label{eq:supp_pcv}
\end{IEEEeqnarray}
resulting in $\gls{Pcv} = 0$ if the maximal uncertainty value $w_{\text{max},k-1}$ cannot cause $\gls{Mk} < \gls{Mmink}$.

While it is possible to consider the norm constraint \eqref{eq:hc2} over multiple steps, we will only consider the norm constraint for the next predicted step $j=1$ with a horizon $N\geq 1$. Applying \eqref{eq:hc2} over the entire horizon $N$ results in a conservative control law similar to \gls{RMPC}.

We reformulate the \gls{MPC} optimization problem \eqref{eq:mpc_gen} such that it includes the norm constraint \eqref{eq:hc2}, resulting in
\begin{IEEEeqnarray}{rll}
\IEEEyesnumber \label{eq:mpc_nc}
V_N^* = \underset{\gls{uNmpc}}{\min}~ &V_N\lr{\gls{xi},\gls{uNmpc}}, \IEEEyessubnumber  \label{eq:cf_nc}\\
\text{s.t. }& \bm{x}_{j+1} = \bm{A}\bm{x}_j+\bm{B} \bm{u}_j,~~&j \in \mathbb{I}_{0:N-1} \IEEEyessubnumber \label{eq:fcon_nc}\\
& \bm{y}_{j} = \bm{C}\bm{x}_j,~~&j \in \mathbb{I}_{0:N} \IEEEyessubnumber \label{eq:fcon2}\\
& \bm{y}_{\text{r},j+1} = \bm{y}_{\text{r},j} +\bm{u}_{\text{r},j} + \bm{w}_{j},~~&j \in \mathbb{I}_{0:N-1} \IEEEyessubnumber \label{eq:uncsys}\\
& \gls{uj} \in \gls{UUxcj},~~&j \in \mathbb{I}_{0:N-1} \IEEEyessubnumber \label{eq:UUxc_nc} \IEEEeqnarraynumspace\\
&\gls{Mkj} \geq \gls{Mminkj},~~&j=1  \IEEEyessubnumber  \label{eq:hc}
\end{IEEEeqnarray}
where \eqref{eq:UUxc_nc} summarizes the constraints of the initial \gls{MPC} problem \eqref{eq:mpc_gen}, according to the definition of \gls{UUxcj} in \eqref{eq:UUxc}.

Substituting the norm constraint \eqref{eq:hc} by the chance constraint \eqref{eq:cc}, i.e.,
 \begin{IEEEeqnarray}{c}
 \prob{\norm{\bm{y}_1 - \bm{y}_{\text{r},1}}_2 < c_1} \leq \beta_1 \label{eq:ccj1}
\end{IEEEeqnarray}
with
 \begin{IEEEeqnarray}{c}
\gls{Pcv1att} \coloneqq \prob{\norm{\bm{y}_1 - \bm{y}_{\text{r},1}}_2 < c_1}
\end{IEEEeqnarray}
yields an \gls{SMPC} optimization problem. 

Only the general \gls{MPC} problem \eqref{eq:mpc_gen} is addressed in Assumptions \ref{ass:XXf} and \ref{ass:termcost}, the norm constraint \eqref{eq:hc} is not considered, as~\eqref{eq:hc} is specifically addressed in the method presented in Sections \ref{sec:method} and \ref{sec:properties}. 

\begin{remark}
\label{rem:cc}
In \eqref{eq:mpc_nc} the norm constraint \eqref{eq:hc2} is only considered in the first step, i.e., at step $j=1$, as we later minimize the probability of constraint violation for the first step. However, if this norm constraint is required to be considered at future steps $j \in \mathbb{I}_{2:N}$, this can be achieved by treating \eqref{eq:hc2} as a chance constraint, similar to \eqref{eq:ccj1}, resulting in
 \begin{IEEEeqnarray}{c}
\prob{\gls{Mj} < c_{j}} \leq \beta_{j},~~j \in \mathbb{I}_{2:N}. \label{eq:ccj} \IEEEeqnarraynumspace
\end{IEEEeqnarray}
This chance constraint \eqref{eq:ccj} is then added to \eqref{eq:mpc_gen} and subsequently needs to be considered in \eqref{eq:UUxc}. Assumptions \ref{ass:XXf} and \ref{ass:termcost} still need to be fulfilled if chance constraints are included for $j \in \mathbb{I}_{2:N}$ in the optimization problem. 
\end{remark}

\subsection{Problem Statement}

Instead of only bounding the chance constraint \eqref{eq:ccj1} by the risk parameter $\beta_1$, we aim at minimizing the constraint violation probability \gls{Pcv1} within the \gls{MPC} optimization problem. The challenge is to solve the \gls{MPC} problem  \eqref{eq:cf_nc} - \eqref{eq:UUxc_nc}, while it needs to be guaranteed that
 \begin{IEEEeqnarray}{c}
\gls{Pcv1} = \min_{\bm{u}_0 \in \mathbb{U}_{\bm{x},0}}  \prob{\norm{\bm{y}_1 - \bm{y}_{\text{r},1}}_2 < c_1}  \label{eq:ps} \IEEEeqnarraynumspace
\end{IEEEeqnarray}
and that the \gls{MPC} problem remains recursively feasible.

Multiple issues arise when implementing chance constraints. As \eqref{eq:cc} is a probabilistic constraint it cannot directly be handled by an optimization solver. The probabilistic chance constraint needs to be transformed into a tractable substitute of the chance constraint.

If chance constraints are used in \gls{SMPC}, two further problems occur. First, recursive feasibility of the \gls{SMPC} optimization problem needs to be ensured. If the \gls{SMPC} optimization problem is solvable at step $k$, it must also be solvable at step $k+1$ to guarantee recursive feasibility. This is a challenge for various \gls{SMPC} methods as the risk parameter $\beta_k$ allows a certain probability of constraint violation, causing infeasibility of the optimization problem for uncertainty realizations with low probability. Additionally, in both \gls{SMPC} and \gls{RMPC} recursive feasibility is not ensured in case of an unexpectedly increasing uncertainty support. Robustness in \gls{RMPC} or a satisfaction of the chance constraint in \gls{SMPC} are typically only ensured for the initially considered uncertainty support.

Second, in safety-critical systems a further aspect reduces the usability of chance constraints in \gls{SMPC}. A solution is valid as long as the probability of violating the safety constraint satisfies the risk parameter. Assuming there exists a solution with lower, or even zero percent, constraint violation probability, the optimization solution will still be `on the chance constraint' if this results in lower objective costs, i.e., allow constraint violations according to the risk parameter.   

We consider again the example in the introduction of a car overtaking a bicycle. Using a chance constraint with $\beta_k > 0$, the car will pass the bicycle but will choose a trajectory around the bicycle that allows a collision with a low probability due to $\beta_k > 0$. Given a finite bicycle uncertainty support, passing the bicycle with slightly more distance yields zero collision probability with only a small increase of cost. However, in practice, this slightly increased cost is acceptable if thereby safety is guaranteed.

In this paper we propose a novel \gls{MPC} approach, CVPM-\gls{MPC}, that first ensures the minimal constraint violation probability~$\gls{Pcv1}$, but then still optimizes the objective function $J_N(\gls{xi},\gls{uNmpc})$. This approach yields a control input resulting in the lowest possible constraint violation probability, given input and state constraints, while still optimizing further objectives. The CVPM-\gls{MPC} method guarantees recursive feasibility, also for a changing uncertainty support, and ensures convergence of the \gls{MPC} algorithm.

\section{Method}
\label{sec:method}

In this section we derive the CVPM-\gls{MPC} method to minimize the constraint violation probability \gls{Pcvj} for the first predicted step $j=1$ in an \gls{MPC} problem. First, a general approach is presented to find a tightened admissible input set that minimizes the first step constraint violation probability. In the following part it is shown how this approach can be incorporated into \gls{MPC}. A visualization of the method is displayed in Figure \ref{fig:idea}. As determining the tightened input set within the CVPM-\gls{MPC} method is difficult in general, we then provide an alternative, computable approach, assuming an uncertainty with symmetric, unimodal \gls{pdf}. A solution approach for a multi-step CVPM-\gls{MPC} is described in Appendix \ref{sec:appendixb}. 

~~

~~

\begin{figure}[t]
\centerline{\includegraphics[width=1.0\columnwidth]{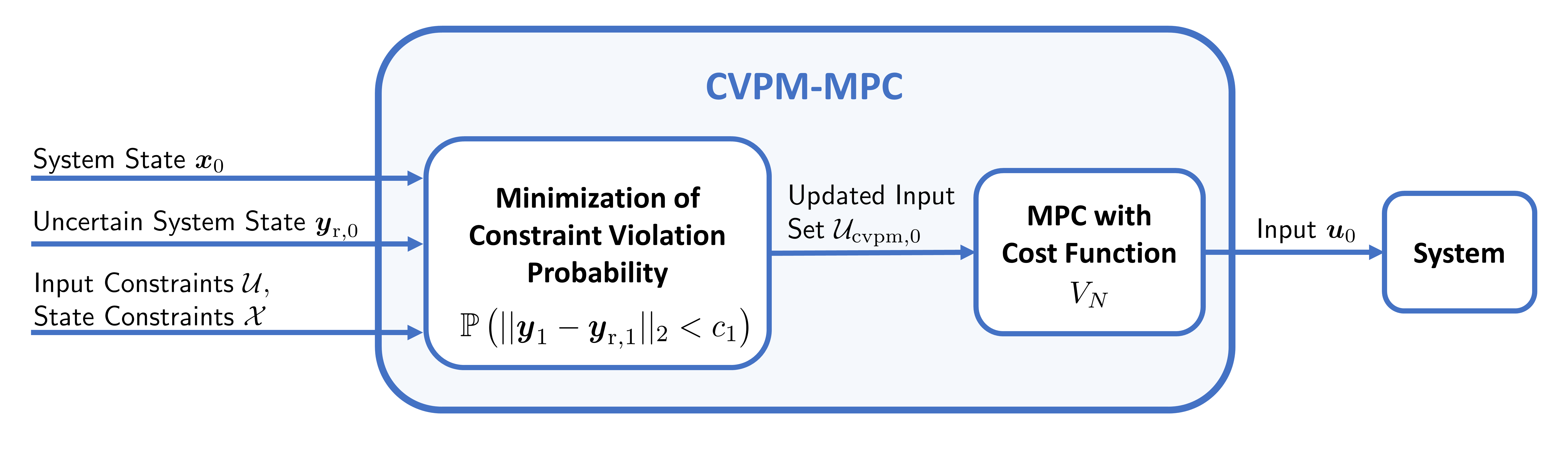}}
\caption{Visualization of CVPM-\gls{MPC} method: Given an input set and state constraints, as well as the current system state and uncertain system state, an updated input set is determined. This updated input set minimizes the norm constraint violation probability for the next step. Then an \gls{MPC} optimization problem is solved. The updated input set ensures constraint violation probability minimization while optimizing for other objectives.}
\label{fig:idea}
\end{figure}

\subsection{General Method to Minimize Constraint Violation Probability for One-Step Problem}
\label{sec:onestepproblem_gen}

When minimizing \gls{Pcv1} over \gls{ui0} within the \gls{MPC} algorithm, three different cases need to be considered. In each case a set $\gls{UUopti0}$ is determined which consists of inputs \gls{ui0} that minimize the constraint violation probability. Ideally, even considering the bounded uncertainty, satisfaction of the constraint can be guaranteed in the next step, for all choices of $\gls{ui0}\in \gls{UUxc0}$, which will be referred to as case 1. However, for stochastic systems we potentially have the situation that case 1 cannot be guaranteed. Here, two cases need to be distinguished. First, given the uncertainty, there is no choice for $\gls{ui0}$ that guarantees constraint satisfaction (case 2). Second, some choices for $\gls{ui0}$ guarantee constraint satisfaction, while other choices do not lead to such a guarantee (case 3).
Depending on the case,~$\gls{UUopti0}$ is determined differently as described in the following. 

\paragraph*{Case 1 (Guaranteed Constraint Satisfaction)} 

The probability of violating the norm constraint is zero independent of the choice for \gls{ui0}, i.e., 
\begin{IEEEeqnarray}{c}
\gls{Pcv1att} = 0 ~~~~ \forall \gls{ui0}\in \gls{UUxc0}. \label{eq:case1g}
\end{IEEEeqnarray}
Therefore, every $\gls{ui0} \in \gls{UUxc0}$ is a valid input, resulting in
\begin{IEEEeqnarray}{c}
\gls{UUopti0} = \gls{UUxc0}. \label{eq:ucase1g}
\end{IEEEeqnarray}

\paragraph*{Case 2 (Impossible Constraint Satisfaction Guarantee)}

There is no choice for $\gls{ui0}$ such that constraint satisfaction can be guaranteed in the presence of uncertainty, i.e.,
\begin{IEEEeqnarray}{c}
\gls{Pcv1att} > 0 ~~~~ \forall \gls{ui0}\in \gls{UUxc0}. \label{eq:case2g}
\end{IEEEeqnarray}
As it is impossible to guarantee $\gls{Pcv1} = 0$, the aim is to minimize $\gls{Pcv1}$. Selecting
\begin{IEEEeqnarray}{c}
\gls{UUopti0} = \setdef[\gls{ui0}]{\gls{ui0} = \argmin{\gls{ui0}\in \gls{UUxc0}} \gls{Pcv1att} }  \label{eq:ucase2g}
\end{IEEEeqnarray}
yields the set \gls{UUopti0} which only consists of inputs \gls{ui0} that minimize \gls{Pcv1}.

\paragraph*{Case 3 (Possible Constraint Satisfaction Guarantee)}

If only some inputs \gls{ui0} guarantee satisfaction of the constraint~\eqref{eq:hc}, i.e.,
\begin{IEEEeqnarray}{c}
\exists~\gls{ui0}\in \gls{UUxc0} \text{~~s.t.~~} \gls{Pcv1att} = 0, \label{eq:case3g}
\end{IEEEeqnarray}
then the set
\begin{IEEEeqnarray}{c}
\gls{UUopti0} = \setdef[\gls{ui0}]{\lr{\gls{Pcv1att} = 0} \land \lr{\gls{ui0}\in \gls{UUxc0}}} \IEEEeqnarraynumspace  \label{eq:ucase3g}
\end{IEEEeqnarray}
consists of these inputs which yield constraint satisfaction.\\

In all three cases \gls{UUopti0} needs to be found, leading to the following strong assumption.
\begin{assumption}
\label{ass:ucvpm}
The set \gls{UUopti0} can be determined for all cases 1-3.
\end{assumption}

While it is possible to approximate \gls{UUopti0} by sampling, finding an analytic solution for \gls{UUopti0} highly depends on the probability distribution. However, if \gls{UUopti0} can be determined, the CVPM-\gls{MPC} method guarantees minimal constraint violation probability for \gls{Pcv1}.

\begin{theorem}
\label{t:CVPM}
If Assumption \ref{ass:ucvpm} holds, minimization of the constraint violation probability of \gls{Pcv1} is guaranteed by selecting \gls{UUopti0} according to cases 1-3.
\end{theorem}

\begin{proof}
The proof follows straightforward from the definition of the three cases. All possibilities are covered regarding the guarantee of constraint satisfaction, i.e., guaranteed constraint satisfaction (case 1), impossible constraint satisfaction guarantee (case 2), and the case where constraint satisfaction is only guaranteed for some but not all $\gls{ui0}\in \gls{UUxc0}$ (case 3). If $\gls{Pcv1} = 0$ is possible, i.e., case 1 or 3, \eqref{eq:ucase1g} and \eqref{eq:ucase3g} guarantee that \gls{UUopti0} consists only of inputs $\gls{ui0}\in \gls{UUxc0}$ which yield $\gls{Pcv1} = 0$. If no $\gls{ui0}\in \gls{UUxc0}$ guarantees $\gls{Pcv1} = 0$, minimal constraint violation is guaranteed by only allowing inputs $\gls{ui0}\in \gls{UUxc0}$ which minimize \gls{Pcv1} according to \eqref{eq:ucase2g}.
\end{proof}

In dynamic environments the worst-case uncertainty $w_{\text{max},k}$ can change over time, which influences the probability of constraint violations. If the support changes, the CVPM-\gls{MPC} approach still minimizes this constraint violation probability.

\begin{corollary}
\label{col:uncsupp}
If the uncertainty support $\supp\lr{\gls{fPw}}$ changes from step $k$ and $k+1$, the CVPM-\gls{MPC} problem solved at step $k+1$ guarantees that the constraint violation probability $p_{\text{cv},k+2}$ is minimized.
\end{corollary}
\begin{proof}
The proof follows directly from the problem definition. First, the CVPM-\gls{MPC} approach ensures that the constraint violation probability is minimized for each step, which allows $p_{\text{cv},k+2} > p_{\text{cv},k+1}$ if the uncertainty support increases. Second, minimizing $p_{\text{cv},k+2}$ is independent of minimizing  $p_{\text{cv},k+1}$.
\end{proof}

The \gls{MPC} problem \eqref{eq:mpc_nc} is now adapted given the set \gls{UUopti0} to guarantee minimal constraint violation probability of \eqref{eq:hc} while still optimizing for further objectives.

\subsection{Model Predictive Control with Minimal First Step Constraint Violation Probability}

Applying the previously determined $\gls{UUopti0}$ yields the CVPM-\gls{MPC} problem
\begin{IEEEeqnarray}{rll}
\IEEEyesnumber \label{eq:mpc_new}
V_N^* = \underset{\gls{uNmpc}}{\min}~ &V_N\lr{\gls{xi},\gls{uNmpc}}, \IEEEyessubnumber  \label{eq:cf}\\
\text{s.t. }& \bm{x}_{j+1} = \bm{A}\bm{x}_j+\bm{B} \bm{u}_j,~~&j \in \mathbb{I}_{0:N-1} \ \IEEEyessubnumber \label{eq:fcon_new}\\
& \bm{y}_{j} = \bm{C}\bm{x}_j,~~&j \in \mathbb{I}_{0:N} \IEEEyessubnumber \label{eq:fcon2_new}\\
& \bm{y}_{\text{r},j+1} = \bm{y}_{\text{r},j} +\bm{u}_{\text{r},j} + \bm{w}_{j},~~&j \in \mathbb{I}_{0:N-1} \IEEEyessubnumber \label{eq:uncsys_new} \IEEEeqnarraynumspace\\
& \gls{uNmpc} \in \gls{UU_total0}.\IEEEyessubnumber  \label{eq:cvm_mpc2} 
\end{IEEEeqnarray}
The set $\gls{UU_total0}$ defines the admissible inputs which yield minimal constraint violation probability combined with keeping the inputs and states within the input and state constraint sets. The set $\gls{UU_total0}$ is given by
\begin{IEEEeqnarray}{c}
\gls{UU_total0} = \setdef[\gls{uNmpc}]{\lr{\gls{ui0} \in \gls{UUopti0}}  \land \lr{\gls{uj} \in \mathbb{U}_{\bm{x},j}, j \in \mathbb{I}_{1:N-1}}} \IEEEeqnarraynumspace \label{eq:Ustar}
\end{IEEEeqnarray}
where $\mathbb{U}_{\bm{x},j}$ is defined in \eqref{eq:UUxc} and $\gls{UUopti0}$ is obtained according to Section \ref{sec:onestepproblem_gen}.

The complete CVPM-\gls{MPC} problem \eqref{eq:mpc_new} allows to optimize a cost function and satisfy state and input constraints, while minimization of the constraint violation probability \gls{Pcv1} is ensured.

\subsection{Minimal Constraint Violation Probability for One-Step Problem with Symmetric Unimodal PDF}
\label{sec:onestepproblem}

The proposed CVMP-\gls{MPC} method in Section \ref{sec:onestepproblem_gen} only guarantees minimal constraint violation probability if Assumption~\ref{ass:ucvpm} is fulfilled. Therefore, it must be possible to always determine \gls{UUopti0}, which is a strong assumption. In the following we provide an adapted approach of the CVMP-\gls{MPC} method which guarantees minimal constraint violation probability if the \gls{pdf} of the uncertainty is symmetric and unimodal.

In the following we first give a definition for symmetric, unimodal \glspl{pdf}. Further, we introduce a substitute for the constraint violation probability \gls{Pcv}. Then, the three cases from Section \ref{sec:onestepproblem_gen} are adapted in order to minimize \gls{Pcv1} for the \gls{pdf} addressed in the following. For each case a convex set of admissible inputs \gls{UUopti0} is determined.

\subsubsection{Symmetric Unimodal \gls{pdf}}

We first define the class of symmetric, unimodal probability distributions.

\begin{definition}[Symmetric Unimodal Distribution]
\label{def:rad}
A probability distribution is symmetric and unimodal if its \gls{pdf} has a single mode which coincides with its mean $\bm{\mu}$ and 
 \begin{IEEEeqnarray}{c}
f\lr{\bm{\mu} + \bm{\delta}_1} = f\lr{\bm{\mu} + \bm{\delta}_2} ~\forall~\norm{\bm{\delta}_1}_2 = \norm{\bm{\delta}_2}_2.
\end{IEEEeqnarray}
\end{definition}

With Definition \ref{def:rad} it is ensured that the \gls{pdf} has its peak at mean $\bm{\mu}$. As the probability distribution is symmetric, all realizations with similar distance to $\bm{\mu}$ have the same relative likelihood. Since there is only one global maximum of the \gls{pdf} at $\bm{\mu}$, realizations with increasing distance to $\bm{\mu}$ have a lower relative likelihood. 

The constraint violation probability \gls{Pcv} is a probabilistic expression and cannot directly be used in the optimization problem. The following assumption will allow to find a deterministic substitute for \gls{Pcv}.

\begin{assumption}[Uncertainty with Symmetric Unimodal \gls{pdf}]
\label{ass:Pcv} 
The \gls{pdf} $\gls{fPw}$ for \gls{Wi} in \eqref{eq:unc_system} is symmetric and unimodal with mean $\bm{\mu} = \bm{0}$. 
\end{assumption}

An example for an admissible probability distribution \gls{Pw} with symmetric, unimodal \gls{pdf} is a truncated isotropic bivariate normal distribution $\mathcal{N}\lr{\bm{0},\Sigma}$ with covariance matrix
\begin{IEEEeqnarray}{c}
\Sigma = \text{diag}\lr{\sigma_1^2, \sigma_2^2} = \sigma^2 \bm{I},~\sigma = \sigma_1 = \sigma_2
\end{IEEEeqnarray}
with variance $\sigma^2$ and identity matrix $\bm{I}$. The support in each direction is required to be equal, which can be achieved by over-approximating. Distributions with $\sigma_1 \neq \sigma_2$ can be over-approximated by choosing 
\begin{IEEEeqnarray}{c}
\Sigma = \sigma_\text{max} \bm{I},~\sigma_\text{max}= \max\lr{\sigma_1,\sigma_2}.
\end{IEEEeqnarray}

We now address the relation between \gls{Pcv} and $\gls{fPw}$ considering Assumption \ref{ass:Pcv}. The following lemma shows that the constraint violation probability \gls{Pcv} can be decreased by choosing $\bm{u}_{k-1}$ such that the distance is increased between the next system output \gls{yi} and the next known, nominal random system output $\overline{\bm{y}}_{\text{r},k}$.

\begin{lemma}
\label{lem:Pcv}
If Assumption \ref{ass:Pcv} holds, the probability \gls{Pcv} is decreasing for an increasing norm \gls{Mkm1b}.
\end{lemma}

\begin{proof}
According to Assumption \ref{ass:Pcv}, $\gls{fPw}$ is symmetric and unimodal, and therefore $\gls{fPw}$ is decreasing with increasing $\norm{\gls{wi}}_2$, i.e., the larger the value $\norm{\gls{wi}}_2$, the lower its probability. The realization with the highest relative likelihood is the mode of \gls{fPcv} with $\gls{wi} = \bm{0}$, yielding the most likely random output $\gls{yrip1} = \overline{\bm{y}}_{\text{r},k+1}$. It follows that
 \begin{IEEEeqnarray}{c}
 \label{eq:lemPcv}
\gls{fPw}\lr{\tilde{\bm{w}}_k} < \gls{fPw}\lr{\gls{wi}}~~\text{for}~~\norm{\tilde{\bm{w}}_k}_2 > \norm{\gls{wi}}_2 
\end{IEEEeqnarray}
where $\tilde{\bm{y}}_{\text{r},k+1} = \overline{\bm{y}}_{\text{r},k+1} + \tilde{\bm{w}}_k$ is less likely than $\gls{yrip1} = \overline{\bm{y}}_{\text{r},k+1} + \gls{wi}$ and
 \begin{IEEEeqnarray}{c}
\norm{\overline{\bm{y}}_{\text{r},k+1} - \tilde{\bm{y}}_{\text{r},k+1}}_2 > \norm{\overline{\bm{y}}_{\text{r},k+1} - \gls{yrip1}}_2
\end{IEEEeqnarray}
due to $\norm{\tilde{\bm{w}}_k}_2 > \norm{\gls{wi}}_2$. 

It follows that the larger the value $\norm{\gls{yip1} - \overline{\bm{y}}_{\text{r},k+1}}_2$, the higher the probability of a large value \gls{Mkp1} due to \eqref{eq:lemPcv}. Therefore, the larger $\norm{\gls{yip1} - \overline{\bm{y}}_{\text{r},k+1}}_2$, the higher the probability of $\gls{Mkp1} \geq \gls{Mmink}$. This results in
\begin{IEEEeqnarray}{c}
\norm{\tilde{\bm{y}}_{k+1}-\overline{\bm{y}}_{\text{r},k+1}}_2 > \norm{\bm{y}_{k+1}-\overline{\bm{y}}_{\text{r},k+1}}_2 	 ~~\Leftrightarrow~~ p_{\text{cv},k+1}\lr{\tilde{\bm{u}}_k,\bm{y}_{\text{r},k}} \leq p_{\text{cv},k+1}\lr{\bm{u}_k,\bm{y}_{\text{r},k}}\IEEEeqnarraynumspace
\end{IEEEeqnarray}
with $\tilde{\bm{y}}_{k+1} \neq \bm{y}_{k+1}$ and $\tilde{\bm{y}}_{k+1} = \bm{C}\lr{\bm{A}\bm{x}_k + \bm{B} \tilde{\bm{u}}_k}$ according to \eqref{eq:system}.

The same holds for \gls{Mkm1b} instead of  $\norm{\gls{yip1} - \overline{\bm{y}}_{\text{r},k+1}}_2$, showing that \gls{Pcv} is decreasing with an increasing \gls{Mkm1b}.
\end{proof}

The lemma shows that the probability of violating the norm constraint \eqref{eq:hc} decreases if the difference between \gls{yi} and $\overline{\bm{y}}_{\text{r},k}$ increases. Lemma \ref{lem:Pcv} now allows to find a substitute function for \gls{Pcv}.

\subsubsection{Substitute Probability Function}

In order to provide a tractable substitution for \gls{Pcvj} in the CVPM-\gls{MPC} problem, we introduce the scalar, twice differentiable, strictly monotonically increasing function 
\begin{IEEEeqnarray}{c}
\gls{gi}: \mathbb{R}_{\geq0} \rightarrow \mathbb{R}
\end{IEEEeqnarray}
where,
\begin{IEEEeqnarray}{c}
\gls{gi} \lr{\gls{Mjb}} \IEEEeqnarraynumspace
\label{eq:h}
\end{IEEEeqnarray}
is used as a substitution for the constraint violation probability $\gls{Pcvj}$, as $\gls{Pcvj}$ is decreasing with the norm \gls{Mjb} according to Lemma \ref{lem:Pcv}, while $\gls{gi} \lr{\gls{Mjb}} $ is increasing with $\gls{Mjb}$. This substitution has the benefit that an increasing \gls{Mjb} is equivalent to a decreasing constraint violation probability.

Considering the constraint violation probability for the first predicted step $j=1$, this probability $\gls{Pcv1}$ is minimized for a maximal $\gls{hiyy1b}$. However, since $\gls{fPw}$ is truncated and $\supp\lr{\gls{Pcv}}$ is bounded, there potentially are multiple admissible inputs which result in an equal constraint violation probability. The aim is now to find the convex set \gls{UUopti0} including all inputs $\gls{umaxi0} \in \gls{UUopti0}$ which result in a minimal \gls{Pcv1}. As $\bm{u}_{\text{r},0}$ is deterministic and known according to Assumption \ref{ass:initunc}, \gls{hiyy1b} is a deterministic expression that can be evaluated.

The set $\gls{UUopti0}$ can then be found by comparing the worst-case uncertainty $w_{\text{max},0}$ with the minimum and maximum possible values of $\gls{hiyy1b}$, i.e., $\gls{hmini1}$ and $\gls{hmaxi1}$, respectively. The maximal value \gls{hmaxi1} is determined by 
\begin{IEEEeqnarray}{c}
\gls{hmaxi1} \coloneqq \max_{\gls{ui0}\in \gls{UUxc0}} \gls{gi} \lr{\gls{Mm1b}} =    \gls{gi} \lr{ \max_{\gls{ui0}\in \gls{UUxc0}} \lr{\gls{Mm1b}}}
\label{eq:hmax}
\end{IEEEeqnarray}
corresponding to the largest distance between \gls{yi1} and $\overline{\bm{y}}_{\text{r},1}$. Analogously \gls{hmini1} can be found by
\begin{IEEEeqnarray}{c}
\gls{hmini1} \coloneqq \min_{\gls{ui0}\in \gls{UUxc0}} \gls{gi} \lr{\gls{Mm1b}} =   \gls{gi} \lr{ \min_{\gls{ui0}\in \gls{UUxc0}} \lr{\gls{Mm1b}}}.
\label{eq:hmin}
\end{IEEEeqnarray}

The result for $\gls{hmini1}$ can be obtained by determining the minimum value of $\gls{Mm1b}$, as the objective function $\gls{gi} \lr{\gls{Mm1b}}$ and $\gls{UUxc0}$ are convex. The following lemma provides a strategy to find $\gls{hmaxi1}$.

\begin{lemma} \label{lem:globmax}
Let the non-empty convex polytope $\set{V}\subset \mathbb{R}^g$, $g \in \mathbb{N}$, be bounded by a finite set of hyperplanes, such that $\set{V}$ has a finite number of edge vertices with a convex function $\gls{convfunc}:\set{V} \rightarrow \mathbb{R}$. Then a global maximum
\begin{IEEEeqnarray}{c}
\var{\gls{convfunc}}{max} = \max_{\bm{v}\in \set{V}}\f{\gls{convfunc}}{\bm{v}} 
\end{IEEEeqnarray}
is obtained by searching for the maximum value of $\gls{convfunc}$ on the boundary $\partial \set{V}$ of its domain $\set{V}$.
\end{lemma}

\begin{proof}
This proof is based on Bauer's maximum principle \cite{Bauer1958}. We consider any two points $\bm{v}_1,\bm{v}_2 \in \partial\set{V}$ on the boundary of $\set{V}$. Any point on the line between $\bm{v}_1,\bm{v}_2$ can be described by $\bm{b} = \lambda \bm{v}_1 + (1-\lambda)\bm{v}_2$, using the definition of convexity. Due to the convexity of $\gls{convfunc}$ it holds that $\f{\gls{convfunc}}{\bm{b}}~\leq~\max \left\{\f{\gls{convfunc}}{\bm{v}_1},\f{\gls{convfunc}}{\bm{v}_2}\right\}$. Any point on the line between $\bm{v}_1,\bm{v}_2$ can be reached by a convex combination. Since $\bm{v}_1,\bm{v}_2$ can be chosen arbitrarily, every point $\bm{b}$ in the interior of $\set{V}$ can be reached. Therefore, a global maximum $\var{\gls{convfunc}}{max}$ is found on the boundary $\partial\set{V}$.
\end{proof}

\subsubsection{Determination of the Updated Admissible Input Set}

Similar to Section \ref{sec:onestepproblem_gen} we regard three cases. The resulting set \gls{UUopti0}, depending on the three cases, is then used in the CVPM-\gls{MPC} problem \eqref{eq:mpc_new} to guarantee minimal constraint violation probability of \eqref{eq:hc}. In order to distinguish between the cases, we will consider the relation
\begin{IEEEeqnarray}{c}
\label{eq:ineq_trafo}
\norm{\bm{y}_1 - \overline{\bm{y}}_{\text{r},1}}_2   \geq c_1 + w_{\text{max},0} \Rightarrow \norm{\bm{y}_1 - \bm{y}_{\text{r},1}}_2 \geq c_1, \IEEEeqnarraynumspace
\end{IEEEeqnarray}
where a detailed derivation of \eqref{eq:ineq_trafo} is shown in Appendix \ref{sec:appendixc}. Here, \gls{Mtot1} represents the necessary distance between $\bm{y}_{1}$ and $\overline{\bm{y}}_{\text{r},1}$, consisting of the required minimal distance \gls{Mmin} at step $j=1$ and the maximal random system step \gls{wimax0} at $j=0$, such that $\norm{\bm{y}_1 - \bm{y}_{\text{r},1}}_2 \geq c_1$ for all $\norm{\bm{w}_0}_2 \leq w_{\text{max},0}$. 
\paragraph*{Case 1 (Guaranteed Constraint Satisfaction)} 
For any $\gls{ui0}~\in~\gls{UUxc0}$ constraint satisfaction is guaranteed, i.e., $\gls{Pcv1}=0$ for
\begin{IEEEeqnarray}{c}
\gls{hmini1} \geq \gls{gi}\lr{\gls{Mtot1}}. \label{eq:case1}
\end{IEEEeqnarray}
The initial state configuration of the controlled and the stochastic system is such that the minimum value possible for $\gls{hi}\lr{\gls{Mm1b}}$, \gls{hmini1}, still yields a larger value than inserting \gls{Mmin} combined with the worst-case \gls{wimax0} into \gls{hi}, which moves $\bm{y}_{\text{r},1}$ closest to $\bm{y}_{1}$. This results in a guaranteed constraint satisfaction $\gls{Pcv1}=0$. 
Therefore, every $\gls{ui0} \in \gls{UUxc0}$ is an admissible input, i.e.,
\begin{IEEEeqnarray}{c}
\gls{UUopti0} = \gls{UUxc0}. \label{eq:ucase1}
\end{IEEEeqnarray}

\paragraph*{Case 2 (Impossible Constraint Satisfaction Guarantee)}
There is no input $\gls{ui0}\in \gls{UUxc0}$ which can guarantee $\gls{Pcv1} = 0$, i.e.,
\begin{IEEEeqnarray}{c}
\gls{hmaxi1} < \gls{gi}\lr{\gls{Mmin} +\gls{wimax0}}. \label{eq:case2}
\end{IEEEeqnarray}
The largest value for \gls{hiyy1b} that can be achieved with $\gls{ui0} \in \gls{UUxc0}$ is $\gls{hmaxi1}$, corresponding to the lowest possible $\gls{Pcv1}$. However, to guarantee constraint satisfaction of \eqref{eq:hc}, \gls{hmaxi1} is required to be larger or at least equal to $\gls{gi}(\gls{Mtot1})$, with the worst-case absolute value \gls{wimax0} for the realization of \gls{wi0}. Constraint satisfaction cannot be guaranteed here.

The solution corresponding to $\gls{hmaxi1}$ is denoted by $\gls{umaxi0}$. Minimal \gls{Pcv1} is achieved with
\begin{IEEEeqnarray}{c}
\gls{umaxi0} = \argmax{\gls{ui0}\in \gls{UUxc0}} \gls{hiyy1b}
\label{eq:u_max}
\end{IEEEeqnarray}
as $\gls{hiyy1b}$ increases and \gls{Pcv1} decreases with an increasing norm.

 Therefore, 
\begin{IEEEeqnarray}{c}
\gls{UUopti0} = \{\gls{umaxi0}\} \label{eq:ucase2}
\end{IEEEeqnarray}
is selected since the input choice \gls{umaxi0} guarantees the lowest constraint violation probability when $\gls{Pcv1} > 0$.

\begin{remark}
If \eqref{eq:u_max} yields more than one solution, \gls{UUopti0} in \eqref{eq:ucase2} can also consist of more than one element, i.e., all solutions of \eqref{eq:u_max}. However, there can be restrictions if convexity of \gls{UUopti0} is required. 
\end{remark}

\paragraph*{Case 3 (Possible Constraint Satisfaction Guarantee)}
The final case yields $\gls{Pcv1}=0$ for some \gls{ui0} and applies if
\begin{IEEEeqnarray}{c}
\gls{hmaxi1} \geq \gls{gi}\lr{\gls{Mtot1}} > \gls{hmini1} . \label{eq:case3} 
\end{IEEEeqnarray}
While some $\gls{ui0} \in \gls{UUxc0}$ cannot guarantee zero constraint violation probability, it is possible to find $\gls{ui0}$ such that 
\begin{IEEEeqnarray}{c}
\gls{hiyy1b} \geq \gls{gi}\lr{\gls{Mtot1}}.
\end{IEEEeqnarray}
Therefore, for some $\gls{ui0}$ constraint satisfaction can be guaranteed in the presence of uncertainty. Hence, the task is to find a set 
\begin{IEEEeqnarray}{c}
\gls{UUopti0} =  \setdef[\gls{ui0}]{\left(\gls{hiyy1b}\geq \gls{gi}\lr{\gls{Mtot1}}\right)   \land \lr{\gls{ui0} \in \gls{UUxc0}}},
\label{eq:Uopt_c3} \IEEEeqnarraynumspace
\end{IEEEeqnarray}
which consists of all inputs $\gls{ui0} \in \gls{UUxc0}$ that yield constraint satisfaction and therefore $\gls{Pcv1} = 0$.

The first part of the set in \eqref{eq:Uopt_c3},
\begin{IEEEeqnarray}{c}
\gls{UUm30} = \setdef[\gls{ui0}]{\gls{hiyy1b} \geq \gls{gi}\lr{\gls{Mtot1}}},
\label{eq:UUm3} \IEEEeqnarraynumspace
\end{IEEEeqnarray}
describes a super-level set, including only inputs \gls{ui0} which lead to $\gls{Pcv1} =0$. This super-level set is generally non-convex. In order to receive a convex set $\gls{UUopti0}$ for the optimization problem, an approximation is performed, based on the boundary 
\begin{IEEEeqnarray}{c}
\gls{hic0} = \setdef[\gls{ui0}]{\gls{hiyy1b} = \f{\gls{gi}}{\gls{Mtot1}}}. \IEEEeqnarraynumspace
\end{IEEEeqnarray}

\begin{proposition}
\label{prop:apUUopt}
An approximated, convex solution of \eqref{eq:Uopt_c3} in case 3 is obtained by 
\begin{IEEEeqnarray}{c}
\gls{UUopti0}  = \gls{apUUopti0} = \setdef[\gls{ui0}]{\gls{UUti0}\cap \gls{UUxc0}} 
\label{eq:Uopt_th}
\end{IEEEeqnarray}
with 
\begin{IEEEeqnarray}{c}
\gls{UUti0} =  \setdef[\gls{ui0}]{\lr{\nabla_{\gls{pp}} \gls{hiyy1us}}^\top \left(\gls{ui0}-\gls{pp}\right)\geq 0}, \IEEEeqnarraynumspace
\label{eq:th1.2}
\end{IEEEeqnarray}
the gradient operator $\nabla_{\gls{pp}}$, and a point $\gls{pp} \in \gls{hic0} \cap \gls{UUxc0}$ which is an admissible input.
\end{proposition}

\begin{remark}
While it was previously not explicitly stated that \gls{yi1} depends on \gls{ui0}, in Proposition \ref{prop:apUUopt} the dependence of \gls{yi1} on \gls{pp} is stated for clarity.
\end{remark}

\begin{proof}
The set $\gls{UUm30}
$ is non-empty and non-convex with the boundary point $\gls{pp} \in \partial\gls{UUm30}$ of $\gls{UUm30}$. There exists a supporting hyperplane to $\gls{UUm30}$ at $\gls{pp}$ \cite{boyd2004convex}. This supporting hyperplane is used to approximate the non-convex set \gls{UUm30}. The gradient $\nabla_{\gls{pp}} \lr{\gls{hiyy1us}}$ is a vector orthogonal to the hyperplane on the boundary $\gls{hic0}$ at $\gls{pp}$, pointing away from the convex set $\gls{UUm30}$. The scalar product of $\nabla_{\gls{pp}} \lr{\gls{hiyy1us}}$ and any point $\gls{ui0}$ on this hyperplane is zero, while the scalar product of $\nabla_{\gls{pp}} \lr{\gls{hiyy1us}}$ and any point in the half plane not containing $\gls{UUm30}$ is positive. Therefore, \eqref{eq:th1.2} approximates \gls{UUm30}. As the intersection of two convex sets yields a convex set \cite{boyd2004convex}, the resulting approximated set $\gls{apUUopti0}$ is convex as well. 
\end{proof}

An approach to finding $\gls{pp}$ is solving the system
\begin{IEEEeqnarray}{c}
\gls{hiyy1us} = \gls{gi}\lr{\gls{Mtot1}}
\end{IEEEeqnarray}
with $\gls{pp} \in \gls{UUxc0}$. The choice of $\gls{pp}$ is not unique. It is possible that $\gls{apUUopti0}$ is empty due to approximating even though case~3 applies. 

\begin{remark}
\label{rem:emptyapp}
If $\gls{apUUopti0} = \emptyset$ in case 3, then $\gls{ui0}$ can be determined by following the procedure of case 2.
\end{remark}

Following the approach in Remark \ref{rem:emptyapp} still provides a solution that minimizes \gls{Pcv1}. However, in case 2 only a single option $\gls{UUopti0} = \gls{umaxi0}$ is given, while case 3 has the advantage of providing a set \gls{UUopti0} with multiple possible inputs \gls{ui0}. Case~3 therefore offers the possibility to then optimize to account for further objectives, given the set of admissible inputs \gls{UUopti0}.

\section{Properties}
\label{sec:properties}

In the following two important properties are analyzed. First, recursive feasibility of the proposed method is shown. This is followed by a proof of convergence. 

\subsection{Recursive Feasibility}

Recursive feasibility guarantees that if the \gls{MPC} optimization problem is solvable at step $k$, it is also solvable at step $k+1$. This needs to hold as \gls{MPC} requires the solution of an optimal control problem at every time step.

\begin{definition}(Recursive Feasibility)
Recursive feasibility of an \gls{MPC} algorithm is guaranteed if
\begin{IEEEeqnarray}{c}
\mathbb{U}^N_k \neq \emptyset \Rightarrow \mathbb{U}^N_{k+1} \neq \emptyset
\end{IEEEeqnarray}
where $\mathbb{U}^N_k$ is the set of admissible inputs $\gls{uNmpck}$ to fulfill the constraint \eqref{eq:cvm_mpc2} from step $k$ to step $k+N$. 
\end{definition}

In the following recursive feasibility will be established for the proposed method. Without loss of generality the \gls{MPC} optimization problem starts at \gls{xi} with $k=0$.

\begin{theorem}
\label{t:recfeas}
The CVPM-\gls{MPC} algorithm in \eqref{eq:mpc_new} is recursively feasible with the general CVPM approach of Section \ref{sec:onestepproblem_gen}.
\end{theorem}

The proof is divided into two parts. First it is shown that $\gls{UUopti0} \neq \emptyset$ at any step, then recursive feasibility of the optimization problem \eqref{eq:mpc_new} is shown.

\begin{proof}
As shown in the proof of Theorem \ref{t:CVPM} the three cases \eqref{eq:case1g}, \eqref{eq:case2g}, and \eqref{eq:case3g} cover all possibilities with individual, nonempty sets \gls{UUopti0}. This yields that there always exists a $\gls{ui0} \in \gls{UUopti0}$.

As $\mathbb{U}_{\bm{x},j}$ is a non-empty set due to Assumption \ref{ass:XXf}, there exist solutions $\bm{u}_j \in \mathbb{U}_{\bm{x},j}$ for $j \in \mathbb{I}_{1:N-1}$. The first condition in~\eqref{eq:Ustar} considers the first input \gls{ui0}, while the second condition covers the following inputs \gls{uj} with $j \in \mathbb{I}_{1:N-1}$. Therefore, the two conditions are independent and $\gls{UU_total0} \neq \emptyset$ for any \gls{MPC} optimization. The \gls{MPC} algorithm \eqref{eq:mpc_new} is guaranteed recursively feasible.
\end{proof}

The proof for the general CVPM-\gls{MPC} method can be extended for the CVPM-\gls{MPC} approach for uncertainties with symmetric, unimodal \glspl{pdf} in Section \ref{sec:onestepproblem}.

\begin{corollary}
\label{col:recfeas}
If Assumption \ref{ass:Pcv} holds, the CVPM-\gls{MPC} algorithm in \eqref{eq:mpc_new} is recursively feasible with the CVPM approach of Section \ref{sec:onestepproblem}. 
\end{corollary}

\begin{proof}
The proof follows straightforward from Theorem~\ref{t:recfeas}, showing that $\gls{UUopti} \neq \emptyset$ for all three cases \eqref{eq:case1}, \eqref{eq:case2}, and \eqref{eq:case3}. According to Lemma \ref{lem:globmax}, $\gls{hmini1}$ and $\gls{hmaxi1}$ can always be found. Given any value for $\gls{gi}\lr{\gls{Mtot1}}$, exactly one of the three cases is applicable, yielding $\gls{UUopti0} \neq \emptyset$. For case 1 and 2 no approximation is necessary. If $\gls{apUUopti0} = \emptyset$ for case 3, the approach of case 2 is used according to Remark \ref{rem:emptyapp}, i.e., $\gls{UUopti0} = \{\gls{umaxi0}\}$. Therefore, $\gls{UUopti} \neq \emptyset$ for all three cases.
\end{proof}

Theorem \ref{t:recfeas} and Corollary \ref{col:recfeas} show that if the \gls{MPC} problem~\eqref{eq:mpc_gen} is designed to be recursively feasible, the CVPM-\gls{MPC} algorithm~\eqref{eq:mpc_new}, based on \eqref{eq:mpc_gen}, remains recursively feasible. According to Corollary \ref{col:uncsupp}, minimizing $p_{\text{cv},k}$ is independent of uncertainty support, therefore, recursive feasibility is guaranteed if the uncertainty support changes.

\subsection{Convergence}

In the following convergence of the proposed method is shown. In this section the \gls{MPC} optimization starts at \gls{xk}. Considering Remark \ref{rem:ref} it is possible to track a reference varying from the origin, however, without loss of generality we will only consider the regulation of the origin here. 

The uncertain output $\gls{yri}$ can potentially lie close to the origin or even directly in the origin. In order to minimize \gls{Pcv}, an area around $\gls{yri}$ is then inadmissible for the system output $\gls{yi}$. This can lead to the case where the origin is inadmissible for the controlled system, i.e., $\bm{0} \in \gls{XXnorm}$, where 

\begin{IEEEeqnarray}{c}
\gls{XXnorm} =  \setdef[\gls{xk}]{\gls{Pcv}\lr{\bm{u}_{k-1}} > 0,~\gls{xk} = \bm{A} \bm{x}_{k-1} + \bm{B} \bm{u}_{k-1} } \IEEEeqnarraynumspace \label{eq:XXnorm}
\end{IEEEeqnarray}
denotes the bounded and open set of states $\gls{xk}$ with $\gls{Pcv} > 0$, i.e., constraint violation is possible for all $\gls{xk} \in \gls{XXnorm}$. An inadmissible origin, of course, is an issue when investigating the stability of the proposed algorithm. However, we will provide a convergence guarantee under the following two Assumptions concerning the stochastic nature of $\gls{yri}$.

\begin{assumption}[Admissible Origin]
\label{ass:yr_kinfty2}
(a) There exists a $k_{0}~<~\infty$ such that for all $k \geq k_{0}$ it holds that
\begin{IEEEeqnarray}{c}
\label{eq:adm_origin_a}
\bm{0} \notin \gls{XXnorm}~~\forall~k \geq k_{0}.
\end{IEEEeqnarray}
(b) There exists a $k_{y0} < \infty$ and a finite sequence of inputs $\gls{ui}$ such that $\bm{y}_{k}~=~\bm{0}$  for all $k \geq k_{y0} \geq k_{0} $. \\(c) There exists a $k_{\text{case}1,3} < \infty$ and for all $k \geq k_{\text{case}1,3} \geq k_{0}$
\begin{IEEEeqnarray}{c}
\label{eq:adm_origin_c}
\exists~\bm{u}_{k-1}~~\text{s.t.}~~p_{\text{cv},k}\lr{\bm{u}_{k-1}} = 0
\end{IEEEeqnarray}
and $\gls{UUopti} \neq \emptyset$.
\end{assumption}

Assumption \ref{ass:yr_kinfty2} (a) ensures that even if $\gls{yri}$ is occupying the space around the origin for some time, eventually \gls{yri} will be distanced enough that the origin becomes admissible for the controlled system, as the boundedness of the stochastic system state yields a closed admissible space for the controlled system. Assumption \ref{ass:yr_kinfty2} (b) ensures that there is a possibility for the controlled system to reach the origin. 

With Assumption \ref{ass:yr_kinfty2} (c) it is guaranteed that either case 1 or case 3 is applicable if Assumption \ref{ass:yr_kinfty2} (a) holds. This ensures that $\gls{Pcv}=0$ at some time after the origin becomes admissible for the controlled system. 

\begin{lemma}
\label{lem:tildeXX}
If Assumptions \ref{ass:yr_kinfty2} holds, there exists a closed, control invariant set $\tilde{\gls{XX}}_k = \gls{XX} \setminus \gls{XXnorm}$ for $k \geq k_{\text{case}1,3}$, which contains the origin.
\end{lemma}
\begin{proof}
As cases 1 or 3 are applied, the space blocked by $\gls{XXnorm}$ around \gls{yri} with non-zero constraint violation probability can be regarded as a hard constraint. This yields $\gls{xk} \notin \gls{XXnorm}$ for all $k \geq k_{\text{case}1,3}$. As \gls{XX} is closed and $\gls{XXnorm}$ is open, the resulting set $\tilde{\gls{XX}}_k$ is closed. As $\gls{xk} \in \tilde{\gls{XX}}_k \subseteq \gls{XX}$, there exists a $\gls{ui}$ such that $\bm{x}_{k+1} \in \gls{XX}$ according to Theorem \ref{t:recfeas}. Assumption~\ref{ass:yr_kinfty2}~(c) ensures that \gls{UUopti} is not empty, therefore $\bm{x}_{k+1} \in \tilde{\gls{XX}}_k$ and $\tilde{\gls{XX}}_k$ is control invariant. 
\end{proof}

The set $\tilde{\gls{XX}}$ consists of the states which ensure constraint satisfaction of \gls{XX} and yield $\gls{Pcv} = 0$ for $k \geq k_{\text{case}1,3}$.

\begin{assumption}[Terminal Constraint Set]
\label{ass:XfinXXt}
The terminal constraint set \gls{XXf} is a subset of $\tilde{\gls{XX}}_k$, i.e., $\gls{XXf} \subset \tilde{\gls{XX}}_k$.
\end{assumption}

In the following convergence of the proposed method is addressed.

\begin{theorem}
\label{t:convergence}
If Assumptions \ref{ass:termcost} and \ref{ass:yr_kinfty2} hold, the proposed CVPM-\gls{MPC} method in Section \ref{sec:onestepproblem_gen} satisfies that $\gls{xk}$ converges to~$\bm{0}$ for $k \rightarrow \infty$.
\end{theorem}
\begin{proof}
First, the \gls{MPC} algorithm in \eqref{eq:mpc_gen} will be considered without the norm constraint \eqref{eq:hc}. As $V_N\lr{\gls{xi},\gls{uNmpc}}$ is a Lyapunov function in \gls{XX}, given Assumption \ref{ass:termcost}, the \gls{MPC} algorithm of \eqref{eq:mpc_new} without \eqref{eq:hc} is asymptotically stable, following the \gls{MPC} stability proof of Rawlings et al.~\cite[Chap.~2.4]{RawlingsMayneDiehl2017}.

Now the CVPM-\gls{MPC} method is considered. According to Theorem \ref{t:recfeas}, for all $k$, $\gls{xk} \in \gls{XX}$ there exists a feasible \gls{uNmpck} such that $\bm{x}_{k+1}$ remains in \gls{XX}. Lemma \ref{lem:tildeXX} ensures that $\bm{x}_{k^*}  \in \tilde{\gls{XX}}_k$ for $k^* \geq k_{\text{case}1,3}$, where $\tilde{\gls{XX}}_k$ replaces \gls{XX} to ensure constraint satisfaction of the norm constraint. The set $\tilde{\gls{XX}}_k$ is closed, control invariant, contains the origin according to Assumption \ref{ass:yr_kinfty2}, and $\gls{XXf} \subseteq \tilde{\gls{XX}}_k$, given Assumption \ref{ass:XfinXXt}. Therefore, the system \eqref{eq:system}, controlled by the CVPM-\gls{MPC} algorithm in \eqref{eq:mpc_new}, is asymptotically stable and converges to $\bm{0}$ for $k > k^*$ and $k \rightarrow \infty$, similar to the \gls{MPC} algorithm in \eqref{eq:mpc_gen}.
\end{proof}

In Theorem \ref{t:convergence} it is only shown that the system converges to the origin once the random system fulfills Assumption \ref{ass:yr_kinfty2}. However, every time the stochastic output allows the system to reach the origin, the system will move towards the origin. The system state~$\gls{xk}$ remains at $\bm{0}$ until $\gls{yri}$ moves in such a way that the origin has non-zero constraint violation probability. As the main goal is to ensure minimum constraint violation probability of \eqref{eq:cc}, \gls{yi} will move away from the origin to minimize \gls{Pcv} if \gls{yri} behaves in such a way that it causes $\gls{Pcv} > 0$ in the origin.

\begin{corollary}
\label{col:convergence}
If Assumptions \ref{ass:yr_kinfty2} holds, the proposed CVPM-\gls{MPC} method in Section \ref{sec:onestepproblem} for uncertainties with symmetric, unimodal \glspl{pdf} satisfies that $\gls{xk} \in \gls{XX}$ for all $k$ and that $\gls{xk}$ converges to $\bm{0}$ for $k \rightarrow \infty$.
\end{corollary}

\begin{proof}
The proof is similar to the proof of Theorem \ref{t:convergence}. The set \gls{XXnorm} in \eqref{eq:XXnorm} can be expressed as
\begin{IEEEeqnarray}{rl}
\gls{XXnorm} = \setdef[\gls{xk}]{\gls{hykb} < \f{\gls{gi}}{c_k + w_{\text{max},k-1}},~\bm{y}_k = \bm{C} \bm{x}_k}. \IEEEeqnarraynumspace
\end{IEEEeqnarray}
Equation \eqref{eq:adm_origin_a} is satisfied by
\begin{IEEEeqnarray}{c}
{\gls{h0ykb}} \geq \f{\gls{gi}}{c_k + w_{\text{max},k-1}}~~\forall~k \geq k_{0}
\end{IEEEeqnarray}
while \eqref{eq:adm_origin_c} transforms into
\begin{IEEEeqnarray}{c}
\exists~\bm{u}_{k-1}~~\text{s.t.}~~\gls{hykb} \geq \f{\gls{gi}}{c_k + w_{\text{max},k-1}}\IEEEeqnarraynumspace
\end{IEEEeqnarray}
for the CVPM-\gls{MPC} method in Section \ref{sec:onestepproblem}.

Similar to Lemma \ref{lem:tildeXX}, given the open and constant set \gls{XXnorm}, $\tilde{\gls{XX}}_k$ is closed, constant, control invariant, and contains the origin given Assumption \ref{ass:yr_kinfty2}. With the \gls{MPC} algorithm \eqref{eq:mpc_gen} and $k > k^*$, $k \rightarrow \infty$ the system \eqref{eq:system} is asymptotically stable and therefore converges to $\bm{0}$.
\end{proof}

Therefore, if the origin is admissible, the controlled system will converge. However, satisfying the norm constraint has priority over converging to the origin.

\section{Discussion of the Proposed CVPM-\gls{MPC} Method}
\label{sec:discussion}

One could argue now that the proposed algorithm is a combination of \gls{RMPC} in the first step and, potentially, \gls{SMPC} in the following steps. While there are some similarities to this combination, we solve a different problem. The most important difference is that the constraint violation probability is minimized in the first predicted step and the initial uncertainty probability is not required to be 0. \gls{RMPC} approaches require constraint satisfaction initially and ensure that constraints are satisfied throughout the prediction horizon.

Our proposed CVPM-\gls{MPC} method is more closely related to \gls{SMPC} than \gls{RMPC}, as constraint violations are possible. Nevertheless, the suggested method can be interpreted as lying between \gls{SMPC} and \gls{RMPC}. The results are more conservative than \gls{SMPC}, as a zero percent constraint violation probability is found if possible, i.e., $\gls{Pcv} = 0$ in \eqref{eq:cc}, but less conservative than \gls{RMPC}. An advantage over both, \gls{SMPC} and \gls{RMPC}, is the ability to minimize the constraint violation probability and to successfully cope with sudden uncertainty support changes, as recursive feasibility can still be guaranteed. The uncertainty support can change due to unexpected events or modeling inaccuracies.

In \gls{SMPC} with chance constraints recursive feasibility is a major issue. For example, an unexpected realization of the uncertainty at step $k$, whose probability lies below the chance constraint risk parameter at step $k$, leads to a state at step $k+1$ with no solution to the optimization problem if the required risk parameter of the chance constraint cannot be met. An option to regain feasibility is to solve an alternative optimization problem or apply an input that was previously defined. However, these alternatives do not necessarily lead to a solution that yields the lowest constraint violation probability. Furthermore, it is possible to soften chance constraints by using slack variables in the cost function. However, this approach is not acceptable in applications where the chance constraint represents a safety constraint. If a slack variable is introduced, it competes with other objectives within the cost function and does not ensure constraint satisfaction. The proposed CVPM-\gls{MPC} method always finds the optimal input that results in the lowest constraint violation probability while remaining recursively feasible. 

\gls{RMPC} guarantees recursive feasibility but at the cost of reduced efficiency, as worst-case scenarios need to be taken into account. Additionally, if the support of the uncertainty can suddenly change over time, e.g., the future motion of an object becomes more uncertain due to a changing environment, \gls{RMPC} can become too conservative to be applicable. A robust solution can only be obtained by always considering the largest possible uncertainty support. The proposed method deals with this by adjusting to changing uncertainty supports at every step, as will be illustrated in Section \ref{sec:results}. A suddenly or unexpectedly increasing uncertainty support, e.g., due to an inaccurate prediction model, may lead to increased constraint violation probability for a limited time after the support changes. Before the support changes, the optimized inputs of the proposed algorithm lead to a less conservative result than \gls{RMPC}, while ensuring that the constraint violation probability is kept at a minimal level immediately after the change.

In the proposed method we only consider minimizing the constraint violation for the first predicted step. It is possible to consider multiple steps by increasing the uncertainty support for each considered step as described in Appendix \ref{sec:appendixb}, however, this leads to a more conservative solution. For every extra predicted step in which the constraint violation probability is minimized, the maximal possible uncertainty value must be considered. This yields a highly restrictive set of admissible inputs which minimize the constraint violation probability over multiple predicted steps. As it is assumed that the support of the uncertainty \gls{pdf} can change over time, considering multiple steps with the initially known support does not guarantee lower constraint violation probability for multiple steps. If the support increases the previously obtained multi-step CVPM-\gls{MPC} solution becomes invalid. Therefore, given an updated uncertainty support at each step, it is a reasonable approach to only minimize the constraint violation probability for the first predicted step, resulting in the safest solution at the current step. It is possible to consider the norm constraint for collision avoidance in multiple predicted steps by either formulating a chance constraint, as mentioned in Remark \ref{rem:cc}, or a robust constraint. However, this can result in infeasibility of the optimization problem, particularly if the uncertainty support varies over time. Despite only considering the norm constraint for the next predicted step, it is still beneficial to use an \gls{MPC} horizon $N > 1$. Other objectives are optimized over the entire horizon, given that the first input is included in the set \gls{UUopti0}, which potentially consists of multiple admissible inputs that all minimize the constraint violation for the next step. 

Applying the CVPM-\gls{MPC} approach possibly results in oscillating behavior. As long as case 1 is valid, the proposed method does not affect the optimization, as $\gls{UUopti0} = \gls{UUxc0}$. Once case 2 is active, a solution is found which minimizes the probability of constraint violation, ignoring the reference and potentially moving from the reference, as only one input is admissible. When case 1 is valid again, the optimized the reference is tracked again until, possibly, case 2 becomes active again. This can be improved by considering the norm constraint as a chance constraint for multiple predicted steps, however, recursive feasibility is not guaranteed, as mentioned before.

The main focus of the suggested method is to minimize the constraint violation probability. It is clear that stability cannot always be guaranteed, as the origin can be excluded from the admissible state set. We consider a narrow road where a bicycle is between the controlled vehicle and the vehicle reference point. If the road is too narrow for the vehicle to pass, it will remain behind the bicycle and never reach the reference point, i.e., Assumption \ref{ass:yr_kinfty2} (b) is violated. However, Assumption~\ref{ass:yr_kinfty2} implies that the origin is not inadmissible at all times, and once the origin is admissible, the controlled system converges. 

It is also important to note that minimizing the constraint violation probability has priority over other optimization objectives. Especially in safety-critical applications this can be of major interest, e.g., an autonomous car must ensure that the collision probability is always minimal, prior to reducing energy or increasing passenger comfort. If \gls{SMPC} were to be applied in such scenarios, the question would arise of how to choose the \gls{SMPC} risk parameter $\beta_k$. A large $\beta_k$ yields efficient behavior but might be unacceptable due to an insufficient safety level. Finding a reduced value for $\beta_k$ in \gls{SMPC} is challenging, as even very small risk parameters allow for constraint violations, while $\beta_k = 0$ does not yield a chance constraint and the advantages of \gls{SMPC} are lost. In the proposed CVPM-\gls{MPC} the task of appropriately choosing the risk parameter is not required.

For the approach in Section \ref{sec:onestepproblem} the \gls{pdf} \gls{fPw} does not need to be known exactly as long as it fulfills Assumption \ref{ass:Pcv}. If \gls{fPw} is symmetric and unimodal, it is ensured that increasing \gls{Mkb} results in a lower constraint violation probability \gls{Pcv}.

The proposed method is especially useful in collision avoidance applications, which are either in $2$- or $3$-dimensional space. While applying the proposed method in $2$D is straightforward, $3$-dimensional applications can be more challenging to solve, especially finding \gls{pp} in \eqref{eq:Uopt_th}.

~

\section{Simulation Results}
\label{sec:results}

In the following a simulation is presented and discussed to further explain the general idea and its application. This collision avoidance scenario with two vehicles illustrates an application where the proposed method is beneficial. The simulations were run in MATLAB on a standard desktop computer using MPT3 \cite{HercegEtalMorari2013} and YALMIP \cite{Lofberg2004}. Solving a single optimization of the \gls{MPC} algorithm takes $\SI{54}{\milli\second}$ on average. All quantities are given in SI units.

\subsection{Collision Avoidance Simulation}

A collision occurs if the distance between two objects becomes too small. This distance can be represented by a norm constraint. The priority is then enforcing the norm constraint, or if not feasible, minimizing the probability of violating the norm constraint.

We consider the example mentioned in Sections \ref{sec:introduction} and \ref{sec:problem} where a controlled vehicle avoids collision with a bicycle, referred to as obstacle in the following. The controlled vehicle is approximated by the radius $\gls{rcv} = 2.0$ and the obstacle is approximated by the radius $\gls{rr} = 0.8$ and is subject to stochastic motion in a bounded area, e.g., a road. The circles are chosen to fully cover the individual shapes of the controlled vehicle and obstacle. The scenario setup is shown in Figure~\ref{fig:scenario}.
\begin{figure}
\centerline{\includegraphics[trim=0 0 0 0,width=0.5\columnwidth]{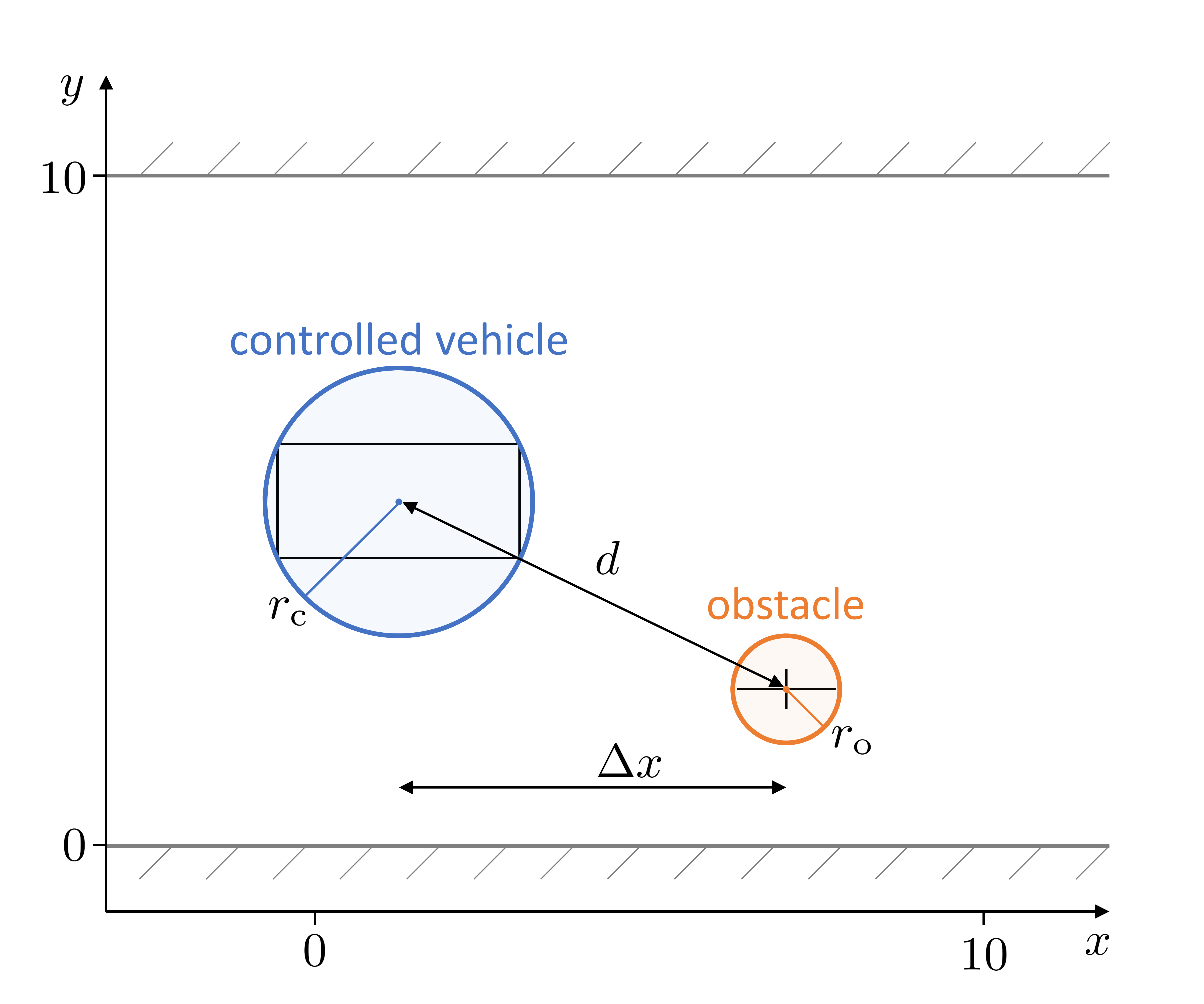}}
\caption{Vehicle avoidance scenario. Approximated shapes of the controlled vehicle (car) and obstacle (bicycle) are indicated by black lines within the objects.}
\label{fig:scenario}
\end{figure}
The continuous system dynamics of the controlled vehicle in $\gls{simux}$ and $\gls{simuy}$ direction are given by 
\begin{IEEEeqnarray}{c}
\dot{\bm{x}} = \begin{bmatrix} 
\dot{x}\\ \dot{y}
\end{bmatrix}
= \begin{bmatrix}
v_{\gls{simux}}\\v_{\gls{simuy}}
\end{bmatrix}
=\begin{bmatrix}
u_{1} \\ u_{2}
\end{bmatrix} \IEEEeqnarraynumspace
\end{IEEEeqnarray}
where $\bm{x} = [\gls{simux}, \gls{simuy}]^\top$ and $[v_{\gls{simux}}, v_{\gls{simuy}}]^\top$ are the position and velocity in a 2D environment, respectively. The inputs are given by $[u_{1}, u_{2}]^\top$. Using zero-order hold with sampling time $\Delta t = 0.1$ yields the discretized system given by \eqref{eq:system} with
\begin{IEEEeqnarray}{c}
\bm{A}=\begin{bmatrix}
1 & 0 \\ 0 & 1 
\end{bmatrix},~~
\bm{B}=\begin{bmatrix}
e^{\Delta t}-1 & 0 \\ 0 & e^{\Delta t}-1 \
\end{bmatrix},~~
\bm{C}=\begin{bmatrix}
1 & 0 \\ 0 & 1
\end{bmatrix}. \IEEEeqnarraynumspace
\end{IEEEeqnarray}
We will consider the input constraints
\begin{IEEEeqnarray}{c}
\label{eq:simu2u}
\gls{UUi} = \setdef[\bm{u} = \begin{bmatrix}
v_{\gls{simux}}\\ v_{\gls{simuy}}
\end{bmatrix}]{1 \leq v_{\gls{simux}} \leq 9,~|v_{\gls{simuy}}| \leq 3.5}. \IEEEeqnarraynumspace
\end{IEEEeqnarray}
In $\gls{simux}$-direction there exists a minimum velocity $v_{\gls{simux},\text{min}} = 1$ to ensure that the controlled vehicle is always moving forward, which also limits the potential oscillating behavior due to the CVPM-\gls{MPC} approach. We also consider the state constraint
\begin{IEEEeqnarray}{c}
\label{eq:simu2x}
\gls{XX} = \setdef[\bm{x} = \begin{bmatrix}
\gls{simux}\\\gls{simuy}
\end{bmatrix}]{y_{\text{lb}}\leq \gls{simuy} \leq y_{\text{ub}}} \IEEEeqnarraynumspace
\end{IEEEeqnarray}
where $y_{\text{lb}}= 2.0$ and $y_{\text{ub}} = 8.0$ are the boundaries of the road minus the radius \gls{rcv}.

The behavior of the obstacle with random behavior is given by 
\begin{IEEEeqnarray}{c}
\gls{yri}  = \bm{y}_{\text{r},0} + \sum_{i=0}^{k-1} \lr{\bm{u}_{\text{r},i} + \bm{w}_i}
\end{IEEEeqnarray}
depending on the initial output $\bm{y}_{\text{r},0}$, the input $\bm{u}_{\text{r},k}$, and the realization $\gls{wi}$ of the random variable $\gls{Wi} \sim \gls{fPw}$  and $\bm{y}_\text{r} = [\gls{simuxr},\gls{simuyr} ]^\top$. We assume $\gls{fPw}$ to be symmetric, unimodal, and truncated, resulting in the support of $\gls{fPw}$
 \begin{IEEEeqnarray}{c}
 \supp\lr{\gls{fPw}} = \setdef[\bm{w}_k]{\norm{\bm{w}_k}_2 \leq \gls{wimax}}
\end{IEEEeqnarray}
where \gls{wimax} is the radius of the support boundary of \gls{Wi}. The physical interpretation of \gls{wimax} is that it is the maximum uncertain distance the obstacle can move in one step, additionally to the deterministic distance \gls{uri}. At step $k$ the controlled vehicle knows the obstacle position \gls{yri} and deterministic input $\bm{u}_{\text{r},k}$, but \gls{wi} is unknown. 

As the main aim of this simulation is to minimize the collision probability, an expression for this probability is necessary in order to analyze the simulation results. The collision probability at step $k$ between the two vehicles will be denoted by \gls{Pcoli} and it has finite support as \gls{fPw} is truncated. In this example a norm constraint is used to avoid a collision, i.e., the norm constraint violation probability is minimized. Therefore, the probability of a collision \gls{Pcoli} is defined analogous to \gls{Pcv} in Section \ref{sec:problem}. The derivation and expression for the collision probability \gls{Pcoli} is omitted here due to readability. Details can be found in Appendix~\ref{sec:appendixa}.

The collision probability \gls{Pcoli} depends on the Euclidean distance
 \begin{IEEEeqnarray}{c}
\gls{d} = \norm{\gls{yi}-\overline{\bm{y}}_{\text{r},k}}_2
\end{IEEEeqnarray}
between the controlled vehicle and obstacle. Similar to \eqref{eq:hc2} a norm-constraint can be formulated where $c_k=\gls{rti}$ can be interpreted as the minimal distance between the controlled vehicle and the obstacle such that a collision is avoided. The support of \gls{Pcoli} results from adding the radius of the controlled vehicle and the obstacle to $\supp\lr{\gls{fPw}}$, i.e.,
\begin{IEEEeqnarray}{c}
\supp\lr{\gls{Pcoli}} = \setdef[\gls{yi}]{\gls{d} \leq \gls{rti}} 
\end{IEEEeqnarray}
where $\gls{rti} = w_{\text{max},k-1}+\gls{rr}+\gls{rcv}$ is the safety distance required to avoid a collision between the controlled vehicle and the obstacle, taking into account the radius of both vehicles, $\gls{rr}$ and $\gls{rcv}$, and the maximal obstacle step $w_{\text{max},k-1}$. Similar to Lemma~\ref{lem:Pcv} for \gls{Pcv}, \gls{Pcoli} is decreasing for increasing \gls{d}.

We choose $h(\xi) = \xi^2$, which is strictly monotonically increasing with $\xi$. This yields 
\begin{IEEEeqnarray}{c}
\gls{hykb}= \gls{Mkm1b}^2\IEEEeqnarraynumspace
\end{IEEEeqnarray}
which can be considered a substitution of the probability function \gls{Pcoli}. 

The controlled vehicle uses the CVPM-\gls{MPC} algorithm \eqref{eq:mpc_new} with $N=10$ and
\begin{IEEEeqnarray}{c}
\bm{Q}=\begin{bmatrix}
1 & 0  \\ 0 & 1
\end{bmatrix},~~
\bm{R}=\begin{bmatrix}
0.1 & 0 \\ 0 & 0.1
\end{bmatrix}.
\end{IEEEeqnarray}
The $x$-position references for the controlled vehicle are obtained by $x_{\text{ref},k}=x_0 + v_{x,\text{ref}} k \Delta t$, where $v_{x,\text{ref}}$ is the reference velocity in $x$-direction.

In the following two scenarios will be analyzed. In the first scenario the controlled vehicle is located close to its state boundary, i.e., the road boundary, showing that the norm constraint can be minimized in the presence of state constraints. In the second scenario the obstacle uncertainty support will suddenly increase. The orientation $\phi$ of the controlled vehicle in Figures \ref{fig:simu1run} and \ref{fig:simu2run} is approximated by
\begin{IEEEeqnarray}{c}
\phi = \arctan \frac{u_2}{u_1}.
\end{IEEEeqnarray}

\subsubsection{Active State Constraint}
In the first simulation it is shown that the proposed method is applicable if state constraints are active. The reference velocity and $y$-position for the controlled vehicle are set to $v_{\gls{simux},\text{ref}} = 5.0$ and $y_{\text{ref}}=8.0$, respectively, with initial position $\bm{y}_0 = [0,4]^\top$. The obstacle motion consists of a deterministic part $\gls{uri} = [0.5,~0]^\top$ combined with random behavior subject to $\gls{wimax} = 0.15$ and mean $\gls{simuy}$-position $y_{\text{r}} = 4.0$. Therefore, the $x$-position reference of the controlled vehicle is the same as the $x$-position of the obstacle in every step. 

The vehicle configurations at different time steps are shown in Figure \ref{fig:simu1run} and the results of the simulation are displayed in Figure \ref{fig:simu1}.  
\begin{figure}
\centering
 \begin{minipage}[b]{0.45\textwidth}
 \includegraphics[trim=20 50 20 10,width=\linewidth]{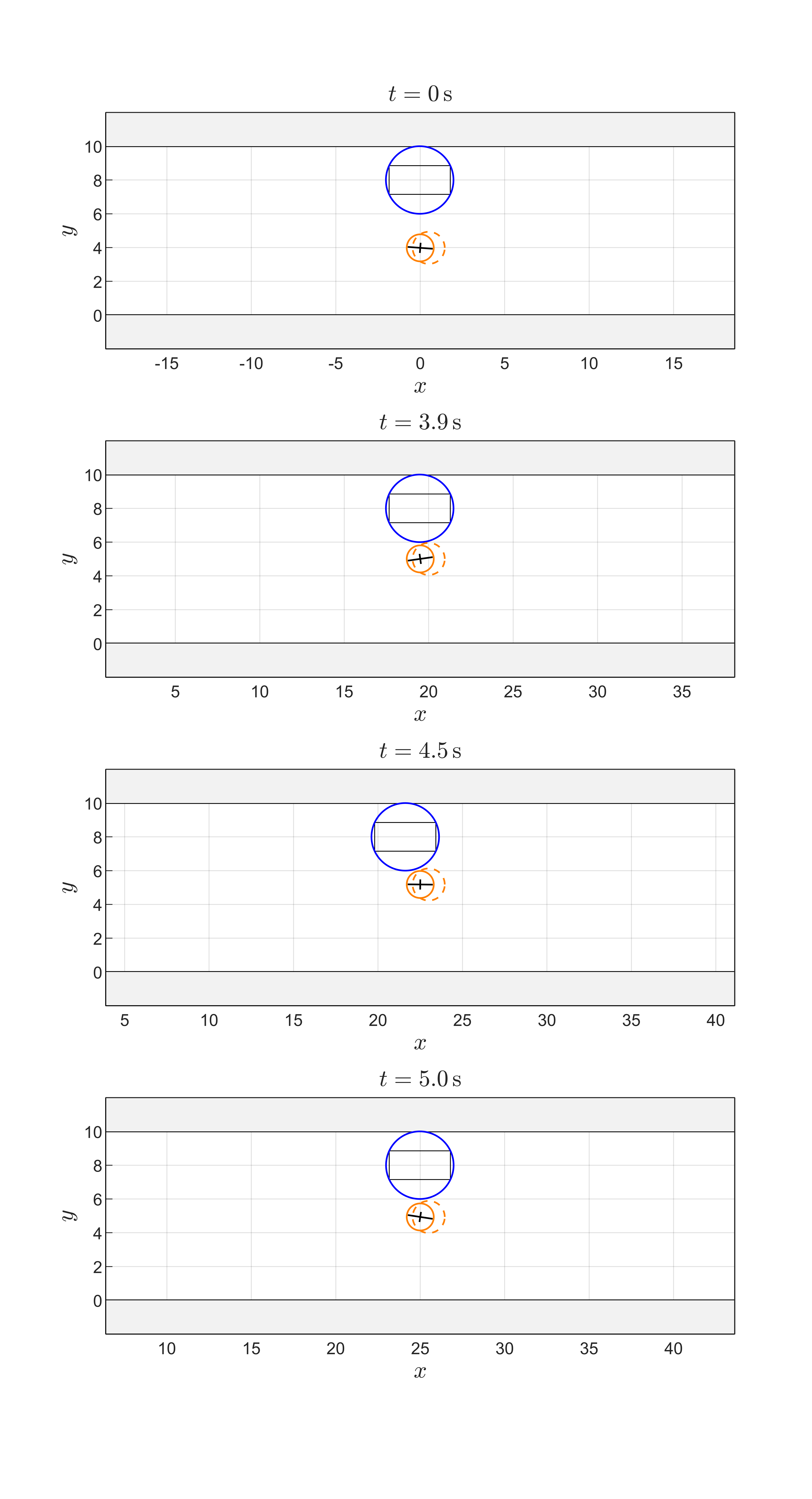}
\caption{Vehicle configurations for the simulation with active state constraints. The controlled vehicle boundary is shown as a solid blue line and the obstacle boundary is a solid orange line. The dashed orange circle represents the possible obstacle location at the next time step.}
\label{fig:simu1run}
\end{minipage}
\hfill
\begin{minipage}[b]{0.45\textwidth}
\includegraphics[trim=20 50 20 10,width=\linewidth]{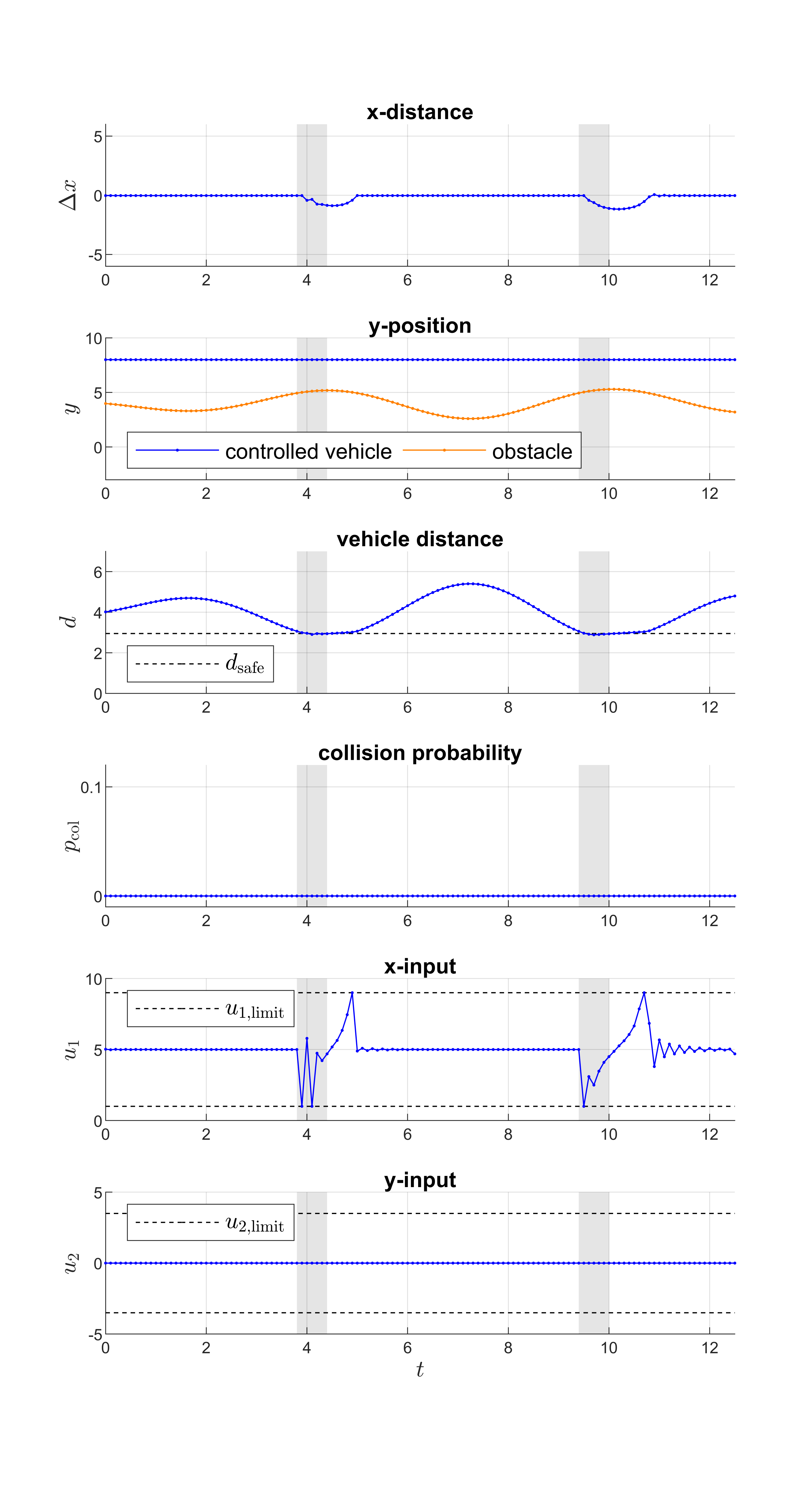}
\caption{Simulation results for the simulation with active state constraint. The controlled vehicle is close to the state constraint. The gray area denotes actions by the controlled vehicle to avoid collision. The collision probability remains 0.}
\label{fig:simu1}
\end{minipage}
\end{figure}
Initially, the controlled vehicle and obstacle have the same $\gls{simux}$-position. Starting at $t=\SI{3.9}{\second}$ the controlled vehicle needs to slow down to maintain a safe distance to the obstacle. As the maximal obstacle uncertainty is known by the controlled vehicle, the collision probability is kept at zero. After $t=\SI{4.5}{\second}$ the obstacle moves away from the controlled vehicle, resulting in increased input $u_1$ in order to get closer to the $x$-position reference. At $t=\SI{5.0}{\second}$ the controlled vehicle catches up with its $x$-position reference, which is then followed by constant inputs. Between $t=\SI{9.0}{\second}$ and $t = \SI{11.0}{\second}$ similar behavior can be observed. It can be seen that the CVPM-\gls{MPC} ensures $\gls{Pcv}=0$ with active state constraints. As mentioned in Section \ref{sec:properties}, the motion of the obstacle can result in an inadmissible origin, i.e., Assumption \ref{ass:yr_kinfty2} (c) is violated and the controlled vehicle cannot keep its reference velocity. However, as shown in Theorem \ref{t:convergence}, once the obstacle moves away the velocity of the controlled vehicle again reaches the reference velocity.

\subsubsection{Change of Uncertainty Support}
In the second simulation we show that the proposed method is capable of dealing with varying uncertainty support of the obstacle. The controlled vehicle aims to obtain the reference velocity $v_{\gls{simux},\text{ref}} = 4.0$ while maintaining $y_{\text{ref}} = 4.0$ with the initial position $\bm{y}_0 = [0,4]^\top$. The obstacle moves with a constant input $\gls{uri} = [0.25,~0]^\top$ at $\gls{simuyr} = 3.0$. We consider here that the obstacle uncertainty support suddenly changes, for example due to a changing environment. At first the expected uncertainty support is $\gls{wimax} = 0.15$ and at $t = \SI{3.0}{\second}$ it changes to $\gls{wimax} = 0.9$, while returning to $\gls{wimax} = 0.15$ at $t = \SI{5.0}{\second}$. In the simulation the obstacle does not move randomly, which helps to better understand the action of the controlled vehicle once the uncertainty support changes. At each time step the controlled vehicle knows the current uncertainty support of the obstacle.

The vehicle configurations at different time steps are shown in Figure \ref{fig:simu2run} and the results of the simulation are displayed in Figure~\ref{fig:simu2}. 
\begin{figure}
\centering
 \begin{minipage}[b]{0.45\textwidth}
 \includegraphics[trim=20 50 20 10,width=\linewidth]{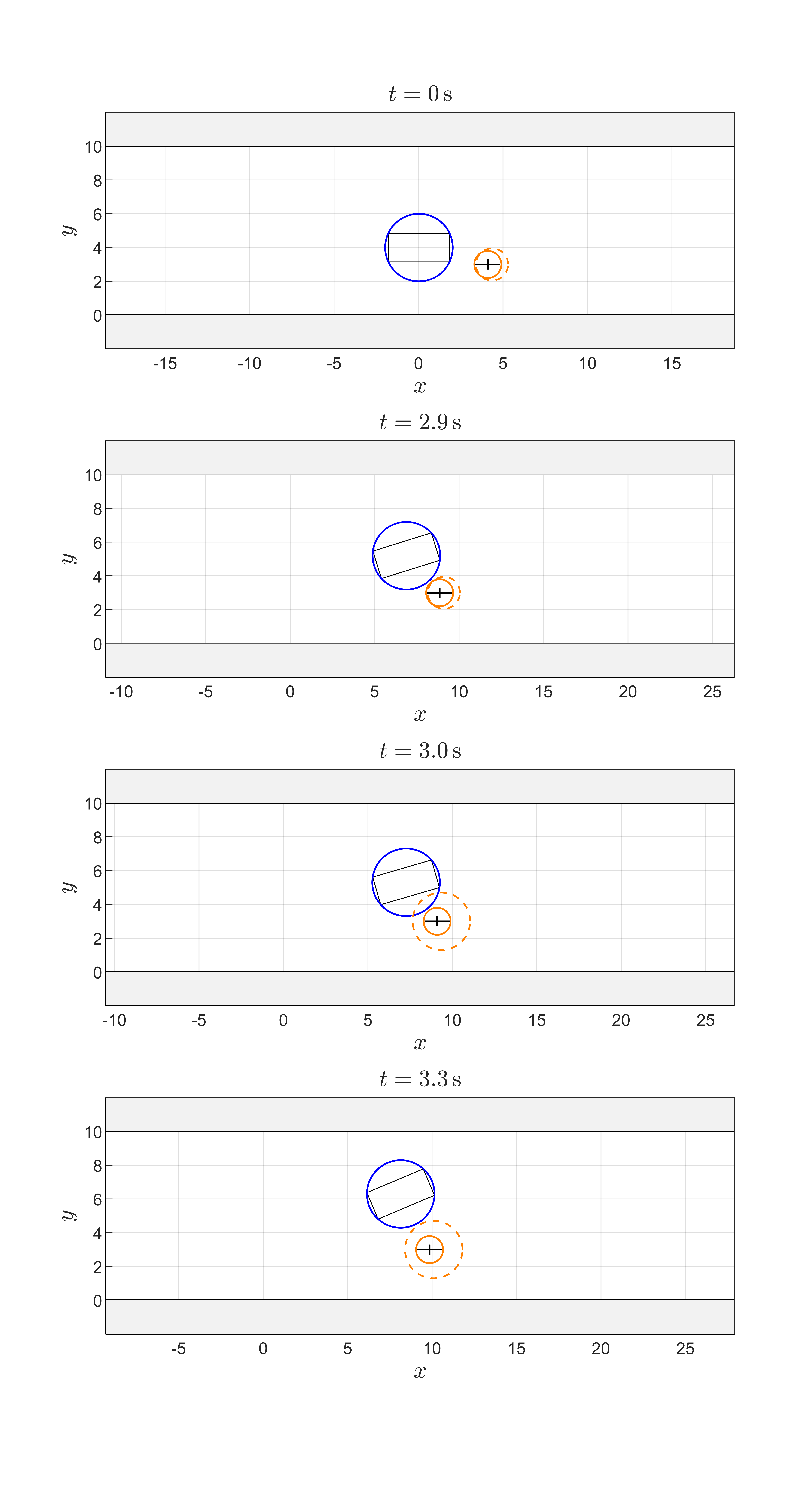}
\caption{Vehicle configurations for the simulation with changing uncertainty support. The controlled vehicle and obstacle boundaries are shown as solid blue and orange lines, respectively. The dashed orange circle represents the possible obstacle location at the next time step.}
\label{fig:simu2run}
\end{minipage}
\hfill
\begin{minipage}[b]{0.45\textwidth}
\includegraphics[trim=20 50 20 10,width=\linewidth]{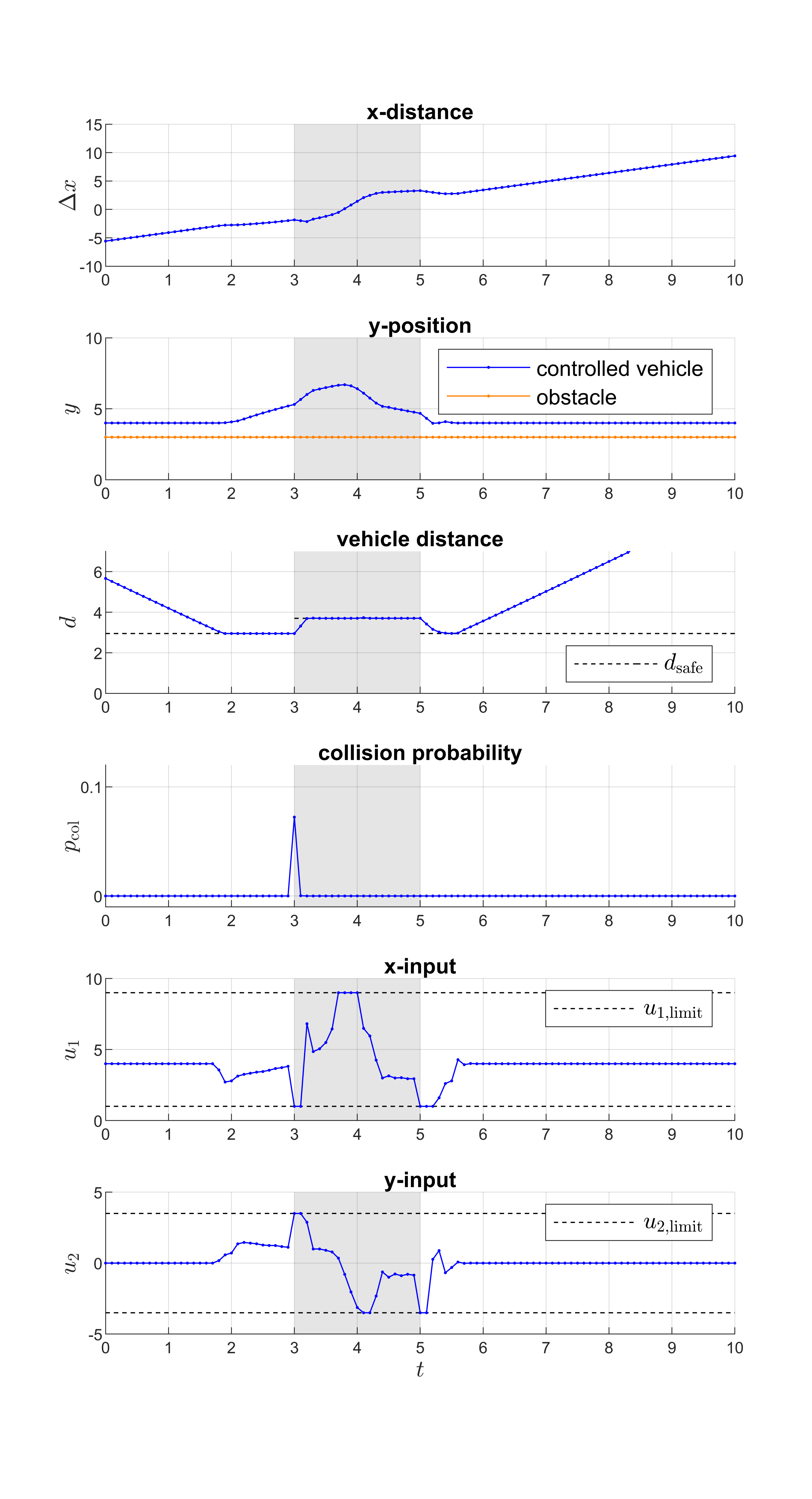}
\caption{Simulation results for the simulation with changing uncertainty support. The gray area represents a higher uncertainty support. Once the uncertainty support changes, the collision probability temporarily increases to the minimal level possible.}
\label{fig:simu2}
\end{minipage}
\end{figure}
As the controlled vehicle has a higher velocity it will eventually pass the obstacle, therefore, the distance $\Delta \gls{simux}~=~\gls{simux} - \gls{simuxr}$ turns positive. At $t=\SI{1.8}{\second}$ the controlled vehicle gets close enough to the obstacle that the controlled vehicle moves away from $y_{\text{ref}}$ to maintain $v_{\gls{simux},\text{ref}}$ and ensures that the distance $\gls{d} = \norm{\gls{yi}-\gls{yri}}_2 \geq \gls{rti}$. As $\gls{wimax}$ increases at $t=3.0$, so does the required distance between the controlled vehicle and obstacle, causing the controlled vehicle to move further away from $y_{\text{ref}}$. Due to input limitations the controlled vehicle cannot move fast enough. This results in $\gls{d} < \gls{rti}$, i.e., $\gls{Pcoli} > 0$ at $t=\SI{3.0}{\second}$, i.e., there is a probability of collision for the next time step. However, $\gls{d}$ is increased to a maximal level, given $\gls{ui} \in \gls{UUxc}$, resulting in a minimal constraint violation probability \gls{Pcoli}. Once the distance satisfies $\gls{d} \geq \gls{rti}$ at $t=\SI{3.2}{\second}$, \gls{Pcoli} becomes zero, and the controlled vehicle moves along the obstacle boundary for the next step, as seen for $t=\SI{3.3}{\second}$. At $t=\SI{5.0}{\second}$, $\gls{wimax}$ decreases, and the controlled vehicle converges to $y_{\text{ref}}$ at $t=\SI{5.8}{\second}$. 

In order to validate the probability of constraint violation, the simulation was run 2000 times with an arbitrary random obstacle step at $t = \SI{3.0}{\second}$, which is the first step with the increased uncertainty bound $\gls{wimax} = 0.9$. The vehicles collided in 144 simulations, yielding a collision probability of $0.072$ compared to the calculated collision probability $0.0723$, as described in Appendix \ref{sec:appendixc}.

\subsubsection{Comparison to \gls{RMPC} and \gls{SMPC}}
If \gls{RMPC} and \gls{SMPC} were applied in the simulations, certain problems would arise, mainly due to infeasibility of the optimization problem. This could be solved by providing rigorous alternative optimization problems, predefined alternative inputs, or highly conservative worst-case considerations. However, there is no ideal \gls{RMPC} or \gls{SMPC} approach to deal with the scenario in the simulation. Therefore, we will compare the simulation results of the proposed 
method to \gls{RMPC} and \gls{SMPC} only qualitatively. 

We will first consider the behavior with \gls{RMPC} applied to the controlled vehicle. In the first simulation \gls{RMPC} would deliver safe results similar to the CVPM-\gls{MPC} method, while remaining behind the obstacle in order to account for the worst-case obstacle behavior. In the second simulation two cases can be distinguished. If the initially considered uncertainty support is $\gls{wimax} = 0.15$, the behavior would be similar to the proposed method until the support changes. As it is impossible to find a state with zero collision probability after the uncertainty support is altered, the \gls{RMPC} optimization problem becomes infeasible. If the considered uncertainty support is initially chosen such that the larger support after $t = \SI{3.0}{\second}$ is covered, \gls{RMPC} yields a safe solution, however, it is passing the obstacle at a larger distance than initially required, yielding a higher cost compared to the proposed CVPM-\gls{MPC} method. In many applications it is also difficult to choose the worst-case uncertainty support a priori, as higher supports might occur later, resulting in even more conservative \gls{RMPC} solutions.

It is now assumed that the controlled vehicle is controlled using \gls{SMPC} with a chance constraint with risk parameter $\beta_k > 0$ for collision avoidance. In the second simulation, before the uncertainty support changes, the controlled vehicle would pass the obstacle a little closer than with the proposed CVPM-\gls{MPC} method, as the chance constraint allows for small constraint violations. The larger $\beta_k$ is chosen, the smaller the distance. However, while the proposed CVPM-\gls{MPC} method ensures safety while only passing the vehicle with little more distance, the \gls{SMPC} approach would pass the obstacle `on the chance constraint', i.e., as close as $\beta_k$ allows, sacrificing guaranteed safety for small cost improvements. In other words, leaving slightly more space between the controlled vehicle and the obstacle would result in $\gls{Pcv} = 0$ with only little higher cost. 

When the uncertainty support changes, the \gls{SMPC} solution is as close to the obstacle as $\beta_k$ allowed in the previous step. The chance constraint cannot be met anymore because the uncertainty support increased, resulting in a constraint violation probability larger than allowed by $\beta_k$. The \gls{SMPC} optimization problem then becomes infeasible, requiring an alternative optimization problem to be defined beforehand. In the first simulation a similar situation would occur. If the chance constraint allows the controlled vehicle to be in a position which will yield $\gls{Pcoli}  > \beta_k$ due to the unconsidered worst-case obstacle motion, this leads to infeasibility of the optimization problem. 

Considering the qualitative comparison, we can see that the proposed method offers certain advantages over \gls{RMPC} and \gls{SMPC}, especially guaranteeing recursive feasibility of the optimization problem in the presence of a changing uncertainty support.

\section{Conclusion}
\label{sec:conclusion}

The proposed CVPM-\gls{MPC} algorithm yields a minimal violation probability for a norm constraint for the next step, while also optimizing further objectives and satisfying state and input constraints. Recursive feasibility and, under certain assumptions, convergence to the origin is guaranteed. While the suggested method is inspired by \gls{RMPC} and \gls{SMPC}, it provides feasible and efficient solutions in scenarios where \gls{RMPC} and \gls{SMPC} encounter difficulties or are not applicable.

As norm constraints are especially useful in collision avoidance applications, the advantages of the presented CVPM-\gls{MPC} method can be exploited in applications such as autonomous vehicles or robots that work in shared environments with humans. A brief example is introduced where a controlled vehicle is overtaking a bicycle while minimizing the collision probability. Here, we focus on minimizing the constraint violation for a norm constraint. However, depending on the application, a multi-step CVPM-\gls{MPC} could be beneficial. Especially for collision avoidance it is also of interest not only to focus on the collision probability but to consider the severity of collision if a collision is inevitable.

\appendices

\section{Minimal Constraint Violation Probability for the Multi-Step Problem}
\label{sec:appendixb}

The method presented in Section \ref{sec:method} minimizes the constraint violation probability for the next step. In the following a possible extension of the one-step CVPM-\gls{MPC} method is shown. Considering multiple steps $\gls{n} > 1$ yields a method closer related to \gls{RMPC}, as it provides advantages with respect to robustness but conservatism is increased.

Considering the stochastic process $\lr{\gls{Wi}}_{k \in \mathbb{I}_{0:j-1}}$, its realization, a sequence $\lr{\gls{wi}}_{k \in \mathbb{I}_{0:j-1}}$ with $j \in \mathbb{N}_{\geq0}$, and the initially known output $\bm{y}_{\text{r},0}$ yields
\begin{IEEEeqnarray}{c}
\label{eq:yrsum}
\bm{y}_{\text{r},k} = \bm{y}_{\text{r},0} +\sum_{i=0}^{k-1} \bm{w}_i.
\end{IEEEeqnarray}

While in the one-step method \gls{Pcvj} only needs to be minimized for the next step $j=1$, for the \gls{n}-step approach \gls{Pcvj} needs to minimized for $1 \leq j \leq \gls{n}$. 

Similar to Section \ref{sec:method} we first address the general method and then provide a solution for \gls{fPw} satisfying Assumption \ref{ass:Pcv}.

\subsection{General Method to Minimize Constraint Violation Probability for Multi-Step Problem}
\label{sec:multistepproblem_gen}

It is necessary to find the set \gls{UUcvpml}, which represents the set of admissible input sequences $\gls{Ucvpml}~=~[\bm{u}_0, ..., \bm{u}_{\gls{n}-1}]^\top$ that minimize \gls{Pcvj} for $1 \leq j \leq \gls{n}$. In the following three cases are again considered. The constraint violation \gls{Pcvjp1} for step $j+1$ depends on the previous output $\bm{y}_{j}$, the input $\bm{u}_{j}$, and the uncertain output $\bm{y}_{\text{r},j}$.

\paragraph*{Case 1 (Guaranteed Constraint Satisfaction)} 

Constraint satisfaction is guaranteed for all steps $j \in \mathbb{I}_{0:l-1}$, i.e.,
\begin{IEEEeqnarray}{c}
\gls{Pcvjp1att} = 0 ~~~~ \forall \gls{uj}\in \gls{UUxcj},~j \in \mathbb{I}_{0:l-1}, \label{eq:l_case1g}
\end{IEEEeqnarray}
resulting in
\begin{IEEEeqnarray}{c}
\gls{UUcvpml} = \setdef[\gls{uj}]{\gls{uj}\in \gls{UUxcj},~j \in \mathbb{I}_{0:l-1}}.   \label{eq:l_ucase1g}
\end{IEEEeqnarray}

\paragraph*{Case 2 (Impossible Constraint Satisfaction Guarantee)}

For a $j$ with $0 \leq j \leq \gls{n}-1$, potentially at multiple steps $j$, constraint satisfaction cannot be guaranteed by any input $\gls{uj}\in \gls{UUxcj}$, i.e., $\gls{Pcvjp1}=0$. This can be expressed by 
\begin{IEEEeqnarray}{c}
\exists~j\in \mathbb{I}_{0:l-1}\text{~~s.t.~~} \gls{Pcvjp1att} > 0 ~~\forall~ \gls{uj}\in \gls{UUxcj}.  \label{eq:l_case2g} \IEEEeqnarraynumspace
\end{IEEEeqnarray}
The set of admissible inputs which minimize the constraint violation probability is then given by
\begin{IEEEeqnarray}{c}
\gls{UUcvpml} =\setdef[\gls{uj}]{\gls{uj} = \argmin{\gls{uj}\in \gls{UUxcj}} \gls{Pcvjp1att},~j \in \mathbb{I}_{0:l-1}}.  \label{eq:l_ucase2g}\IEEEeqnarraynumspace
\end{IEEEeqnarray}

\paragraph*{Case 3 (Possible Constraint Satisfaction Guarantee)}

At each step $0 \leq j \leq \gls{n}-1$ it is possible, but not guaranteed, that the norm constraint is satisfied for $j+1$, i.e.,
\begin{IEEEeqnarray}{c}
\exists~\gls{uj}\in \gls{UUxcj} \text{~~s.t.~~} \gls{Pcvjp1att} = 0~~\forall~j\in \mathbb{I}_{0:l-1}. \label{eq:l_case3g}
\end{IEEEeqnarray}
This yields
\begin{IEEEeqnarray}{c}
\gls{UUcvpml} =\nonumber\setdef[\gls{uj}]{\lr{\gls{Pcvjp1att} = 0} \land \lr{\gls{uj}\in \gls{UUxcj}},~j\in \mathbb{I}_{0:l-1}}.   \label{eq:u_case3g}\IEEEeqnarraynumspace
\end{IEEEeqnarray}

\subsection{Minimal Constraint Violation Probability for Multi-Step Problem with Radially Decreasing PDF}

After defining the general case, we now address the multi-step CVPM-\gls{MPC} method for a symmetric, unimodal \gls{pdf}. We make the following assumptions.

\begin{assumption}[Constant Minimal Norm Value]
The minimal norm value $\gls{Mminkj} = \gls{Mcon}$ is constant.
\end{assumption}

\begin{assumption}[Known Deterministic Input]
The deterministic input $\bm{u}_{\text{r},j}$ is known for $j \in \mathbb{I}_{0:\gls{n}-1}$.
\end{assumption}

A simple approach to find \gls{UUcvpml} is to maximize \gls{Mlm1b} with
\begin{IEEEeqnarray}{c}
\overline{\bm{y}}_{\text{r},\gls{n}} =\bm{y}_{\text{r},0} + \sum_{i=0}^{\gls{n}-1} \bm{u}_{\text{r},i},
\end{IEEEeqnarray}
as this automatically results in a maximization of $\norm{\bm{y}_{j} - \overline{\bm{y}}_{\text{r},j}}_2$ for $j \in \mathbb{I}_{0:\gls{n}-1}$ because \gls{Pcvj} is decreasing with increasing $\norm{\bm{y}_{j} - \overline{\bm{y}}_{\text{r},j}}_2$. Therefore, if $p_{\text{cv},l}$ is minimized, \gls{Pcvj} is also minimized for \gls{jl}. 

Similar to \eqref{eq:hmax} and \eqref{eq:hmin}, we define
\begin{IEEEeqnarray}{rl}
\gls{hmaxil} &=    \gls{gi} \lr{ \max_{\bm{u}_i \in \gls{UUxci},~i \in \mathbb{I}_{0:\gls{n}-1}} \lr{\gls{Mlm1b}}  }, \\
\gls{hminil} &=    \gls{gi} \lr{ \min_{\bm{u}_i \in \gls{UUxci},~i \in \mathbb{I}_{0:\gls{n}-1}} \lr{\gls{Mlm1b}}  }.
\end{IEEEeqnarray}

We now regard the three possible cases and determine \gls{UUcvpml}.

\paragraph*{Case 1 (Constraint Satisfaction Guarantee)} 
For any \gls{Ucvpml} it follows that $\gls{Pcvjp1}=0$ for \gls{jl}, i.e.,
\begin{IEEEeqnarray}{c}
\gls{hminil} \geq \gls{gi}\lr{\gls{Mtotli}}. \label{eq:case1l}
\end{IEEEeqnarray}
In comparison to \eqref{eq:case1}, \gls{wii} needs to be considered for $i \in \mathbb{I}_{0:\gls{n}-1}$. This yields
\begin{IEEEeqnarray}{c}
\gls{UUcvpml} = \setdef[\gls{Ucvpml}]{\gls{uj} \in \gls{UUxcj},~\gls{il}}. \label{eq:ucase1l}
\end{IEEEeqnarray}

\paragraph*{Case 2 (Impossible Constraint Satisfaction Guarantee)}
It is not possible to guarantee $\gls{Pcvjp1}=0$ for \gls{jl} as
\begin{IEEEeqnarray}{c}
\gls{hmaxil} < \gls{gi}\lr{\gls{Mtotli}}. \label{eq:case2l}
\end{IEEEeqnarray}
Therefore, \gls{UUcvpml} consists of the input sequence \gls{Ucvpml} which minimizes \gls{Pcvl}, resulting in
\begin{IEEEeqnarray}{c}
\gls{UUcvpml} =  \setdef[\gls{Ucvpml}]{\gls{Ucvpml} = \argmax{\gls{uii} \in \gls{UUxci},~\gls{il}} \gls{hi} \lr{\gls{Mlm1b}}}. \label{eq:ucase2l}\IEEEeqnarraynumspace
\end{IEEEeqnarray}

\paragraph*{Case 3 (Possible Constraint Satisfaction Guarantee)}
There are possible input sequences \gls{Ucvpml} such that $\gls{Pcvl}=0$. Similar to case 3 for the one-step CVPM-\gls{MPC}, we again need to find a set \gls{UUcvpml} which only allows input sequences \gls{Ucvpml} that result in constraint satisfaction of \eqref{eq:hc2} for $j \in \mathbb{I}_{1:\gls{n}}$. This is achieved by choosing
\begin{IEEEeqnarray}{c}
\gls{UUcvpml} = \Bigg\{ \gls{Ucvpml}~ \Bigg|~ \gls{hi}\lr{\gls{Mlm1b}}\geq \gls{gi}\lr{\gls{Mtotli}} \cap \gls{uj} \in \gls{UUxcj},~\gls{jl} \Bigg\}
\label{eq:Uopt_c3l} \IEEEeqnarraynumspace
\end{IEEEeqnarray}
where the approximation \gls{apUUcvpml} can be found analogously to Proposition \ref{prop:apUUopt}.\\

The \gls{n}-step CVPM-\gls{MPC} algorithm can then be formulated as in \eqref{eq:mpc_new} with
\begin{IEEEeqnarray}{c}
\gls{UU_total0} = \setdef[\gls{uNmpc}]{\gls{Ucvpml} \in \gls{UUcvpml}  \wedge \gls{uj} \in \mathbb{U}_{\bm{x},j},~j \in \mathbb{I}_{\gls{n}:N-1}}. \IEEEeqnarraynumspace \label{eq:Ustarl}
\end{IEEEeqnarray}

\section{Inequality Derivation}
\label{sec:appendixa}

In the following it is shown that $\norm{\gls{yj} - \gls{yrj}}_2 \geq \gls{Mminj}$ holds if $\norm{\gls{yj} - \overline{\bm{y}}_{\text{r},j}}_2   \geq \gls{Mminj} + w_{\text{max},j-1}$.
From \eqref{eq:unc_system} it follows that
\begin{IEEEeqnarray}{c}
\label{eq:ineq_trafo_a1}
\norm{\gls{yj} - \gls{yrj}}_2 = \norm{\gls{yj} - \lr{\overline{\bm{y}}_{\text{r},j} + \bm{w}_{j-1}}}_2 \geq \gls{Mminj}.  \IEEEeqnarraynumspace
\end{IEEEeqnarray}
Using the reverse triangle inequality yields
\begin{IEEEeqnarray}{c}
\norm{\gls{yj} - \overline{\bm{y}}_{\text{r},j} - \bm{w}_{j-1}}_2 \geq \vertlr{\norm{\gls{yj} - \overline{\bm{y}}_{\text{r},j}}_2 -\norm{ \bm{w}_{j-1}}_2}  \geq \gls{Mminj} \IEEEeqnarraynumspace
\end{IEEEeqnarray}
with
\begin{IEEEeqnarray}{c}
\vertlr{\norm{\gls{yj} - \overline{\bm{y}}_{\text{r},j}}_2 -\norm{ \bm{w}_{j-1}}_2}  \geq \norm{\gls{yj} - \overline{\bm{y}}_{\text{r},j}}_2 -\norm{ \bm{w}_{j-1}}_2. \IEEEeqnarraynumspace
\end{IEEEeqnarray}
Given \eqref{eq:supp_pw} it follows that
\begin{IEEEeqnarray}{c}
\norm{\gls{yj} - \overline{\bm{y}}_{\text{r},j}}_2 -\norm{ \bm{w}_{j-1}}_2  \geq \norm{\gls{yj} - \overline{\bm{y}}_{\text{r},j}}_2 - w_{\text{max},j-1}  \IEEEeqnarraynumspace
\end{IEEEeqnarray}
for all $\norm{\bm{w}_{j-1}}_2 \leq w_{\text{max},j-1}$. Therefore, if
\begin{IEEEeqnarray}{c}
\norm{\gls{yj} - \overline{\bm{y}}_{\text{r},j}}_2 - w_{\text{max},j-1}   \geq \gls{Mminj}, \IEEEeqnarraynumspace
\end{IEEEeqnarray}
is fulfilled, \eqref{eq:ineq_trafo_a1} holds, i.e.,
\begin{IEEEeqnarray}{c}
\norm{\gls{yj} - \overline{\bm{y}}_{\text{r},j}}_2   \geq \gls{Mminj} + w_{\text{max},j-1} \Rightarrow \norm{\gls{yj} - \gls{yrj}}_2 \geq \gls{Mminj}. \IEEEeqnarraynumspace
\end{IEEEeqnarray}
Equation \eqref{eq:ineq_trafo} in Section \ref{sec:onestepproblem} is obtained for $j=1$.

\section{Collision Probability Function}
\label{sec:appendixc}

Here the collision probability \gls{Pcoli} is described in detail, which is only needed for the evaluation of the simulation but not the proposed method. The \gls{pdf} \gls{fPw} is chosen to be
\begin{IEEEeqnarray}{c}
\gls{fPw}\lr{r_k}=
\begin{cases}
\frac{1}{\sigma z \sqrt{2\pi}} e^{-\frac{r_k^2}{2\sigma^2}} &\text{if } 0 \leq r_k \leq \gls{wimax},\\
0 &\text{otherwise}
\end{cases}\IEEEeqnarraynumspace
\end{IEEEeqnarray}
where $r_k$ is used instead of \gls{wi} and 
\begin{IEEEeqnarray}{c}
\supp\lr{\gls{fPw}}= \setdef[r_k]{0 \leq r_k \leq \gls{wimax}} 
\end{IEEEeqnarray}
with variance $\sigma = 1$ and
\begin{IEEEeqnarray}{rl}
z&=\Phi\lr{\gls{wimax}}-\Phi\lr{0},\\ 
\Phi\lr{r} &= 0.5\lr{1+\erf\lr{\frac{r}{\sqrt{2}}}},
\end{IEEEeqnarray}
such that
\begin{IEEEeqnarray}{c}
 \int\limits_{\supp\lr{\gls{fPw}}} \f{\gls{fPw}}{r_k} \mathrm{d} r_k =1.
\end{IEEEeqnarray}

As the main aim of this simulation is to minimize the constraint violation probability, i.e., the collision probability, an expression for this probability is necessary in order to analyze the simulation results. The controlled vehicle and the obstacle collide if their bodies overlap, i.e., $r_{\text{comb}} > \norm{\bm{y}_{k}-\bm{y}_{\text{r},k}}_2$ with the combined radius $r_{\text{comb}} = \gls{rcv}+\gls{rr}$. A collision at step $k$ is inevitable, if $\norm{\gls{yi}-\overline{\bm{y}}_{\text{r},k}}_2 + w_{\text{max},k-1}< r_{\text{comb}}$, i.e., even for the best-case $w_{\text{max},k-1}$ the objects will collide at step $k$. For $\norm{\gls{yi}-\overline{\bm{y}}_{\text{r},k}}_2- w_{\text{max},k-1} \geq r_{\text{comb}}$ it follows that $\gls{Pcoli} = 0$.

The collision probability is calculated according to the following procedure. We consider a circle where the radius is the required distance $r_{\text{comb}}$ and a circle with radius $r_k$. The intersection of both circles can be interpreted as the collision probability, by integrating the intersection area of both circles, weighted with $\gls{fPw}\lr{r_k}$. This is illustrated in Figure~\ref{fig:collcalc}. 
\begin{figure}
\centerline{\includegraphics[trim=0 0 0 0,width=0.5\columnwidth]{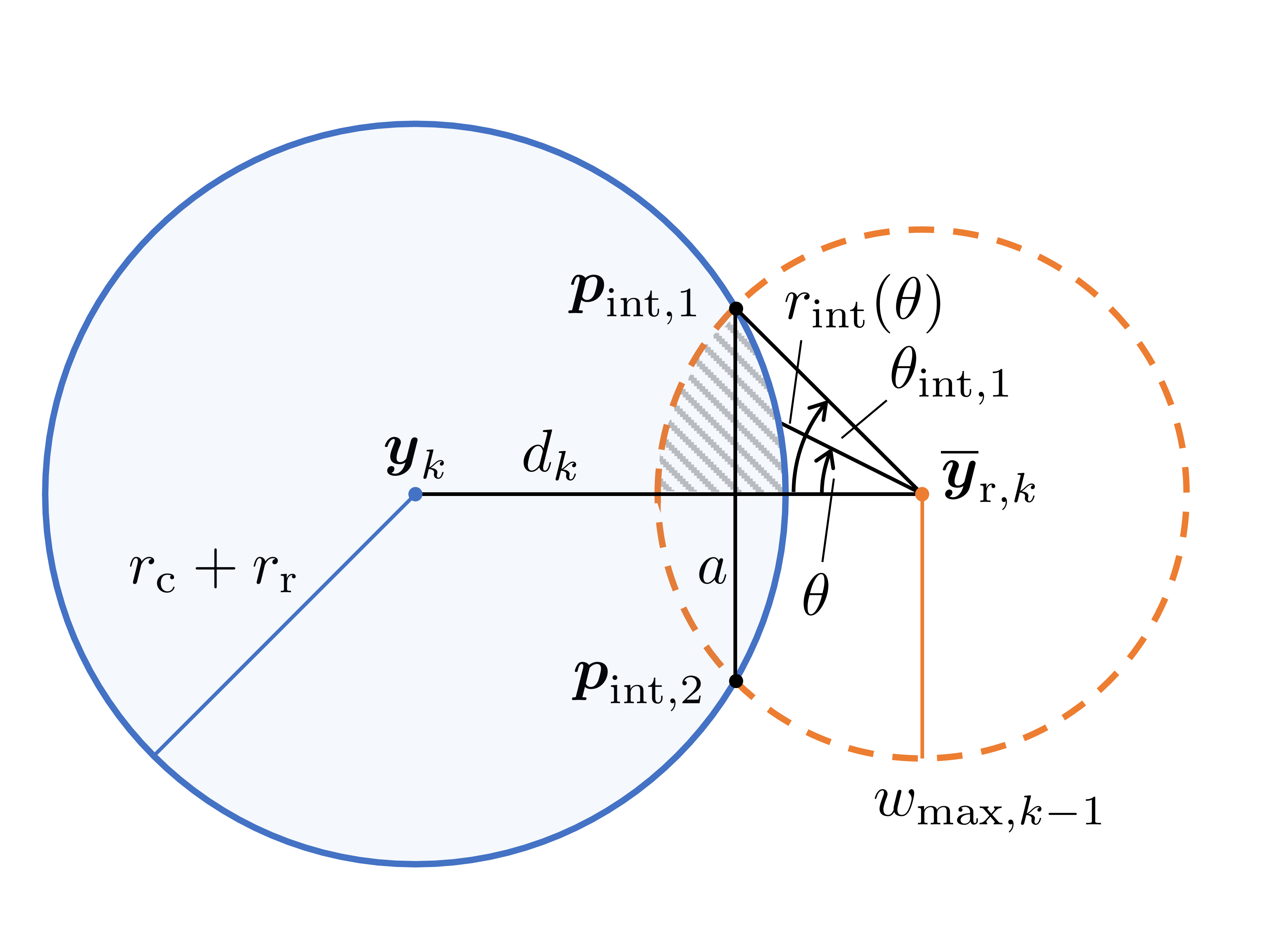}}
\caption{Collision probability calculation. The blue circle combines the radius of the controlled vehicle and the obstacle, the dashed orange circle represents the area potentially covered by the uncertainty. The striped area represents one half of the intersection between the two circles.}
\label{fig:collcalc}
\end{figure}
In case that there is no intersection area, then $\gls{Pcoli} = 0$. If an intersection exists, there are two intersection points. The intersection area is therefore bounded on one side by the arc with radius $r_{\text{comb}}$ and on the other side by the arc of the boundary of the uncertainty. As the intersection area is symmetric, it is sufficient to derive the calculation for one half, i.e., the area between the line connecting $\overline{\bm{y}}_{\text{r},k}$ and \gls{yi} and the intersection point $\bm{p}_{\text{int},1}$ as depicted by the striped area in Figure \ref{fig:collcalc}. This yields an angle $\theta_{\text{int},1} \in \left[0;0.5\pi\right]$ between the two lines connecting $\overline{\bm{y}}_{\text{r},k}$ and \gls{yi} as well as $\overline{\bm{y}}_{\text{r},k}$ and $\bm{p}_{\text{int},1}$. The distance $\f{r_{\text{int}}}{\theta}$ between $\overline{\bm{y}}_{\text{r},k}$ and the controlled vehicle boundary between the two intersection points follows from the law of cosines
\begin{IEEEeqnarray}{c}
\gls{rcomb}^2 = \f{r_{\text{int}}}{\theta}^2 + d_{k}^2 - 2 d_{k} \f{r_{\text{int}}}{\theta} \cos(\theta)
\end{IEEEeqnarray}
where $d_{k} = \norm{\gls{yi}-\overline{\bm{y}}_{\text{r},k}}_2$ and $\theta \in \left[0;\theta_{\text{int},1}\right]$ with
\begin{IEEEeqnarray}{rl}
\theta_{\text{int},1} &= \f{\sin^{-1}}{\frac{a}{2 \gls{wimaxm1}}}, \\
a &=\frac{\sqrt{4d_{k}^2 {\gls{wimaxm1}}^2-\left(d_{k}^2-\gls{rcomb}^2+{\gls{wimaxm1}^2}\right)^2}}{d_{k}}. \IEEEeqnarraynumspace
\end{IEEEeqnarray}
This yields
\begin{IEEEeqnarray}{c}
\f{r_{\text{int}}}{\theta} =  0.5 \lr{ 2 d_{k} \cos(\theta) - \sqrt{ \lr{2 d_{k} \cos(\theta)}^2 - 4 \lr{d_{k}^2 - \gls{rcomb}^2}} }. \IEEEeqnarraynumspace
\end{IEEEeqnarray}

The intersection area on both sides of the line between \gls{yi} and $\overline{\bm{y}}_{\text{r},k}$, weighted with the \gls{pdf} \gls{fPw2}, yields the collision probability
\begin{IEEEeqnarray}{c}
\label{eq:pcol}
\gls{Pcoli} = 2\int\displaylimits_{0}^{\theta_{\text{int},1}} \frac{1}{2\pi}  \int\displaylimits_{\f{r_{\text{int}}}{\theta}}^{\gls{wimaxm1}}  \f{\gls{fPw}}{r} \mathrm{d} r \mathrm{d} \theta
\end{IEEEeqnarray}
for $d_{k}+w_{\text{max},k-1} \geq r_{\text{comb}}$, depending on the angle $\theta_{\text{int},1}$.

This yields the overall collision probability
\begin{IEEEeqnarray}{c}
\gls{Pcoli} =
\begin{cases}
1 &~\text{if } d_{k}+w_{\text{max},k-1} < r_{\text{comb}},\\
0 & ~\text{if } d_{k}- w_{\text{max},k-1} \geq r_{\text{comb}},\\
\eqref{eq:pcol} &~\text{otherwise}.
\end{cases}\IEEEeqnarraynumspace
\end{IEEEeqnarray}

For reasons of readability, the dependence on $k$ for $\bm{p}_\text{int}$, $\theta_{\text{int},1}$, $r_\text{int}$, and $a$ is omitted.

\bibliography{references/Dissertation_bib}
\bibliographystyle{unsrt}

\end{document}